%% file: ms.tex
\pdfoutput=1
\documentclass[a4paper,UKenglish]{lipics-v2019}
\input{headers_and_macros}
\input{title_and_authors}
\begin{document}
\nolinenumbers
\sloppy
\maketitle
\input{abstract}
\input{introduction}
\input{results_overview}
\input{prelims}
\input{arity_2_poly}
\input{general_arity}
\input{algorithmic_version}
\input{ack}
\bibliographystyle{plainurl}
\bibliography{atri,shi}
\onecolumn
\appendix
\input{appendix_results}
\input{appendix_arity2}
\input{appendix_algorithms}
\input{appendix_garity}
\end{document}

%% file: headers_and_macros.tex
\usepackage{mathtools}
\usepackage{amsmath}
\usepackage{amssymb}
\usepackage{mathtools}
\usepackage{algpseudocode}
\usepackage{latexsym}
\usepackage{lipsum}
\usepackage{wrapfig}

\usepackage{hyperref} % For hyperlinks in the PDF

\usepackage[table]{xcolor}
\usepackage{color}
\usepackage{algorithm}
\usepackage{verbatim}
\usepackage{bm}
\usepackage[utf8]{inputenc}
\usepackage{array}
\usepackage{multirow}
\usepackage{hhline}
\usepackage{mdframed}
\usepackage{graphicx}
\usepackage{enumitem}
\usepackage{xspace}
\usepackage{tikz}
\usepackage{capt-of}
\usepackage{subcaption}

\tikzset{my loop/.style =  {to path={
			\pgfextra{}
			[looseness=12,min distance=10mm]
			\tikz@to@curve@path},font=\sffamily\small
}}

\newcommand{\poly}{{\mathrm{poly}}}
\newcommand{\Time}{{\mathrm{t}}}
\let\oldReturn\Return
\renewcommand{\Return}{\State\oldReturn}

\newcommand{\myvec}[1]{\mathbf{#1}}

\DeclarePairedDelimiter\ceil{\lceil}{\rceil}
\DeclarePairedDelimiter\floor{\lfloor}{\rfloor}

\DeclareMathOperator{\E}{\mathbb E}
\DeclareMathOperator{\union}{\bigcup}

\renewcommand{\tilde}{\widetilde}

\newcommand{\q}{\mathrm{Q}}

\newcommand{\calA}{{\mathcal A}_{\mathrm Q}}
\newcommand{\calB}{{\mathcal B}}
\newcommand{\calC}{{\mathcal C_{\mathrm Q, \Delta}}}
\newcommand{\calD}{{\mathcal D_{\mathrm Q, \Delta}}}

\newcommand{\calS}{{\mathcal T}_{\mathrm Q, \Delta}}

\newcommand{\calSS}{{\mathcal T}_{\mathrm{Q}_{0}, 2}}

\newcommand{\calDD}{{\mathcal D_{\mathrm{Q}_{0}, 2}}}

\newcommand{\calU}{{\mathcal U}}

\newcommand{\set}[1]{\left\{#1\right\}}

\newcommand{\eps}{\epsilon}

\graphicspath{{figs/}}

\newcommand\numberthis{\addtocounter{equation}{1}\tag{\theequation}}

\newcommand{\pack}{\mathsf{PkNum}}
\newcommand{\cover}{\mathsf{CvNum}}

\newcommand{\joinUB}{\mathsf{PrjBnd}}

\newcommand{\Dom}{\mathrm{Dom}}
\newcommand{\Dist}{\mathrm{Dist}}

\newtheorem{problem}[theorem]{Problem}

\newcommand{\DC}{\mathrm{DC}}
\newcommand{\J}{\mathrm{J}_{Q}}
\newcommand{\JJ}{\mathrm{J}_{Q_0}}

\newcommand{\BCQ}{\mathrm{BCQ}}
\newcommand{\DBP}{\mathrm{DBP}}
\newcommand{\AGM}{\mathrm{AGM}}
\newcommand{\cc}{\mathrm{cc}}
\newcommand{\lastpick}{\mathrm{lastpick}}
\newcommand{\topick}{\mathrm{topick}}
\newcommand{\PMB}{\mathrm{PMB}}
\newcommand{\ZPP}{\mathrm{ZPP}}
\newcommand{\NP}{\mathrm{NP}}

\newcommand{\ub}{\mathrm{ub}}
\newcommand{\lb}{\mathrm{lb}}
\newcommand{\LP}{\mathrm{LP}}
\newcommand{\gap}{\mathrm{Gap}}
\newcommand{\s}{n - \Delta + 1}
\newcommand{\RS}{\mathrm{RS}}
\newcommand{\CRT}{\mathrm{CRT}}
\newcommand{\wt}{\mathrm{wt}}

\newcommand{\ED}{\mathrm{ED}}

\DeclareMathOperator*{\argmax}{arg\,max}
\DeclareMathOperator*{\argmin}{arg\,min}

\tikzset{my loop/.style =  {to path={
			\pgfextra{}
			[looseness=3,min distance=3mm]
			\tikz@to@curve@path},font=\sffamily\small
}}  

%% file: title_and_authors.tex
\title{Covering the Relational Join}

\author{Shi Li}{University at Buffalo}{}{}{}
\author{Sai Vikneshwar Mani Jayaraman}{University at Buffalo}{}{}{}

\author{Atri Rudra}{University at Buffalo}{}{}{}

\authorrunning{Shi Li, Sai Vikneshwar Mani Jayaraman and Atri Rudra}

%% file: abstract.tex
\begin{abstract}
In this paper, we initiate a theoretical study of what we call the {\em join covering problem}. We are given a natural join query instance $Q$ on $n$ attributes and $m$ relations $(R_i)_{i \in [m]}$. Let $J_{Q} = \ \Join_{i=1}^m R_i$ denote the join output of $Q$. In addition to $Q$, we are given a parameter $\Delta: 1\le \Delta\le n$ and our goal is to compute the smallest subset $\mathcal{T}_{Q, \Delta} \subseteq J_{Q}$ such that every tuple in $J_{Q}$ is within Hamming distance $\Delta - 1$ from some tuple in $\mathcal{T}_{Q, \Delta}$. The {\em join covering problem} is a fairly general problem that captures two well-studied problems as special cases --  $(a)$ computing the natural join from database theory and $(b)$ constructing a {\em covering code} with {\em covering radius} $\Delta - 1$ from coding theory. 

We start with the combinatorial version of the {\em join covering problem}, where our goal is to determine the {\em worst-case} $|\mathcal{T}_{Q, \Delta}|$ in terms of the structure of $Q$ and value of $\Delta$. One obvious approach to upper bound $|\mathcal{T}_{Q, \Delta}|$ is to exploit a distance property (of Hamming distance) from coding theory and combine it with the worst-case bounds on output size of natural joins ($\AGM$ bound hereon) due to Atserias, Grohe and Marx [SIAM J. of Computing'13]. Somewhat surprisingly, this approach is {\em not} tight even for the case when the input relations have arity at most two. Instead, we show that using the polymatroid degree-based bound of Abo Khamis, Ngo and Suciu [PODS'17] in place of the $\AGM$ bound gives us a tight bound (up to constant factors) on the $|\mathcal{T}_{Q, \Delta}|$ for the arity two case. We prove lower bounds for $|\mathcal{T}_{Q, \Delta}|$ using a well-known class of error-correcting codes called the {\em Reed-Solomon codes} and their number-theoretic variants called the {\em Chinese Remainder Theorem codes}. We can extend our results for the arity two case to general arity with a polynomial gap between our upper and lower bounds. 

Finally, we translate our combinatorial results to algorithmic ones for computing an approximation to $\mathcal{T}_{Q, \Delta}$ with a simple search-to-decision reduction. Our algorithms have runtimes that are no more than polynomially worse compared to the optimal ones.
\end{abstract}

%\keywords{Degree-based Bounds for Joins; Hamming distance; Natural Join; Covering Codes}

%% file: introduction.tex
\section{Introduction} \label{sec:intro}
In this paper, we initiate a theoretical study of what we call the {\em join covering problem}. We are given a (multi-)hypergraph $G = (V, E)$\footnote{For simplicity of notation, we will assume that the inputs are hypergraphs instead of (multi-)hypergraphs. The transformation is straightforward, where we collapse all hyperedges on the same subset of vertices into a single one.}, where each vertex $v \in V$ is an attribute with domain $\Dom(v)$. For each hyperedge $e \in E$, we are given a relation $R_{e} \subseteq \prod_{v \in e} \Dom(v)$. We refer to $G$ as a query (multi-)hypergraph and define a join query instance as $\q= (G, (R_e)_{e \in E})$. Let $\J = \ \Join_{e \in E} R_e$ denote the join output of $\q$, where $\J$ is a relation with attributes $V$ and for every tuple $\mathbf{t} \in \J$, we have $\mathbf{t}_{e} \in R_{e}$ for every $e \in E$. Here, $\mathbf{t}_{e}$ denotes the projection of $\mathbf{t}$ onto attributes in $e \subseteq V$. Note that $\J \subseteq \prod_{v \in V} \Dom(v)$ and we define $|V| = n$. In addition to $\q$, we are given an additional input parameter $\Delta: 1\le \Delta\le n$.

We are now ready to formally define the {\em join covering problem}. We start with the notion of a {\em join cover} given $\q$ and $\Delta$.
\begin{definition}[Join Cover] \label{defn:cover}
Given $\q = (G, \{R_e: e \in E\})$ and $\Delta$, a join cover is a subset $\calU \subseteq \J$ such that for every tuple $\mathbf{t} \in \J$, there exists a tuple $\mathbf{t'} \in \calU$ with $\Dist(\mathbf{t},\mathbf{t'}) < \Delta$, where $\Dist$ is some distance metric defined on tuples in $\J$.
\end{definition}
Note that $\calU = \J$ is a valid join cover. We define the join covering problem now.
\begin{problem}[Join Covering problem] \label{prob:covering}
Given $\q = (G, \{R_e: e \in E\})$ and $\Delta$, the goal is to output a join cover $\calS$ such that
\[\calS = \argmin_{\calU : \calU \text{ is a join cover of } \J} |\calU| .\]
In words, $\calS$ is a join cover with the smallest size.  
\end{problem}
Throughout this paper, we assume that $\Dist$ is Hamming distance, which we define as follows. For any pair of distinct tuples $\mathbf{t}, \mathbf{t'}$ in $\J$, we have
\begin{equation} \label{eq:hd}
\Dist(\mathbf{t}, \mathbf{t'}) = |\{ v \in V : \mathbf{t}_{v} \neq \mathbf{t'}_{v}\}|,
\end{equation} 
where $\mathbf{t}_{v}$ (or $\mathbf{t'}_{v}$) denotes $\mathbf{t}$ projected on to attribute $v$. We now illustrate the computation of $\calS$ with an example.
\begin{example} \label{ex:intro}
We are given a join query instance $\q_{0}$ where $G_{0}$ is a $4$-cycle (see Figure~\ref{fig:example}) and the input relations $R_{(1, 2)}, R_{(2, 3)}, R_{(3, 4)}, R_{(4, 1)}$ are given in Table~\ref{table:instance}. Note that $n = 4$ and the attributes could be interpreted as $1$ (`Conference') with $\Dom(1) = \{ICDT\}$, $2$ (`Year') with $\Dom(2) = \{2017, 2018, 2019, 2020\}$, $3$ (`Continent') with $\Dom(3) = \{Europe\}$ and $4$ (`Country') with $\Dom(4) = \{Austria, Denmark, Italy, Portugal\}$. In addition, we are given $\Delta = 2$.
%\begin{wrapfigure}[!ht]{o}{0.5\textwidth}
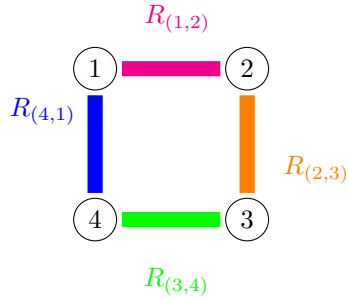
\begin{figure}[!htbp]
\centering
	\begin{tikzpicture}[myn/.style={circle, draw,inner sep=0.1cm,outer sep=2pt}]
	\node[myn] (A) at (0,0) {$1$};
	\node[myn] (B) at (2, 0) {$2$};
	\node[myn] (C) at (2, -2) {$3$};
	\node[myn] (D) at (0, -2) {$4$};
	\draw[line width = 2mm, magenta] (A) -- (B) node [magenta, above left = 7.5pt, sloped = 90] (TextNode) {$R_{(1, 2)}$};
	\draw[line width = 2mm, orange] (B) -- (C) node [orange, above right = 7.5pt, sloped = 90] (TextNode) {$R_{(2, 3)}$};
	\draw[line width = 2mm, green] (C) -- (D) node [green, below right = 12pt, sloped = 90] (TextNode) {$R_{(3, 4)}$};
	\draw[line width = 2mm, blue] (D) -- (A) node [midway, blue, above left = 1pt, sloped = 90] (TextNode) {$R_{(4, 1)}$};
	\end{tikzpicture}
	\caption{$G_{0}$ is a $4$-cycle i.e., a cycle with $4$ vertices. For each edge in $G_{0}$, there is an input relation indexed by it -- $R_{(1, 2)}, R_{(2, 3)}, R_{(3, 4)}$ and $R_{(4, 1)}$.} \label{fig:example}
\end{figure}
%\end{wrapfigure}

\begin{table*}[!htbp]
	\begin{tabular}{llll}
			\begin{tabular}{| c | c |}
				\hline
				\rowcolor{magenta!40} $1$ & $2$ \\
				\hline
				\rowcolor{magenta!40} ICDT & $2017$ \\
				\hline
				\rowcolor{magenta!40} ICDT & $2018$ \\
				\hline
				\rowcolor{magenta!40} ICDT & $2019$ \\
				\hline
				\rowcolor{magenta!40} ICDT & $2020$ \\
				\hline
			\end{tabular}
			\begin{tabular}{| c | c |}
				\hline
				\rowcolor{orange} $2$ & $3$ \\
				\hline
				\rowcolor{orange} $2017$ & Europe \\
				\hline
				\rowcolor{orange} $2018$ & Europe  \\
				\hline
				\rowcolor{orange} $2019$ & Europe \\
				\hline
				\rowcolor{orange} $2020$ & Europe\\
				\hline
			\end{tabular}
			\begin{tabular}{| c | c |}
				\hline
				\rowcolor{green} $3$ & $4$ \\
				\hline
				\rowcolor{green} Europe & Austria \\
				\hline
				\rowcolor{green} Europe & Denmark \\
				\hline
				\rowcolor{green} Europe & Italy \\
				\hline
				\rowcolor{green} Europe &  Portugal \\
				\hline
			\end{tabular}
			\begin{tabular}{| c | c |}
				\hline
				\rowcolor{blue!40} $4$ & $1$ \\
				\hline
				\rowcolor{blue!40} Austria & ICDT\\
				\hline
				\rowcolor{blue!40} Denmark & ICDT\\
				\hline
				\rowcolor{blue!40} Italy & ICDT\\
				\hline
				\rowcolor{blue!40} Portugal &  ICDT\\
				\hline
			\end{tabular}  
		\end{tabular}
		\caption{We are given four relations in order (from left to right): $R_{(1, 2)} = \Dom(1) \times \Dom(2), R_{(2, 3)} = \Dom(2) \times \Dom(3), R_{(3, 4)} = \Dom(3) \times \Dom(4), R_{(4, 1)} = \Dom(4) \times \Dom(1)$. Each relation has $4$ entries.} \label{table:instance}
\end{table*}

Given $\q_{0}$ and $\Delta = 2$, recall that our goal is to compute $\calSS$. We present the join output $\JJ = R_{(1, 2)} \Join R_{(2, 3)} \Join R_{(3, 4)} \Join R_{(4, 1)}$ and a $\calSS$ in Table~\ref{table:cover-join}. We still need to argue that $\calSS$ is a valid solution to the join covering problem. We start by noting that $\calSS \subseteq \JJ$ and for every tuple $\mathbf{t} \in \JJ$, there exists a tuple $\mathbf{t'} \in \calSS$ such that $\Dist(\mathbf{t}, \mathbf{t'}) < 2$. Note that if $|\calSS| < 4$, it will miss one of the `Years' and one of the `Countries'. In particular, this implies there exists a tuple in $\mathbf{t} \in \JJ$, such that for every $\mathbf{t'} \in \calSS$, $\Dist(\mathbf{t}, \mathbf{t'}) \ge 2$ and hence, $\calSS$ cannot be a join cover.
%To complete the argument, we still need to show that there cannot exist a $\calSS$ such that $|\calSS| < 4$, which we prove by contradiction. Let $\calSS'$ be a solution to the join covering problem such that $|\calSS'| = 3$. Since $\calSS' \subseteq J_{0}$, note that there should exist at least one tuple $\mathbf{t} \in J$ such that for every tuple $\mathbf{t'} \in \calSS'$, $\Dist(\mathbf{t}, \mathbf{t'}) \ge 2$. In particular, this contradicts our earlier assumption that $\calSS'$ is a valid join cover.
\begin{table*}[!hbt]
		\centering
				\begin{tabular}{ll}
			\begin{tabular}{| c | c | c | c |}
			\hline
			\rowcolor{violet!40} $1$ & $2$ & $3$ & $4$ \\
			\hline
			\rowcolor{violet!40}  ICDT & $2017$ & Europe & Austria \\
			\hline
			\rowcolor{violet!40}  ICDT & $2017$ & Europe & Denmark \\
			\hline
			\rowcolor{green!40}  ICDT & $2017$ & Europe & Italy \\
			\hline
			\rowcolor{violet!40}  ICDT & $2017$ & Europe & Portugal \\
			\hline
			\rowcolor{green!40}  ICDT & $2018$ & Europe & Austria \\
			\hline
			\rowcolor{violet!40}  ICDT & $2018$ & Europe & Denmark \\
			\hline
			\rowcolor{violet!40}  ICDT & $2018$ & Europe & Italy \\
			\hline
			\rowcolor{violet!40}  ICDT & $2018$ & Europe & Portugal \\
			\hline
			\rowcolor{violet!40}  ICDT & $2019$ & Europe & Austria \\
			\hline
			\rowcolor{violet!40}  ICDT & $2019$ & Europe & Denmark \\
			\hline
			\rowcolor{violet!40}  ICDT & $2019$ & Europe & Italy \\
			\hline
			\rowcolor{green!40}  ICDT & $2019$ & Europe & Portugal \\
			\hline
			\rowcolor{violet!40}  ICDT & $2020$ & Europe & Austria \\
			\hline
			\rowcolor{green!40}  ICDT & $2020$ & Europe & Denmark \\
			\hline
			\rowcolor{violet!40}  ICDT & $2020$ & Europe & Italy \\
			\hline
			\rowcolor{violet!40}  ICDT & $2020$ & Europe & Portugal \\
			\hline
		\end{tabular}
	
			\begin{tabular}{| c | c | c | c |}
				\rowcolor{green!40} $1$ & $2$ & $3$ & $4$ \\
				\hline
				\rowcolor{green!40}  ICDT & $2017$ & Europe & Italy \\
				\hline
				\rowcolor{green!40}  ICDT & $2018$ & Europe & Austria \\
				\hline
				\rowcolor{green!40}  ICDT & $2019$ & Europe & Portugal \\
				\hline
				\rowcolor{green!40}  ICDT & $2020$ & Europe & Denmark \\
				\hline
			\end{tabular}
			
		\end{tabular}
		\caption{The first table is join output $\JJ = R_{(1, 2)} \Join R_{(2, 3)} \Join R_{(3, 4)} \Join R_{(4, 1)}$ and the second table is a solution $\calSS$ to Problem~\ref{prob:covering}.} \label{table:cover-join}
	\end{table*}
\end{example}
When $\Delta = 1$, we have $\calS = \J$ and the join covering problem reduces to the problem of computing the natural join. In Example~\ref{ex:intro}, if we set $\Delta = 1$, then $\mathcal{T}_{\q_{0}, 1} = \JJ$ (see $\JJ$ from Table~\ref{table:cover-join}) is a valid solution to the join covering problem. On the other hand, when $\J = \ \Join_{e \in E} R_e = \prod_{v \in V} \Dom(v)$, the join covering problem corresponds to the problem of explicitly constructing what is known as a {\em covering code} with {\em covering radius} $\Delta - 1$~\cite{covering}. In Example~\ref{ex:intro}, we have $\JJ = R_{(1, 2)} \Join R_{(2, 3)} \Join R_{(3, 4)} \Join R_{(4, 1)} = \Dom(1) \times \Dom(2) \times \Dom(3) \times \Dom(4)$ and $\calSS$ is in fact a covering code with covering radius $1$ (see $\JJ$ and $\calSS$ from Table~\ref{table:cover-join}). Thus, the join covering problem is a natural problem and generalizes both the natural join and covering codes, which have been studied extensively in database theory and coding theory for over five decades~\cite{db-book,covering}.

In order to solve the join covering problem, we begin by solving its combinatorial version, which we define here. Given $G$, $N$ and $\Delta$, we would like to compute
\begin{equation} \label{eq:cover-comb}
\cover(G, N, \Delta) = \max_{R_e: |R_e| \le N, e \in E} |\calS|,
\end{equation}
where $\q = (G, \{R_e: e \in E\})$. In words, for a fixed $G$, $N$ and $\Delta$, we are interested in the {\em size of a worst-case} $\calS$. Most of the technical work in this paper is dedicated to proving upper and lower bounds under size constraints on $R_e$s on $\cover(G, N, \Delta)$ (we discuss known limitations in exactly determining $\cover(G, N, \Delta)$ at the end of Section~\ref{sec:contri}). We use the following related combinatorial problem to prove an upper bound on $\cover(G, N, \Delta)$.
\begin{problem} \label{prob:free-vars}
Given $G$, $N$ and $s : 1 \le s \le n$, the goal is to determine
\begin{equation} \label{eq:free-vars}
\joinUB(G, N, s) = \min_{S \subseteq V: |S| = s} \max_{R_e: e \in E, |R_e| \le N} \left | \Join_{e \in E} \pi_{e \cap S}(R_e) \right|
\end{equation}
Here, $\pi_{e \cap S}(R_e)$ denotes the projection of the relation $R_e$ on to attributes in $e \cap S$.
\end{problem}
In particular, we show that for any hypergraph $G$, $N$ and $\Delta$ the following is true (see Theorem~\ref{thm:proj-cover}):
\[\cover(G, N, \Delta) \le \joinUB(G, N, n - \Delta + 1).\]
Throughout the paper, we assume $s = n - \Delta + 1$ unless explicitly stated otherwise. Note that the corresponding algorithmic problem for $\joinUB(G, N, s)$ would involve computing:
\[\argmin_{S \subseteq V: |S| = s}  \left | \Join_{e \in E} \pi_{e \cap S}(R_e) \right|\]
for a given instance $\q = (G, \{R_e: e \in E \})$. Continuing Example~\ref{ex:intro} with $\q_{0}$ and $\Delta = 2$, we need to consider four possibilities for $S$ (since $s = 3$). We show one such query when $S = \{ 1, 2, 3\}$ in Table~\ref{table:ub} below.
\begin{table*}[!hbt]
	\centering
	\begin{tabular}{llll}
		\begin{tabular}{| c | c |}
			\hline
			\rowcolor{magenta!40} $1$ & $2$ \\
			\hline
			\rowcolor{magenta!40} ICDT & $2017$ \\
			\hline
			\rowcolor{magenta!40} ICDT & $2018$ \\
			\hline
			\rowcolor{magenta!40} ICDT & $2019$ \\
			\hline
			\rowcolor{magenta!40} ICDT & $2020$ \\
			\hline
		\end{tabular}
		\begin{tabular}{| c | c |}
			\hline
			\rowcolor{orange} $2$ & $3$ \\
			\hline
			\rowcolor{orange} $2017$ & Europe \\
			\hline
			\rowcolor{orange} $2018$ & Europe  \\
			\hline
			\rowcolor{orange} $2019$ & Europe \\
			\hline
			\rowcolor{orange} $2020$ & Europe\\
			\hline
		\end{tabular}
		\begin{tabular}{| c |}
			\hline
			\rowcolor{green} $3$\\
			\hline
			\rowcolor{green} Europe \\
			\hline
		\end{tabular}
		\begin{tabular}{| c | c | c|}
			\hline
			\rowcolor{red!40} $1$ & $2$ & $3$\\
			\hline
			\rowcolor{red!40} ICDT & $2017$ & Europe \\
			\hline
			\rowcolor{red!40} ICDT & $2018$ & Europe \\
			\hline
			\rowcolor{red!40} ICDT & $2019$ & Europe \\
			\hline
			\rowcolor{red!40} ICDT & $2020$ & Europe \\
			\hline
		\end{tabular}
	\end{tabular}
	\caption{We are given the three input relations in order (from left to right): $R_{(1, 2)}, R_{(2, 3)}, \pi_{3}(R_{(3, 4)})$, followed by the output $\pi_{S}(\JJ) = R_{(1, 2)} \Join R_{(2, 3)} \Join \pi_{3}(R_{(3, 4)})$, where $S = \{1, 2 ,3\}$.} \label{table:ub}
\end{table*}
In addition to the theoretical motivation stated earlier, it turns out that our results for the join covering problem imply bounds and algorithms for certain variants of entity resolution (ER). Further, our results for 
$\joinUB(G, N,s)$ implies bounds and algorithms for determining pivot attributes for faster query execution. We discuss these connections in detail next.

\subsection{Motivation} \label{sec:motivation}
We would like to state here that the practical application of our results to the following two problems is somewhat speculative. Our goal in presenting these applications is to showcase the generality of the join covering problem. We start with connections to specific variants of ER, followed by extending our results to edit distance (another similarity metric). Then, we state a practically relevant version of ER and conclude with connections to faster (join) query execution in distributed computing.
 
\subsubsection{Connection to Specific Variants of ER}
We start by showing how our results for the join covering problem has implications for certain variants of ER. At a high level, ER is the task of {\em covering} a dataset based on a {\em distance metric}. Afrati et al.~\cite{ASMPU12} considered a variant of the ER problem, which is precisely the join covering problem with $G = (V, E = \{R_1\})$ ($R_1$ spans all attributes in $V$) and $\Dist$ is Hamming distance. In particular, they presented algorithms to compute $\calS$ in this setting. In a later work, Altwaijry et al.~\cite{AMK15}4 considered Problem~\ref{prob:covering} in its actual form (i.e., $G$ is a general join query hypergraph) and presented heuristics for computing $\calS$ for different distance metrics. To the best of our knowledge, ours is the first work to study this more general problem from a theoretical standpoint. While Hamming distance might not be the most practically relevant metric for similarity, we believe it is a reasonable choice for a first theoretical study. Note that our results hold only for ER variants that can be invoked using the join covering problem with $\Dist$ as Hamming distance, whose implications we discuss here. The notion of a join cover corresponds to a {\em covering} subset in the ER language. In this paper, we bound the worst-case size of the smallest covering subset (which we denote by $\cover(
G, N, \Delta)$) for a given query. Since computing $\cover(G, N, \Delta)$ exactly is a known hard problem~\cite{CovLB1,CovUB1}, we settle for proving upper and lower bounds for it. For the case when $G$ is a graph, we prove matching upper and lower bounds for $\cover(G, N, \Delta)$ up to a factor of $2^{n}$. Since $n$ is typically treated as a constant in the ER setting as well~\cite{AMK15}, our results are tight (up to constant factors) for this case. For general hypergraphs $G$, we present a weaker result, where our upper and lower bounds for $\cover(G, N, \Delta)$ can differ by a polynomial in the lower bound (which in turn is a polynomial in $N$). We consider the problem of computing a covering subset (not necessarily of the smallest size). When $G$ is a graph, we present an algorithm that compute a covering subset in time cubically worse than the optimal runtime. For general hypergraphs $G$, our algorithm has a runtime polynomially worse than the optimal runtime. 

\subsubsection{Extending Our Results to Edit Distance}
Next, we show that our upper bounds on $\cover(G, N, \Delta)$ hold for Edit distance metric~\cite{edit2,HammingER} as well under some constraints. The Edit distance $\ED(\mathbf{t}, \mathbf{t'})$ between two tuples with length $m$ and $n$ is defined as the minimum number of operations to transform $\mathbf{t}$ into $\mathbf{t'}$ using a predefined set of operations:
\begin{itemize}
\item {Inserting a new symbol at any position $i$ on $\mathbf{t}$. Note that this offsets the sub-vector from $i$ to $m$ by one position to the right. The new length is $m + 1$.}
\item {Deleting an existing symbol at any position $i$ on $\mathbf{t}$. Note that this offsets the sub-vector from $i$ to $m$ by one position to the left. The new length is $m - 1$.}
\item {Substituting an existing symbol at any position $i$ by another symbol. This operation does not affect the position of other vectors.}
\end{itemize}
Note that only if substitution operation is allowed and $m = n$, we have
\[\ED(\mathbf{t}, \mathbf{t'}) = \Dist(\mathbf{t}, \mathbf{t'}),\]
where $\Dist$ is Hamming distance defined by~\eqref{eq:hd}. On the other hand, $m \le n$, we have 
\[\ED(\mathbf{t}, \mathbf{t'}) \le \Dist(\mathbf{t}_{n}, \mathbf{t'}),\]
where $\mathbf{t}_{n}$ is $\mathbf{t}$ padded with $n - m$ dummy symbols that are not present in $\mathbf{t'}$. Thus, upper bounds obtained for $\cover(G, N, \Delta)$ hold for Edit distance as well, although they could be crude when $m \ll n$ but our lower bounds do not hold.

\subsubsection{Practically Relevant ER}
Before going further, we would like to note that the ER variant we consider in this paper does not seem to be very popular in practice. We state a practically relevant toy version here~\cite{AGK06,CGK06}:
\begin{definition}
Given a relation $R_1$ on one attribute $v$ with domain $\Dom(v)$, the goal is to output the smallest $\mathcal{E} \subseteq R_1$ such that for every tuple $\mathbf{d} \in R_1$, there exists a tuple $\mathbf{d'} \in \mathcal{E}$ satisfying $\Dist'(\mathbf{d}, \mathbf{d'}) < \Delta$. Here, every value in $\Dom(v)$ is treated as a vector of a fixed length $\ell$ and the Hamming distance $\Dist'$ is defined as $\Dist'(\mathbf{d}, \mathbf{d'}) =|\{ i \in [\ell] : \mathbf{d}_{i} \neq \mathbf{d'}_{i}\}|$.
\end{definition}
In particular, Hamming distance is used in the above definition when $v$ is a numeric attribute~\cite{HammingER}. Unfortunately, our results do not say anything non-trivial for this case.

\subsubsection{Connection to faster query execution in distributed computing}
Finally, we show the connection between Problem~\ref{prob:free-vars} and determining pivot attributes for faster query execution. The queries of interest here are joins with aggregates. A popular technique employed by real-time systems for (preparation of) execution of these queries in a distributed setting is query-aware partitioning of data, where the input tables are sorted/partitioned on a particular attribute~\cite{partition1,partition2} (which we call a {\em pivot} attribute for simplicity). In particular, this is done so that the records with a specific range of attribute values are physically co-located in memory/in the same machine, which comes in handy during query execution. More recently, there has been a push to partition the input tables simultaneously on multiple pivot attributes instead of one~\cite{fast19}. Let $\mathcal{F} \subseteq V$ be the set of pivot attributes on which the data is going to be partitioned. One of the most important requirements in choosing $\mathcal{F}$ is that the output size of the original query when projected on to variables in $\mathcal{F}$ should be as small as possible. This is needed since this subquery on $\mathcal{F}$ is actually computed and stored, and then used (for pre-filtering tuples) in the computation of the original query. Having a smaller output naturally leads to lesser storage and lesser computational overhead. Our results on computing $\joinUB(G, N, |\mathcal{F}|)$ give meaningful upper bounds for this requirement that can aid the choice of $\mathcal{F}$. Once $\mathcal{F}$ is fixed, the induced join query on $\mathcal{F}$ can be computed in time proportional to $\joinUB(G, N, |\mathcal{F}|)$ using any worst-case optimal join algorithm~\cite{NPRR,NRR13,Leapfrog}. While the metrics used for choosing $\mathcal{F}$ are more nuanced than $\joinUB(G, N, |\mathcal{F}|)$~\cite{muthianp}, it is still a valid upper bound. 

\subsection{Our Contributions} \label{sec:contri}
In this paper, we make progress on solving both the combinatorial and algorithmic versions of the join covering problem for all $G$, $N$, and $\Delta$. We present our results in Table~\ref{tab:intro_results} for all $N$, $\Delta$ and $s = n - \Delta + 1$. We start with our results for bounding $\cover(G, N, \Delta)$. Recall that we upper bound $\cover(G, N, \Delta)$ by $\joinUB(G, N, s)$. For the case when $G$ is a graph, we prove matching upper and lower bounds for $\cover(G, N, \Delta)$ up to a factor of $2^{n}$. Since $n$ is typically treated as a constant in database settings, our results are tight (up to constant factors). For general hypergraphs $G$, we prove a weaker result, where our upper and lower bounds for $\cover(G, N, \Delta)$ can differ by at most
\begin{equation} \label{eq:algo-ub}
\gap(G, N, \Delta) = O\left(2^{n} \cdot \min\left(N \cdot \cover(G, N, \Delta)^{9.37}, N^{O(\log(s))} \cdot \cover(G, N, \Delta)^{2.73}\right) \right).
\end{equation}
We discuss the techniques we use to prove upper and lower bounds on $\cover(G, N, \Delta)$ in detail in Section~\ref{sec:overview}.
 
Our combinatorial results can be seamlessly converted to algorithms for computing an approximation to $\calS$ for any instance $\q = (G, \{R_e : e \in E\})$ using a search to decision reduction. For the case when $G$ is a graph, we present an algorithm that computes a join cover $\calU$ for any instance $\q$ such that $|\calU| \le \joinUB(G, N, s)$ and $\frac{\joinUB(G, N, s)}{\cover(G, N, \Delta)} \le 2^{n}$ in time $O(N \cdot \joinUB(G, N, s) \cdot \Time(\calA))$, where $\Time(\calA)$ is the runtime of an optimal Blackbox algorithm $\calA$ that checks if $\J = \ \Join_{e \in E} R_e$ is empty or not.\footnote{The problem that $\calA$ solves is called the Boolean Conjunctive query in Database literature.} For general hypergraphs $G$, we show an analogous result where a join cover $\calU$ with $|\calU| \le \joinUB(G, N, s)$ can be computed in time $O(N \cdot \joinUB(G, N, s) \cdot \Time(\calA))$ and satisfies $\frac{\joinUB(G, N, s)}{\cover(G, N, \Delta)}$ is at most $\gap(G, N, \Delta)$. Recall that $\cover(G, N, \Delta)$ is a polynomial in $N$. We would like to note here that an optimal algorithm that computes $\cover(G, N, \Delta)$ (if it exists) will take time at least $\Omega(\max(\cover(G, N, \Delta), \Time(\calA)))$ in the RAM model of computation. Otherwise, this would contradict the fact $\calA$ is an optimal algorithm. Thus, the runtime of our algorithms are cubically (in $N$) worse when $G$ is a graph and polynomially (in $N$) worse for general hypergraphs $G$, compared to the optimal ones.

\begin{table*}[hbt!]
	{\small
		\centering
		{
			\hspace{0.5cm} {
				\begin{tabular}{| c | c | c | c | c |}
					\hline
					$G$ & Combinatorial Gap & Ref & Algorithmic Gap & Ref \\
					\hline 
					\rowcolor{green} Simple & $O(1)$ & Thm~\ref{lowerboundtheorem} & $O(N \cdot \min(\cover(G, N, \Delta), \Time(\calA)))$ & Lem~\ref{lemma5.2} \\
					\hline
					\rowcolor{green!20} Hypergraphs & $\gap(G, N, \Delta)$~\eqref{eq:algo-ub} & Thm~\ref{generalarityubtheorem} & $O\left(N \cdot \min(\gap(G, N, \Delta), \Time(\calA)) \right)$ & Lem~\ref{lemma5.3} \\
					\hline
				\end{tabular}
			}
			\caption{The first column denotes the class of graphs under consideration. The second and third columns denote the combinatorial gap given by $\frac{\joinUB(G, N, s)}{\cover(G, N, \Delta)}$ and a pointer to the relevant result. The fourth and final columns denote the algorithmic gap given by the gap in our runtimes vs the optimal ones (with runtime $\Omega(\cover(G, N, \Delta) + \Time(\calA))$) and a pointer to the relevant result.} \label{tab:intro_results}
		}
	}
\end{table*}
We now discuss the gaps in our results and show how we cannot resolve some of them without solving some fundamental problems in computational complexity and coding theory. The gaps in our bounds are twofold -- first, the gaps between the upper and lower bounds we prove for $\cover(G, N, \Delta)$ and second, the problem of exactly computing $\cover(G, N, \Delta)$. We address them one-by-one here. When $G$ is a graph, our upper and lower bounds for $\cover(G, N, \Delta)$ differ by at most a factor of $2^{n}$. As stated earlier, this is reasonable in the database setting since $n$ is constant. However, it turns out that this gap cannot be eliminated unless $\ZPP=\NP$ (see~\cite{AGM}). For general hypergraphs $G$, our upper and lower bounds for $\cover(G, N, \Delta)$ differ by at most a polynomial in $N$. We complement this result by showing that this gap is inherent if we use $\joinUB(G, N, \s)$ to upper bound $\cover(G, N, \Delta)$. In particular, we show that $\frac{\joinUB(G, N, \s)}{\cover(G, N, \Delta)}$ should be at least $N^{1 + \frac{1}{\bar{e}} - \delta}$, where $\bar{e} \approx 2.72$ and $\delta$ is a constant greater than zero. In other words, we cannot hope for a multiplicative factor gap independent of $N$ with this proof technique.

Finally, we discuss the problem of exactly computing $\cover(G, N, \Delta)$. For a given $G$, $N$ and $\Delta$, when $J =\ \Join_{e \in E} R_e = \prod_{v \in V} \Dom(v)$ and all domains are the same, determining $\cover(G, N, \Delta)$ exactly is a very hard problem~\cite{CovLB1,CovUB1} even for specific values of $\Delta$ (corresponds to bounding the size of a {\em covering code} with {\em covering radius} $\Delta - 1$). Further, explicit constructions of covering codes are known for a very small spectrum of covering radius values~\cite{CovUB1}. This is why we settle for proving upper and lower bounds on $\cover(G, N, \Delta)$ and we discuss the techniques used to prove them in Section~\ref{sec:overview}.

\subsection{Paper Organization} 
We begin with a detailed overview of our techniques in Section~\ref{sec:overview} followed by preliminaries in Section~\ref{section2}. We present our general upper bound in Section~\ref{sectionub}. Our results for arity two are in Section~\ref{section3} and general arity are in Section~\ref{sec:gen-arity}. We conclude with our algorithms in Section~\ref{sec:algos}. Due to space constraints, most of the proofs are deferred to the Appendix.

%% file: results_overview.tex
\section{Overview of Our Techniques} \label{sec:overview}
In this section, we present a detailed overview of our techniques to prove our upper and lower bounds for $\cover(G, N, \Delta)$ here.  
\subsection{Upper Bounds} \label{sec:results_ub}
We start by proving the following theorem.
\begin{theorem} \label{thm:proj-cover}
For every $G$, $N$ and $\Delta$ we have
\begin{equation} \label{eq:cover-pack-ub}
\cover(G,N,\Delta) \le \joinUB(G, N, \s).
\end{equation}
\end{theorem}
We present only an overview here and the formal proof is in Appendix~\ref{sec:proj-cover}. To prove this theorem, we make use of a related object, which we define as follows. 

For any instance $\q$, $N$ and $\Delta$, let $\calC \subseteq \J =\ \Join_{e \in E} R_e$ be the largest subset such that for any pair of distinct tuples $(\mathbf{t}, \mathbf{t'}) \in \calC$, we have $\Dist(\mathbf{t}, \mathbf{t'}) \ge \Delta$. Recall that $\Dist$ is defined by~\eqref{eq:hd}. 
%Reconsidering $\q_{0}$ from Example~\ref{ex:intro}, we can set $\calCC = \calSS$ (where $\calSS$ is from Table~\ref{table:cover-join}). We cannot have a $\calCC$ with $|\calCC| > 4$, where each pair of tuples differs in at least $2$ positions. 

We now claim that $\calC$ is a join cover. This follows from the fact that for any tuple $\mathbf{t} \in \J$, there exists a tuple $\mathbf{c} \in \calC$ such that $\Dist(\mathbf{t}, \mathbf{c}) < \Delta$. In particular, this implies 
\[\cover(G, N, \Delta) = \max_{R_e: |R_e| \le N, e \in E} |\calS| \le \max_{R_e: |R_e| \le N, e \in E} |\calC|.\]
To complete the proof, we will argue 
\[\max_{R_e: |R_e| \le N, e \in E} |\calC| \le \joinUB(G, N, \s).\]
We prove this using a property of Hamming distance, which we describe here. For a fixed subset $S$ of attributes with $|S| = s = \s$, if we project every tuple in $\calC$ on to attributes in $S$, the tuples are still pairwise distinct. In particular, this follows from the fact any pair of tuples in $\calC$ differ in at least $\Delta$ positions. 
%One such example is given by $\JJ$ in Figure~\ref{table:ub} for $\q_0$ with $\Delta = 2$ and $S = \{1, 2, 3\}$. 
Since our argument is independent of the choice of $S$, we have
\begin{align*}
\max_{R_e: e \in E, |R_e| \le N} \min_{S \subseteq V: |S| = s} |\pi_{S}(\calC)| & \le \min_{S \subseteq V: |S| = s} \max_{R_e: e \in E, |R_e| \le N} \left|\Join_{e \in E} \pi_{e \cap S} (R_e) \right| \\
& = \joinUB(G, N, s),
\end{align*}  
where the inequality follows from the fact that $|\pi_{S}(\calC)| \subseteq |J_{S}|$ (since $\calC \subseteq \J)$ for any $S \subseteq V, |S| = s$.

In summary, proving an upper bound on $\cover(G, N, \Delta)$ reduces to the problem of finding the output size of the induced join query on attributes in $S$ and then taking the minimum among all $S \subseteq V, |S| = s$. The induced join query is defined as $\q_{S} = (G_{S} = (S, \{e \cap S: e \in E\}), \{\pi_{e \cap S}(R_e): e \in E \} )$. One obvious upper bound on each such $\q_{S}$ is the $\AGM$ bound~\cite{AGM} and then to compute $\joinUB(G, N, s)$, we can take the minimum $\left|\Join_{e \in E} \pi_{e \cap S} (R_e) \right|$ among all $S \subseteq V$ with $|S| = s$. Somewhat surprisingly, this turns out to be {\em not} tight. The smallest such example that we are aware of is the $4$-cycle $G_0$ from Example~\ref{ex:intro} with $\Delta=2$ and $s = 3$. We present an example of a join query induced on $S_0 = \{1, 2, 3\}$ in Figure~\ref{fig:g_s}, where $G_{S_0} = (\{1, 2, 3\}, \{ (1, 2), (1, 3), (3)\}$. It turns out that the $\AGM$ bound on $G_{S_0}$ is $|R_{(1, 2)}| \cdot |R_{(1, 3)}| \le N^{2}$, where the inequality follows from the fact that $|R_{(1, 2)}|, |R_{(1, 3)}| \le N$. The same bound $O(N^2)$ is true for all other subsets $S, |S| = s$ as well i.e., $\joinUB(G, N, s) \le O(N^2)$.
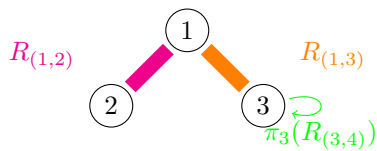
\begin{figure}[!htbp]
	\centering
	\begin{tikzpicture}[myn/.style={circle, draw,inner sep=0.1cm,outer sep=2pt}]
	\node[myn] (A) at (0,0) {$1$};
	\node[myn] (B) at (-1, -1) {$2$};
	\node[myn] (C) at (1, -1) {$3$};
	\draw[line width = 2mm, magenta] (A) -- (B) node [magenta, above left = 7.5pt, sloped = 90] (TextNode) {$R_{(1, 2)}$};
	\draw[line width = 2mm, orange] (A) -- (C) node [orange, above right = 7.5pt, sloped = 90] (TextNode) {$R_{(1, 3)}$};
    \path  (C)   edge[loop right, green] node[below = 2pt]  {$\pi_{3}(R_{(3, 4)})$} (C);
	\end{tikzpicture}
	\caption{$G_{S_0}$ is a star with $3$ vertices and relations $R_{(1, 2)}, R_{(2, 3)}$ and $\pi_{3}(R_{(3, 4)})$. Each relation has size at most $N$.} \label{fig:g_s}
\end{figure}
However, it turns out that we can prove $\joinUB(G, N, s) = O(N^{\frac{3}{2}})$, which we obtain by using the polymatroid bound ($\PMB$ hereon) defined in~\cite{panda,ngo-survey}. For general hypergraphs $G$, we use the AGM bound on $G_{S}$ since most of our arguments for the simple graphs case exploit very specific structural properties of these graphs, which we do not (yet) know to adapt to general hypergraphs. Before we move on to our lower bounds for $\cover(G, N, \Delta)$, we would like to note that when $G$ is a graph, using the degree-based bound of Joglekar and R\'e~\cite{JR-paper} is sufficient to prove our upper bounds. However, we state our results with $\PMB$ since our arguments are cleaner in the $\PMB$ language.

\subsection{Lower Bounds} \label{sec:lb}
For any $G$, $N$ and $\Delta$, if we explicitly construct an instance $\{R_e : |R_e| \le N, e \in E\}$ with a smallest join cover $\calD \subseteq \J =\ \Join_{e \in E} R_e$, then $\cover(G, N, \Delta)$ is lower bounded by $|\calD|$ (by definition). Reconsidering Example~\ref{ex:intro} with $\q_0$, $N = 4$ and $\Delta = 2$, the $R_{e}$'s in Table~\ref{table:instance} with a smallest join cover $\calDD= \calSS \subseteq \JJ$ in Table~\ref{table:cover-join} is one such valid instance. It follows that $\cover(G_{0}, 4, 2) \ge 4$. 

There are two challenges in constructing these instances -- $(1)$ arguing that $\calD \subseteq \J$ is the smallest join cover and $(2)$ showing that $|R_e| \le N$ for every $e \in E$. For $(1)$, we construct a $\J$, where each pair of tuples $\mathbf{t}, \mathbf{t'} \in \J$ satisfies $\Dist(\mathbf{t}, \mathbf{t'} ) \ge \Delta$, implying that $\J$ itself is the smallest join cover. At a high level, the $\J$ we construct is an {\em error-correcting code}, which we define as follows.
\begin{definition} {$(n, M, \Delta)_{q}$-Code}
An  $(n, M, \Delta)$-error correcting code $\mathcal{C}$ satisfies the following conditions. The code $\mathcal{C}$ is a subset of $ \prod_{v \in V} \Dom(v)$ with $|\mathcal{C}| = M$. Each tuple $\mathbf{c} \in \mathcal{C}$ has length $n$ and each pair of tuples $\mathbf{c}, \mathbf{c'} \in \mathcal{C}$ satisfies $\Dist(\mathbf{c}, \mathbf{c'}) \ge \Delta$. $\Delta$ is generally referred to as the minimum distance of $\mathcal{C}$.
\end{definition}

For this construction, we use known ideas on classes of error-correcting codes called Reed-Solomon Codes~\cite{RS} and their number-theoretic variant, Chinese Remainder Theorem codes~\cite{crt}, which we define one-by-one. A Reed Solomon (RS) code $\mathcal{C}$ is of the form $\left(n, M = q^{\s}, \Delta \right)_{q}$ on $G$ (where $q$ is a prime power), assuming all $\Dom(v)$s are the same with $|\Dom(v)| = q$ and $q \ge n$. For most of our lower bounds, we use this code. In particular, we set $q \le \floor{\sqrt{N}}$ \footnote{$n$ is typically treated is as a constant in database settings and $n \ll N$.} and we have $\calD = \mathcal{C}$. For every $e \in E$, we set $R_e = \pi_{e}(\calD)$. Since $q = O(\sqrt{N})$, we have that $|\pi_{e = (u, v)}(\calD)| \le |\pi_{u}(\calD)| \cdot |\pi_{v}(\calD)| \le |\Dom(u)| \cdot |\Dom(v)| = q^{2} \le N$. We are now ready to define RS codes.
\begin{definition}[RS Code] \label{def:rs}
Let $\mathbb{F}_q$ be a finite field\footnote{$\mathbb{F}_{q}$ has $q$ elements: addition, subtraction, multiplication and division (except zero) are all modulo $q$. Note that $q$ is always assumed to be a prime power.}. Let $\alpha_1, \alpha_2, \dots, \alpha_n$ be distinct elements (also called \textit{evaluation points}) from $\mathbb{F}_q$. For $n$ and $\Delta$, we have $\s \le n \le q$. We define an encoding function for the RS code as $E(\mathbf{m}): \mathbb{F}_q^{\s} \rightarrow \mathbb{F}_q^n$ as follows. A message $\mathbf{m} = (m_0, m_1, \dots, m_{n - \Delta})$ with $m_i \in \mathbb{F}_q$ is mapped to a degree $n - \Delta$ polynomial $\mathbf{m} \mapsto f_{m}(X)$, where $f_{m}(X) = \sum_{i = 0}^{n - \Delta}m_{i}X^i$. Note that $f_m(X) \in \mathbb{F}_q[X]$ is a polynomial of degree at most $n - \Delta$. The encoding of $\mathbf{m}$ is then the evaluation of $f_{m}(X)$ at all the $\alpha_i$'s: $\RS(\mathbf{m}) = (f_m(\alpha_1), f_m(\alpha_2), \dots, f_m(\alpha_n))$. The code $\mathcal{C} = \{\RS(\mathbf{m}): \mathbf{m} \in \mathbb{F}_q^{\s} \}$ is constructed on a finite field over $\mathbb{F}_{q}$.
\end{definition}
Note that RS code-based instances work only when $\Dom(v)$ is the same of each $v \in V$. For our arguments, we need instances where $\Dom(v)$ is different for each $v \in V$ and we use Chinese Remainder Theorem (CRT) codes for this purpose (which are number-theoretic counterparts of RS codes). We define a $(n, M, \s)_{(q_1, \dots, q_n)}$-CRT code below.
\begin{definition} [CRT Code] \label{defn:crt}
Let $1 \le \s \le n$ be integers and $q_1 < q_2 < \dots < q_n$ be $n$ distinct primes and $M = \prod_{i = 1}^{\s}q_i$. $\mathbb{Z}_q$ stands for the integers modulo $q$, i.e., the set $\{0, 1, \dots, q - 1\}$. We define an encoding function for the CRT code as $\CRT: \mathbb{Z}_{M} \rightarrow \mathbb{Z}_{q_1} \times \mathbb{Z}_{q_2} \times \dots \times \mathbb{Z}_{q_n}$ as follows. A message $\mathbf{m} = (m_0, m_1, \dots, m_{n - \Delta})$ with $m_i \in \mathbb{Z}_{q_{i + 1}}$ is encoded as $\CRT(\mathbf{m}) = (\mathbf{m} \bmod p_1, \mathbf{m} \bmod q_2, \ldots ,\mathbf{m} \bmod q_n)$. The code $\mathcal{C} = \{\CRT(\mathbf{m}): \mathbf{m} \in \mathbb{F}_q^{\s} \}$ is constructed on a finite field over $\mathbb{F}_{q}$.
\end{definition}
Note that for any RS and CRT code $\mathcal{C}$ (as defined above), each pair of tuples in $\mathcal{C}$ differ in at least $\Delta$ positions.

However, arguing $(2)$ requires bit more care for which we use ideas from lower bounds proved for the $\AGM$ bound~\cite{AGM}. In particular,~\cite{AGM} converts an optimal dual solution to a linear program to a join query instance. The constraints on the linear program ensure that $|R_e| \le N$ for every $e \in E$. In this paper, we use both feasible and optimal solutions to similar linear programs to construct our instances. Finally, we would like to mention that the $\J$'s obtained from our instances are mostly variants of Cartesian products or trivial joins, which is the case with instances in~\cite{AGM} as well.

%% file: prelims.tex
\section{Notation, Existing Results and Upper Bound for any Hypergraph $G$} \label{section2} \label{sec:query-prelims}
For simplicity of notation, we denote a multiplicative gap of $2^{O(n)}$ in our bounds by $O_{n}$, $\Omega_{n}$ and $\Theta_{n}$ respectively and we ignore polynomial multiplicative gaps in $n$. We use $\bar{e}$ to denote the Euler number $\approx 2.72$ and all our logarithms are base $2$ unless stated otherwise. 

\subsection{Existing Results}
We consider an existing upper bound on the output size of any natural join query called the polymatroid bound ($\PMB$ hereon)~\cite{panda,ngo-survey}. The $\PMB$ bound is a refined upper bound than the $\AGM$ bound~\cite{AGM} and exploits degree constraints along with the size bounds on input relations. While the original paper defines the $\PMB$ bound on $S = V$, for our purposes, we need a bound on a subset $S \subseteq V$ of vertices. We note that this generalization is direct. We start with the notion of a degree constraint: 
\begin{definition} [Degree Constraint] \label{defn:dc}
A degree constraint is a triple $(X, Y, N_{Y|X})$, where $X \subsetneq Y \subseteq V$ and $N_{Y | X} \in \mathbb{N}$. The relation $R_e$ is said to guard the degree constraint $(X, Y, N_{Y|X})$ if $Y \subseteq e$ and $\deg_{e}(A_{Y} | A_{X}) := \max_{\mathbf{t}} |\pi_{A_{Y}}( \{\mathbf{t'} \in R_e: \pi_{A_{X}}(\mathbf{t}') = \mathbf{t} \}) | \le N_{Y | X}$.
\end{definition}
We note that each relation $R_e: e \in E$ is a degree constraint of the form $\deg(\emptyset, e, N)$ since $|R_e| \le N$. For our arguments, we will be considering degree constraints that are guarded by some relation $R_e$. We are now ready to define $\PMB(\q, S)$.
\begin{definition} \label{def:dbpgs}
Given $\q = (G, \{R_{e}: e \in E\})$ and $N$, for each $e \in E$, let $\DC$ denote a finite set of degree constraints of the form $(X, Y, N_{Y | X})$ (including the size bound $N$ on relations). Then, $\PMB(\q, S)$ is computed using the following linear program:
\begin{align*}
& \min \sum_{(X, Y, N_{Y | X}) \in \DC}{\delta_{Y | X} \log_{2}(N_{Y | X})} \\
& \text{s.t. } \sum_{(X, Y, N_{Y | X}) \in \DC, v \in S \cap (Y \setminus X)} \delta_{Y | X} \ge 1 \quad \forall v \in S \\
&  \delta_{Y | X} \ge 0 \quad \forall (X, Y, N_{Y | X}) \in \DC. \numberthis \label{eq:pmb}
\end{align*}
\end{definition}
We would like to note here that\footnote{We would like to note here that the LP~\eqref{eq:pmb} for computing $\PMB(\q, S)$ is valid iff the degree constraint graph defined by $\mathcal{G} = (V = S, E = \{ x \rightarrow y : (x, y) \in X \times (Y \setminus X), \forall (X, Y, N_{Y | X}) \in \DC \})$ is acyclic. Note that $E(\mathcal{G})$ is empty when there are only cardinality constraints. In other cases, we explicitly argue how $\PMB(\q, S)$ is still a valid upper bound using ideas from~\cite{ngo-survey}.} $\PMB(\q, S)$ is defined in the log scale. Next, we point out that $\PMB(\q, S)$ is upper bounded by the AGM bound~\cite{AGM} on the subquery induced by attributes in $S$.
\begin{lemma} \label{dbpagmlemma}
Given $\q = (G, \{R_{e}: e \in E\})$, where $|R_e| \le N$ for $e \in E$ and $S \subseteq V$, we have
\[\max_{R_e: |R_e| \le N, e \in E} \PMB(\q, S) \le \AGM(G, S),\]
where $\AGM(G, S)$ is computed by the following linear program for $(G, S)$:
\begin{align*}
& \min{\sum_{e\in E} x_e \log_{2}(N)} \\
& \text{s.t. } \sum_{e \ni v} x_e \ge 1, v \in S \\
&  x_e  \ge 0,  e \in E. \numberthis \label{eq:ubs}
\end{align*}
\end{lemma}
The proof is in Appendix~\ref{sec:pmb-agm}. We conclude this section by proving an upper bound on $\joinUB(G , N, \s)$ for all hypergraphs based on$\PMB(\q, S)$ and $\AGM(G, S)$. 
\subsection{Upper Bound for any Hypergraph $G$} \label{sectionub}
We prove the following result.
\begin{theorem} \label{upperboundtheorem}
Given $G$, $N$ and $\Delta$, we have
\begin{equation} \label{eq:upperbound}
\joinUB(G , N, \s) \le \min_{S \subseteq V:|S|=\s} \max_{R_e: |R_e| \le N, e \in E} 2^{\PMB(\q,S)} \le 2^{\AGM(G, S)}.
\end{equation}
\end{theorem}
The proof is in Appendix~\ref{sec:pmb-ub} and uses the fact that $\PMB(\q, S)$ is a valid upper bound for any join query instance and $\AGM(G, S)$ is a valid upper bound for any worst-case instance $\q$.

%% file: arity_2_poly.tex
\section{$G$ is a graph} \label{section3}
In this section, we look at the case when $G$ is a graph i.e., all relations in the join query have arity at most two. We start with some preliminaries.
\subsection{Preliminaries}
We define some notions related to $G$ that will be used in our arguments.
\begin{definition} [Star]
A \textit{star} is a tree on $n > 1$ vertices with one internal node and $n - 1$ leaves.
\end{definition}
\begin{definition}[Singleton]
A \textit{singleton} is a graph with one vertex, with each edge incident on it i.e., a `self-loop'. 
\end{definition}
\begin{definition}[Disedge]
We call $G$ a \textit{disedge}, if it is a union of vertex-disjoint edges and \textit{singletons}.
\end{definition}
We define the path length as the number of vertices in a path. 
\begin{definition}[Maximum Matching]
A matching $M$ of $G$ is a subset of edges with no common vertices between them. The largest such $M$ (in terms of the number of edges) is a maximum matching.
\end{definition}
\subsection{Main Theorem} \label{sec:atwo-main-thm}
We prove the following result.
\begin{theorem} \label{lowerboundtheorem}
For any graph $G$, $N$ and $\Delta$, we have
\begin{equation} \label{eq:lowerbound}
\joinUB(G, N, \s) \le O_{n}(\cover(G, N, \Delta)).
\end{equation}
\end{theorem}
This result, when combined with Theorem~\ref{thm:proj-cover}, is sufficient to conclude that $\joinUB(G, N, s)$ and $\cover(G, N, \Delta)$ are tight within a factor of $2^{O(n)}$. We prove this theorem in Section~\ref{sec:arity2proof} and present an overview here. We choose a subset $S \subseteq V$ such that $\joinUB(G, N, \s) \le \calB(G, N, \s)$ for some valid upper bound $\calB(G, N, \s)$. Then, we argue that $\cover(G, N, \Delta) \ge \frac{\calB(G, N, \s)}{2^{O(n)}}$ using a case-based analysis depending on the structure of $G$ and the value of $s$. In particular, for each case, these two arguments imply that the inequality~\eqref{eq:lowerbound} holds. 

\subsection{Proof of Theorem~\ref{lowerboundtheorem}}  \label{sec:arity2proof}
We present only an overview here. The complete proof follows from the proofs of five lemmas but we present only the case where the $\AGM$ bound is not tight. Details are in Appendix~\ref{sec:arity2lemmas}.
\begin{proof} [Proof Overview]
We perform a case-based analysis based on the structure of $G$ and the range of $s = n - \Delta + 1$ to prove this theorem. For each case, we prove~\eqref{eq:lowerbound} in two steps. We start by assuming that there exists a subset $S^{*} \subseteq V, |S^*| = s$ such that $\underset{R_e: |R_e| \le N, e \in E}{\max} 2^{\PMB(\q, S^*)}$ is within a constant factor of $\min_{S \subseteq V: |S| = s} \max_{R_e: |R_e| \le N, e \in E} 2^{\PMB(q, S)}$. Then, we provide an algorithm to pick $S^*$ that runs in $O(\poly(|G|))$ time (the latter fact will be used in our algorithmic results). Finally, we argue 
\begin{align*}
\joinUB(G, N, s) \le \max_{R_e: |R_e| \le N, e \in E} 2^{\PMB(\q, S^*)}, 
\end{align*}
followed by arguing
\begin{align*}
\cover(G, N, \Delta) \ge \frac{\max_{R_e: |R_e| \le N, e \in E}2^{\PMB(\q, S^*)}}{2^{O(n)}}. 
\end{align*}
This, when combined with Theorem~\ref{thm:proj-cover}, completes the proof.

We now discuss our case analysis briefly and start by noting that our results hold for very specific decompositions of $G$ (and not any decomposition). We decompose a given $G$ into three (specific) induced subgraphs $G_c$, $G_s$ and $G_t$ such that $|V(G)| = |V(G_s)| + |V(G_t)| + |V(G_c)| = n$. In our decomposition, we ensure that $G_s$ is a collection of {\em stars} of size greater than $1$ and $G_t$ is a collection of {\em singletons}. Further, we have (again through our decomposition) $V(G_c) = V(G) \setminus \{ V(G_s) \cup V(G_t)\}$ and $E(G_c) = \{ e \in E: e \subseteq V(G_c) \}$. Let $M$ be a maximum matching on $G$ and $|M|$ denote the total number of vertices in $M$ ($|M|$ is even). We define $n_I = |V(M) \setminus V(G_c)|$. We summarize our cases using $|M|$ and $n_I$ in Table~\ref{tab:results}. Note that our cases are exhaustive i.e., they cover all possible cases of $s \in [n]$ for every graph $G$. 

\input{table}

Each row in Table~\ref{tab:results} denotes a specific case and we prove~\eqref{eq:lowerbound} using ideas stated in the beginning of this proof. To complete the proof, we need to come up with specific  decompositions of $G$ as stated above (where $G_s$ is a collection of stars of size greater than $1$ and $G_t$ is a collection of singletons). We present these details in Appendix~\ref{sec:g-decomp}.
\end{proof}
For the sake of brevity, we present only two results from Table~\ref{tab:results} here -- $(1)$ the third row where $\AGM$ bound is not tight and $(2)$ the fourth row where we present a non-trivial construction for the lower bound. The remaining results are in Appendix~\ref{sec:arity2lemmas}.
\begin{lemma} \label{lemma2}
Consider the case from row $3$ in Table~\ref{tab:results}. $G$ is not a disedge and $s = n - \Delta + 1 \in [3, |M|]$ is odd. Then, there exists a subset $S \subseteq V$ of size $s$ such that
\begin{equation} \label{eq:lemma2-1}
\joinUB(G, N, s) \le N^{\frac{s}{2}} \le O_{n}(\cover(G, N, \Delta)).
\end{equation}
\end{lemma}
\begin{proof}
Our proof proceeds in two steps. First, we argue
\begin{equation} \label{eq:lemma2-2}
\joinUB(G, N, s) \le N^{\frac{s}{2}},
\end{equation}
followed by arguing that
\begin{equation} \label{eq:lemma2-3}
\cover(G, N, \Delta) \ge \Omega_{n} \left(N^{\frac{s}{2}} \right).
\end{equation}
Note that these two inequalites would immediately imply~\eqref{eq:lemma2-1}. We begin with proving~\eqref{eq:lemma2-2}. In particular, if we prove
\begin{equation} \label{eq:lemma2-4}
\min_{S \subseteq V:|S|=s} \max_{R_e: |R_e| \le N, e \in E} 2^{\PMB(\q,S)} \le n \cdot N^{\frac{s}{2}},
\end{equation}
then~\eqref{eq:lemma2-2} follows from Theorem~\ref{upperboundtheorem}.  To prove~\eqref{eq:lemma2-4}, we claim that there exists a $S$ with size $s$ such that
\begin{equation} \label{eq:lemma2-5}
\PMB(\q, S) \le \log(n + 1) + \dfrac{s}{2} \cdot \log{N}.
\end{equation}
Assuming the above claim is true, we have
\begin{align*}
\min_{S \subseteq V:|S|=s} \max_{R_e: |R_e| \le N, e \in E} 2^{\PMB(\q,S)} \le N^{\frac{s}{2}},
\end{align*}
which follows from~\eqref{eq:lemma2-4}. We drop the multiplicative $n$ factor since our upper and lower bounds anyway have an exponential gap in $n$. In particular, this implies $\joinUB(G, N, s) \le n \cdot N^{\frac{s}{2}}$ from Theorem~\ref{upperboundtheorem}, proving~\eqref{eq:lemma2-2}.

We now prove our claim~\eqref{eq:lemma2-5} and consider two cases.
\begin{itemize}
\item{\textbf{Case $1$}: When there exists at least one value $v_h$ for a vertex $h \in V$ and $v_h \in \Dom(h)$ such that there exists an edge $e = (h, u) \in E$ and there are $\ge \sqrt{N}$ tuples $(v_h, v_u) \in R_e$. In particular, this implies there cannot be more than $\sqrt{N}$ such values in $\Dom(h)$ in all relations incident on $h$ have size at most $N$. We call $h$ a \textit{heavy} vertex.} 
\item{\textbf{Case $2$}: For every vertex $\ell \in V$, every edge $e = (\ell, u)  \in E$ incident on $\ell$ and every value $v_{\ell} \in \Dom(\ell)$, there are at most $\sqrt{N}$ tuples such that $(v_{\ell}, v_u) \in R_e$.}
\end{itemize}
We can decompose any instance as follows -- for every vertex $h \in V$, we check if it is a {\em heavy} vertex and construct a relation on $h$ (i.e., a self-loop) with entries $R_{h}  = \{v_h \in \Dom(h):|(v_h, \cdot) \in R_e|\ge \sqrt{N}, e \in E_h\}$, where $E_h = \{e: h \in e, e \in E\}$. In particular, for each value $v_h \in R_{h}$, there exists at least one relation $R_e$ incident on $h$ such that there are at least $\sqrt{N}$ tuples of the form $(v_h, \cdot) \in R_e$. We now update $R_{e}$ for every relation $e \in E$ incident on $h$ as follows -- $R_{e} \leftarrow R_{e} \setminus \{(v_h, \cdot) \in R_{e}: v_h \in R_{h} \}$. In particular, we remove all tuples in $R_e$ where $\pi_{h}(R_e) = v_h$ for every value $v_h \in R_h$. Let the updated instance be denoted by $\q_h = (G' = (V' = V, E' = (E \setminus\{e \in E:  e \ni h\}) \cup \{h\}), \{R_h\} \cup \{R_e: e \in E' \})$. 

Note that once this is done for every $h \in V$, we have $n$ instances of the form $\q_h$. Further, for each relation $R_e, e \in E$ (these denote the updated relations from the process of removing heavy vertices above) and each vertex incident on the relation, every value of the vertex satisfies the condition in \textbf{Case $2$}. Let this instance denoted by $\q' = (G = (V, E), \{R_e: e \in E\})$. Here, $R_e$s denote the updated relations as defined above and we claim that $\J \subseteq J_{Q'} \cup_{h \in V} J_{Q_{h}}$, where $J_{Q'} =\ \Join_{e \in E} R_e$ and $J_{Q_h} =\ \Join_{e \in E'} R_e$. We prove this claim in Appendix~\ref{sec:heavy-light-d}. Assuming it is true, there are $n + 1$ cases in total and for each case, we only need to pick a subset $S$ to prove
\[\PMB(\q_{h}, S) \le \dfrac{s}{2} \cdot \log{N} \text{ and } \PMB(\q', S) \le \dfrac{s}{2} \cdot \log{N}.\]
Finally, we compute $\sum_{h \in V} \PMB(\q_{h}, S) + \PMB(\q', S)$ (there are $n + 1$ of them), satisfying~\eqref{eq:lemma2-5}, as desired.

\textbf{Case $1$}: Our goal is to pick a subset $S$ with size $s$. Using Definition~\ref{defn:dc}, we can write the cardinality constraint $|R_h|$ on $h$ as $(\emptyset, h, \sqrt{N})$. We pick $S$ as follows -- $S = \{h\} \cup V(M')$, where $M' \subseteq M $ is a subset of edges where no edge is incident on $h'$ and $|M'| = s - 1$. The constraint $s< |M|$ ensures the existence such a $S$ and $S$ can be picked in $\poly{(|G|)}$ time. 

We now compute $\PMB(\q_h, S)$, which is given by
\[\frac{1}{2} \log(N) + \frac{|M'|}{2} \log(N) = \frac{1}{2} \log(N) + \frac{s - 1}{2} \log(N) = \frac{s}{2} \log(N).\]
If we argue that $\delta_{h | \emptyset} = 1$ and $\delta_{(u, v) | \emptyset} = 1$ for every $e = (u, v) \in M'$ is a feasible solution to the LP~\eqref{eq:pmb}, then the stated result follows. We prove this in Appendix~\ref{sec:heavy-case}. Since our argument is independent of $h$, the same $\PMB(\q_h, S)$ is true for every $h \in V$ and we have $n$ bounds in total. We now look at the second case.

\textbf{Case $2$}: Our goal to pick a subset $S$. Using Definition~\ref{defn:dc}, there are two degree constraints on each $e = (\ell, u) \in E$ -- $(\ell, u, \sqrt{N})$ and $(u, \ell, \sqrt{N})$. We pick $S$ in such a way that if any subset of vertices $S' \subseteq S$ belongs to the same connected component $\cc$ in $G$, then every vertex in $S'$ is reachable from every other vertex in $S'$ using only edges in $E_{S'} = \{e \in E: e \subseteq S'\}$. Further, if any subset of vertices $S' \subseteq S$ belongs to the same connected component $\cc$ in $G$, then $|S'| \ge 2$. We argue why such a pick is possible in Appendix~\ref{sec:s-pick}.

We now compute $\PMB(\q', S)$. For simplicity, we assume that all vertices in $S$ belong to the same connected component in $G$ and the set of edges $E_{S} = \{e \in E: e \subseteq S\}$ is a path of the form $\{ e_1 = (p_1, p_2), e_2 = (p_2, p_3), \dots, e_{s} = (p_{s - 1}, p_{s})\}$ where $p_i \in S$ for every $i \in [s]$. Note that $|R_{e_1}| \le N, |R_{e_2}| \le N$, whereas $|R_{e_1} \Join R_{e_2}| \le N^{3/2}$ since 
there are only at most $\sqrt{N}$ tuples of the form $\{(v_{p_2}, v_{p_3}) \in R_{e_2}: v_{p_3} \in \Dom(p_3)\}$ for a fixed $v_{p_2} \in \Dom(p_2)$. It turns out that we can generalize this observation to bound $\PMB(\q', S)$ by
\[\log(N) + \frac{s - 2}{2} \log(N) = \frac{s}{2} \log(N).\]
If we argue that $\delta_{e_1 = (p_1, p_2) | \emptyset} = 1$ with $p_1, p_2 \in S, e_1 \in E(S)$ and $\delta_{p_{i + 1} | p_{i}} = 1$ for every degree constraint $(p_{i}, p_{i + 1}, \sqrt{N})$ guarded by the relation $e_i = (p_{i}, p_{i + 1}) \in E_{S}, i \ge 2$ is a feasible solution to the LP~\eqref{eq:pmb}, then the stated result follows. We prove this in Appendix~\ref{sec:light-case}. It turns out that we can generalize this argument to the case when $E_{S}$ is a tree and when $S$ has multiple conneced components, we take the sum of the bounds obtained on each connected component. Due to space constraints, we defer these details to Appendix~\ref{sec:general-es}. 

We now compute $\PMB(\q, S)$, which is indeed the sum of the bounds obtained on $\PMB(\q_{h}, S)$ for every $v \in V$ and $\PMB(\q', S)$, given by $\log(n + 1) + \frac{s}{2} \log(N)$, proving~\eqref{eq:lemma2-4} as desired.

To complete the proof, we argue~\eqref{eq:lemma2-3}. As shown in Section~\ref{sec:lb}, we only need to construct an instance $\q = (G, \{R_e: e \in E, |R_e| \le N\})$ and a join cover $\J$. Let $q$ be the largest power of $2$ such that $q\le \sqrt{N}$. In particular, we have $q \ge \frac{\sqrt{N}}{2}$. We now instantiate a $\left(n, q^{(n - \Delta + 1)}, \Delta\right)_{q}$ RS code $\mathcal{C}$ over $\mathbb{F}_{q}$. For each edge $e\in E$, we set $R_e =\mathcal{C}_{e}$, where $\mathcal{C}_e \subseteq \mathbb{F}_{q} \times \mathbb{F}_{q} \le N$ implying that $|R_e| = |\mathcal{C}_e| \le q^2 \le N$, as desired. Note that $\mathcal{C} = \J =\ \Join_{e \in E} R_e$. Finally, we have 
\begin{align*}
\cover(G, N, \Delta) \ge |\mathcal{C}| = q^{n - \Delta + 1} = \Omega_{n}\left(N^{\frac{\s}{2}} \right),
\end{align*}
which proves~\eqref{eq:lemma2-3}, as required.
\end{proof}
We now prove the result corresponding to the fourth row in Table~\ref{tab:results}.
\begin{lemma} \label{lemma4}
Consider the case from row $4$ in Table~\ref{tab:results}. $G$ is a {\em disedge} and $s = \s \in [3, |M| - 1]$ is odd. Then, there exists a subset $S \subseteq V$ with size $s$ such that
\begin{equation} \label{eq:lemma4}
\joinUB(G, N, s) \le N^{\frac{s + 1}{2}} \le O_{n}(\cover(G, N, \Delta)).
\end{equation}
\end{lemma}
For proving this result, we state a few known results related to codes, which we will use in proving the lower bound. 
\begin{lemma} \label{extendedrs}
Given a $(n, q^{n - \Delta + 1}, \Delta)_{q}$ RS code $\mathcal{C}$, it can be extended by adding $t$ more evaluation points to obtain a $(n + t, q^{n - \Delta + 1}, \Delta + t)$ RS Code, assuming $q \ge n+t$.
\end{lemma}
\begin{lemma}\label{lemma:linearcode}
The Hamming weight $\wt(\mathbf{t})$ of a tuple $\mathbf{t}$ in a $(n, M, \Delta)$-code $\mathcal{C}$ is the number of non-zero elements in $\mathbf{t}$. The following is true:
\[\Delta = \min_{\mathbf{t}; \mathbf{t} \in \mathcal{C}, \mathbf{t} \neq \mathbf{0}}\wt(\mathbf{t}).\]
Recall that any pair of tuples in $\mathcal{C}$ differ in at least $\Delta$ positions.
\end{lemma}
We define $\Delta_s = (n_s + 1) - (n - \Delta + 1)$, where $n_s = |M|$ and claim the following, 
\begin{claim} \label{claim:delta-ge}
$\Delta_s$ is a positive even when $G$ is a {\em disedge} and $s = \s \in [3, |M| - 1]$ is odd.
\end{claim}	
Assuming the above claim is true, we are ready to prove Lemma~\ref{lemma4}.
\begin{proof} [Proof of Lemma~\ref{lemma4}]
Our proof proceeds in two steps. First, we argue
\begin{equation} \label{eq:lemma4-1}
\joinUB(G, N, s) \le  N^{\frac{s + 1}{2}}
\end{equation}
followed by
\begin{equation} \label{eq:lemma4-2}
\cover(G, N, \Delta) \ge \Omega_{n}\left(N^{\frac{s + 1}{2}} \right).
\end{equation}
Note that these two inequalities would immediately imply~\eqref{eq:lemma4}. We begin with proving~\eqref{eq:lemma4-1}. In particular, if we prove
\begin{equation} \label{eq:lemma4-3}
\min_{S \subseteq V, |S| = s}  \max_{R_e: |R_e| \le N, e \in E} 2^{\PMB(\q, S)} \le N^{\frac{s + 1}{2}},
\end{equation}
then~\eqref{eq:lemma4-1} follows from Theorem~\ref{upperboundtheorem}. If we can show that there exists a $S$ of size $s$ such that
\begin{equation} \label{eq:lemma4-4}
2^{\PMB(\q, S)} \ge \frac{s + 1}{2} \cdot \log(N)
\end{equation}
then using Definition~\ref{def:dbpgs}, we have $\min_{S \subseteq V, |S| = s}  \max_{R_e: |R_e| \le N, e \in E} 2^{\PMB(\q, S)} \le N^{\frac{s + 1}{2}}$, which in turn implies $\joinUB(G, N, s) \le N^{\frac{s + 1}{2}}$ from Theorem~\ref{upperboundtheorem}, proving~\eqref{eq:lemma4-1}. 
	
We now prove~\eqref{eq:lemma4-4}. Towards this end, we pick our subset $S$. We pick as many edges as possible from $M$ and for each edge, we add its endpoints to $S$, until we are left with only one vertex $\{h\}$ to pick. We can now pick $\{h\}$ from any of the unpicked vertices and add it to $S$. The constraints $s < |M|$ and is odd ensure the existence of one such choice of $S$.
	
We compute $\PMB(\q, S)$, which is given by $\frac{s + 1}{2} \cdot \log(N)$ corresponding to a feasible solution $\delta_{e | \emptyset} = 1$ for every $e \in M, e \cap S \neq \emptyset$ and $\delta_{e | \emptyset} = 0$ for all the other edges $e \in E \setminus \{e \in M, e \cap S \neq \emptyset \}$. 
	
We now argue why this is a feasible solution. Note that there are constraints for every $v \in S$ of the form $\sum_{e \in E: v \ni e} \delta_{e | \emptyset} \ge 1$. We set $\delta_{e | \emptyset} = 1$ for every $e \in M, e \cap S \neq \emptyset$. This satisfies all the constraints above since for every $v \in S$, there is a unique edge $e \in M$ incident on $v$ (this follows from our picking algorithm where we pick only matching edges and the fact that $G$ is a disjoint union of a matching and singletons). For all the remaining edges, we can set $\delta_{e | \emptyset} = 0$ and we have a feasible solution.
	
To complete the proof, we prove~\eqref{eq:lemma4-2}. We only need to construct an instance $\q = (G, \{R_e: e \in E, |R_e| \le N\})$ and a join cover $\J$. Let $q$ be the largest power of $2$ such that $q \le N$. In particular, we have $q \ge \frac{N}{2}$. We now instantiate a $\left(\frac{n_s}{2}, q^{\frac{n_s - \Delta_s}{2} + 1}, \frac{\Delta_s}{2} \right)_{q}$ RS code $\mathcal{C}$ over $\mathbb{F}_{q}$. Note that we apply $\mathcal{C}$ on only one endpoint $v(e)$ for every edge $e = (u(e), v(e)) \in M$. Since $\mathcal{C}$ is a RS code, it can be extended over $|V(G_t)|$ vertices (by Lemma~\ref{extendedrs}). We now define a new code $\mathcal{C'}$:
\[
\mathcal{C'}(w) = 
\begin{cases}
\mathcal{C}(w) & \textit{ if } w \in V(G_t) \\
\mathcal{C}(v(e)) & \textit{ if } w = u(e) \textit{ or } w = v(e).\\
\end{cases}
\]
To calculate the distance of $\mathcal{C'}$, we start with a non-zero message $\myvec{m}$ before extension. We already know that $C(\myvec{m})$ has a Hamming weight $\frac{\Delta_{s}}{2}$ from our definition of $\mathcal{C}$. After extension, we get a Hamming weight of $\frac{\Delta_{s}}{2} + n_t$. Recall from our definition of $\mathcal{C'}$ that we end up duplicating only vertices in $G_s$. All vertices in $G_t$ are extended as-is from $\mathcal{C}$. As a result, for $\mathcal{C'}$, we have a Hamming weight of $2 \cdot \frac{\Delta_{s}}{2} + n_t$, resulting in a total Hamming distance of $\Delta = \Delta_s + n_t$ (using Lemma~\ref{lemma:linearcode}). Thus, $\mathcal{C'}$ is a $\left(n = n_s + n_t, q^{\frac{n_s - \Delta_s}{2} + 1}, \Delta \right)_{q}$ code.
	
For each edge $e \in M$, we set $R_e = \mathcal{C'}_{e}$, where $\mathcal{C'}_e \subseteq \{(i, i) | i \in \mathbb{F}_q\}$ implying that $|R_e| = |\mathcal{C}_e| \le N$. Observe that for every $e \in G_t$ (only `self-loops'), we set $R_e = \mathcal{C}_e$, where $\mathcal{C}_{e} \subseteq \{i | i \in \mathbb{F}_q\}$ implying that $|R_e| = |\mathcal{C}_e| \le N$. Finally, we get $\cover(G, N, \Delta) \ge |\mathcal{C}| = q^{\frac{n_s - \Delta_s}{2} + 1} = q^{\frac{n_s+n_t- \Delta}{2} + 1} = \Omega_{n}\left(N^{\frac{s + 1}{2}} \right)$, which proves~\eqref{eq:lemma4-2}, completing the proof.
\end{proof}
We conclude this section by proving Claim~\ref{claim:delta-ge}.
\begin{proof} [Proof of Claim~\ref{claim:delta-ge}]
We start by noting that $\s$ is odd and $n_{s}$ is even (since $n_s = |M|$). Note that this implies $n_s + 1$ is odd and thus we have $\Delta_{s} = (n_s + 1) - (n - \Delta + 1)$ is even. We still need to argue that $\Delta_{s} > 0$. For this, we only need to show $n_s = |M| > n - \Delta$, which follows from our initial assumption that $|M| > n - \Delta + 1$. 
\end{proof}

%% file: table.tex
\begin{table*}[thbp!]
	{\small
	\centering
	{
		\hspace{0.5cm} {
	\begin{tabular}{| c | c | c | c |}
		\hline
		$n - \Delta + 1$ & $G$ & Is AGM tight? \\
		\hline
		\rowcolor{red!30} $= 1$ & - & Yes \\
		\hline
		\rowcolor{red!30} $\in [2, |M|]$, even &  -  & Yes \\
		\hline
		\rowcolor{green!50} $\in [3,|M| - 1]$, odd & not \textit{disedge} & No \\
		\hline
		\rowcolor{red!30} $\in [3,|M| - 1]$, odd & \textit{disedge} & Yes \\
		\hline
		\rowcolor{red!30} $\in [|M| + 1, |M| + n_I]$ & -  & Yes \\
		\hline
		\rowcolor{red!30}  $\in [|M| + n_I + 1, n]$ & - & Yes \\
		\hline
	\end{tabular}}
	\caption{The first column denotes the value of $n - \Delta + 1$ and the second column denotes constraints on $G$. `-' denotes no restrictions on $G$. Note that only for rows three and four, we will be imposing restrictions on $G$. The final column shows if the $\AGM$ bound is tight i.e., it is tight if $\max_{R_e, |R_e| \le N, e \in E} \PMB(\q = (G, \{R_e: e \in E\}), S) = \AGM(G, S)$ and not tight otherwise. We state only the result in row $3$ in the main paper and the remaining results are deferred to the Appendix~\ref{sec:arity2lemmas}.} \label{tab:results}
}
}
\end{table*}

%% file: general_arity.tex
\section{General Hypergraphs $G$} \label{sec:gen-arity}
In this section, we extend our results to general hypergraphs $G$. We prove the following result.
\begin{theorem} \label{generalarityubtheorem}
For all hypergraphs $G$, $N$, $\Delta$, we have
\begin{equation} \label{eq:gubt1}
\joinUB(G, N, \s)  \le \min\left(O_{n}(N \cdot \cover(G, N, \Delta)^{10.37}), O_{n}(N^{O(\log(s))} \cdot \cover(G, N, \Delta)^{3.73})\right). 
\end{equation}
and 
\begin{equation} \label{eq:glb1}
\frac{\joinUB(G, N, \s)}{\cover(G, N, \Delta)} \ge N^{1 + \frac{1}{\bar{e}} - \delta}
\end{equation}
for some small constant $\delta > 0$.
\end{theorem}
This result, when combined with Theorem~\ref{upperboundtheorem}, is sufficient to conclude that $\joinUB(G, N, \s)$ and $\cover(G, N, \Delta)$ differ by at most a polynomial in $\cover(G, N, \Delta)$ (which in turn is a polynomial in $N$) and that a polynomial gap in $N$ is necessary. We prove this Theorem in Appendix~\ref{sec:garity-proof} and present an overview here. 

We start by defining two linear programs:
\begin{align*}
&  \LP_{\lb}(G, s) := \max{L} &  \LP_{\ub}(G, s) := \min{\sum_{e \in E} x_{e}} \\
& \text{s.t. } \sum_{v \in e} y_{v} \le 1, \text{ for all } e \in E &  \text{s.t. } \sum_{e \ni v}x_e \ge 1\text{ for all } v \in V  \\
& \sum_{v \in S: S \subseteq V, |S| = s} y_{v} \ge L & \sum_{v \in V} z_{v} \ge s \\
& y_{v} \ge 0, v \in V.  &  z_{v} \le 1, v \in V, x_e  \ge 0,  e \in E.
\end{align*}
For any $G$, $N$ and $\s$, we claim the following. Any feasible solution to the LP on the left (i.e., $\LP_{\lb}$) can be converted to an instance $\{R_e: e \in E\}$ with a join cover $\calU = \J =\ \Join_{e \in E} R_e$ such that 
\begin{align*}
\cover(G, N, \Delta) \ge |\calU| = \Omega_{n}\left(N^{\LP_{\lb}(G, \s)} \right).
\end{align*} 
The proof is in Appendix~\ref{sec:cover-lb}. 

Next, we argue that an integral version of LP (where $z_{v} \in \{0, 1\}$) on the right (i.e., $\LP_{\ub}$) computes $\AGM(G, S)$. We denote both this integral version (and with a slight abuse of notation, its objective value) by $\LP^*_{\ub}(G, \s)$. By defintion, we have $\LP^*_{\ub}(G, \s) \le \LP_{\ub}(G, \s)$. More formally, we prove
\begin{align*}
\joinUB(G, \s) & \leq O_{n}\left(N^{\LP^*_{\ub}(G, \s)}\right)\\
& = \min_{S \subseteq V: |S| = s} 2^{\AGM(G, S)}.
\end{align*}
The proof is in Appendix~\ref{sec:proj-ub}.

Given this setup, to prove~\eqref{eq:gubt1} and~\eqref{eq:glb1}, we need to obtain upper and lower bounds on the gap between $\LP_{\ub}(G, \s)$ and $\LP_{\lb}(G, \s)$. It turns out that this gap $\LP_{\ub} (G, \s)$ is exactly captured by the following -- the (dual of) $\LP_{\lb}(G, \s)$ computes a fractional edge covering that covers at least $\s$ vertices in $V$ fractionally and $\LP_{\ub}(G, \s)$ computes an optimal fractional edge covering that covers at least $\s$ vertices integrally. We start with the following result (which proves~\eqref{eq:gubt1}).
\begin{theorem} \label{thm:ubga1}
For all hypergraphs $G = (V, E)$ and $1 \le s \le n$, we have
\begin{enumerate} [label=(\ref{thm:ubga1}\alph*), itemsep=0pt, leftmargin=*]
\item $\LP_{\ub}(G, s) \le 10.37 \cdot \LP_{\lb}(G, s) + 1.$ \label{thm:ubga2}
\item $\LP_{\ub}(G, s) \le 3.73 \cdot \LP_{\lb}(G, s) + O(\log(s)).$ \label{thm:ubga3}
\end{enumerate}
\end{theorem}
We prove two results to show a tradeoff between the multiplicative factor and an additive factor depending on $n$, combining which gives~\eqref{eq:gubt1} (the proof is Appendix~\ref{gaub}). Our proof proceeds by rounding the (dual of) $\LP_{\lb}(G, s)$ using a randomized dependent rounding algorithm to convert an optimal solution for $\LP_{\lb}(G, s)$ to a feasible solution for $\LP_{\ub}(G, s)$. 

Finally, we complement this result by showing the following (which proves~\eqref{eq:glb1}): 
\begin{theorem} \label{thm:lbga1}
There exists an instance $G' = (V', E')$ such that for $s = \Theta\left(\frac{n}{\log(n)} \right)$, we have
\[\dfrac{\LP_{\ub}(G, s)}{\LP_{\lb}(G, s)} \ge 1 + \frac{1}{\bar{e}} - \delta,\]
where $\delta$ is a small constant $> 0$.
\end{theorem}
The proof is in Appendix~\ref{sec:galb}.		

%% file: algorithmic_version.tex
\section{Algorithms} \label{sec:algos}
In this section, we present algorithms that solve the join covering problem approximately given $\q = (G, \{R_e: e \in E\})$, $N$ and $\Delta$. For the case when $G$ is a graph, we present an algorithm that computes a join cover $\calU \subseteq \J$ that can be off from $\cover(G, N, \Delta)$ by at most a factor of $2^{n}$. For general hypergraphs $G$, the join cover $\calU$ computed by our algorithm can be off from the optimal $\cover(G, N, \Delta)$ by a polynomial factor in $\cover(G, N, \Delta)$. Our algorithmic results rely on a search to decision reduction using the Boolean Conjunctive query problem on $\q$ (denoted by $\BCQ_{\q}$), which returns $1$ if $\J =\ \Join_{e \in E} R_e \neq \emptyset$ and $0$ otherwise for any instance $\q$. Note that if we compute a join cover $\calU$ on the instance $\q$, then we can solve $\BCQ_{\q}$ as well. In particular, if $\calU \neq \emptyset$, $\BCQ_{\q} = 1$ and $0$ otherwise. We consider the reverse direction, starting with the case when $G$ is a graph.
\begin{lemma} \label{lemma5.2}
For any $\q$ (where $G$ is a graph), $N$, $\Delta$ and $s = n - \Delta + 1$, given a Blackbox algorithm $\calA$ that solves the instance $\BCQ_{\q}$, there exists an algorithm $\calB_{\q}$ that computes a join cover $\calU$ such that 
\begin{align*}
\frac{|\calU|}{\cover(G, N, \Delta)} \le \frac{\joinUB(G, N, s)}{\cover(G, N, \Delta)} \le 2^{n}
\end{align*}
in time $O_{n}(N \cdot \cover(G, N, \Delta) \cdot \Time(\calA))$. Here, $\Time(\calA)$ is the runtime of $\calA$.
\end{lemma}
We present the key ideas of $\calB_{\q}$ here and the detailed proof is in Appendix~\ref{sec:arity2algo}. We first compute $\calU_{S}$ such that $|\calU_{S}| \le \joinUB(G, N, s)$ in worst-case optimal time $O(\joinUB(G, N, s))$ using Algorithm $3$ in~\cite{ngo-survey}. We still need to construct to $\calU$ from $\calU_{S}$ and for this purpose, we make use of $\calA$. In particular, at each step, we pick a vertex $v \in V \setminus S$ and compute $\calU_{S \cup \{v\}} = \calU \times \Dom(v)$. Then, we argue that we can filter $\calU_{S \cup \{v\}}$ in such a way that $|\calU_{S \cup \{v\}}| \le |\calU|$ and the filtered $\calU_{S \cup \{v\}}$ is a join cover on attributes in $S \cup \{v\}$ where for every tuple $\mathbf{t} \in \J$, there exists a tuple $\mathbf{t'} \in \calU_{S \cup \{v\}}|$, we have $\Dist(\mathbf{t}, \mathbf{t'}) \ge \Delta - 1$. We can now argue that the final $\calU$ is a join cover by induction.

For general hypergraphs $G$, we present a weaker result.
\begin{lemma} \label{lemma5.3}
For general hypergraphs $G$, $N$, $\Delta$ and $s = n - \Delta + 1$, given a Blackbox algorithm $\calA$ that solves the instance $\BCQ_{\q}$, there exists an algorithm $\calB_{\q}$ that computes a join cover $\calU$ such that $\frac{|\calU|}{\cover(G, N, \Delta)} \le \frac{\joinUB(G, N, s)}{\cover(G, N, \Delta)} \le \min(O_{n}(N \cdot N \cdot \cover(G, N, \Delta)^{9.37} \cdot \Time(\calA)), O_{n}(N \cdot N^{O(\log(n - \Delta + 1))} \cdot \cover(G, N, \Delta)^{2.73} \cdot \Time(\calA)))$
in time $\min(O_{n}(N \cdot N \cdot \cover(G, N, \Delta)^{10.37} \cdot \Time(\calA)), O_{n}(N \cdot N^{O(\log(n - \Delta + 1))} \cdot \cover(G, N, \Delta)^{3.73} \cdot \Time(\calA)))$. Here, $\Time(\calA)$ is the runtime of $\calA$.
\end{lemma}
In fact, the same proof as that of Lemma~\ref{lemma5.2} works here.

%% file: ack.tex
\section*{Acknowledgments}
We are greatly indebted to Hung Ngo for helpful discussions at initial stages of this work. We thank Surajit Chaudhuri, Manas Joglekar, Chris R\'{e} and Muthian Sivathanu for useful discussions. This work was supported by NSF CCF-$171734$. SVMJ thanks MSR for hospitality, where a part of this work was done.

%% file: appendix_results.tex
\section{Missing Details in Section~\ref{sec:overview}} \label{appendix:prelims}

\subsection{Proof of Theorem~\ref{thm:proj-cover}} \label{sec:proj-cover}
We state and prove a more detailed version of Theorem~\ref{thm:proj-cover} here. We begin by defining a related problem called the join packing problem. Given $G$, $N$ and $\Delta$, we would like to compute
\begin{equation}
\pack(G, N, \Delta) = \max_{R_e: |R_e| \le N, e \in E} |\calC|,
\end{equation}
where $\q = (G, \{R_e: e \in E\})$ and $\calC$ is the output of the join packing problem defined below.
\begin{problem}[Join Packing problem] \label{prob:packing}
Given $\q = (G, \{R_{e} : e \in E\})$ and $\Delta$, the goal is to output a $\calC$ such that
\[\calC = \argmax_{\calD: \calD \text{ is a join packing of } \J} |\calD|.\]
Here, a join packing is a subset $\calD \subseteq \J = \Join_{e \in E} R_e$ such that for any pair of distinct tuples $\mathbf{c}, \mathbf{c'} \in \calD$, $\Dist(\mathbf{c},\mathbf{c'}) \ge \Delta$, where $\Dist$ is defined by~\eqref{eq:hd}.   
\end{problem}
We are now ready to prove the theorem.
\begin{theorem} \label{thm:proj-cover-2}
For every $G$, $N$, $\Delta$ and $s = n - \Delta + 1$, we have
\begin{equation}
\cover(G,N,\Delta) \le \pack(G, N, \Delta) \le \joinUB(G, N, s).
\end{equation}
\end{theorem}
We are now ready to prove Theorem~\ref{thm:proj-cover-2}.
\begin{proof} [Proof of Theorem~\ref{thm:proj-cover-2}]
We start with the first inequality:
\begin{equation} \label{eq:cover-pack}
\cover(G,N,\Delta) \le \pack(G, N, \Delta).
\end{equation}
For any instance $\q$, $N$ and $\Delta$, let $\calC$ be a solution to the join packing problem. Note that for any tuple $\mathbf{t} \in J$, there exists a tuple $\mathbf{c} \in \calC$ such that $\Dist(\mathbf{t}, \mathbf{c}) < \Delta$. Otherwise, we can add $\mathbf{t} \in \J$ to $\calC$ contradicting the fact that $\calC$ is the largest subset of $\J$. Hence, $\calC$ is a valid join cover (i.e., $\calU = \calC$). It follows that 
\[|\calS| =  \min_{\calU : \calU \text{ is a join cover of } \J} |\calU| \le |\calC|.\]
Note that this implies
\[\cover(G, N, \Delta) = \max_{R_e: |R_e| \le N, e \in E} |\calS| \le \max_{R_e: |R_e| \le N, e \in E} |\calC| = \pack(G, N, \Delta), \]
which in turn proves~\eqref{eq:cover-pack}.

We now turn our attention to the second inequality:
\[\pack(G, N, \Delta) \le \joinUB(G, N, s).\]
We prove this using a property of Hamming distance, which we describe here. Let $\calC$ be a join packing such that $|\calC| = \pack(G, N, \Delta)$. For a fixed subset $S$ of attributes with $|S| = s$, we define $\calC(S) = \{\mathbf{c}_{S}: \mathbf{c} \in \calC \}$, where $\mathbf{c}_{S}$ is $\mathbf{c}$ projected down to attributes in $S$. Then, we have that the tuples in $\calC(S)$ are pairwise distinct (note that $\calC(S)$ is a set) i.e., $\Dist(\mathbf{c}_{S}, \mathbf{c'}_{S}) = |\{ s' \in S :\mathbf{c}_{s'} \neq \mathbf{c'}_{s'}\}| \ge 1$. If not, there exists at least a pair of tuples $\mathbf{c}, \mathbf{c'} \in \calC$ such that $\Dist(\mathbf{c}, \mathbf{c'}) \le \Delta - 1$, which contradicts the assumption that $\calC$ is a packing. Since our argument is independent of the choice of $S$, we have
\[\pack(G, N, \Delta) \le \max_{R_e: e \in E, |R_e| \le N} \min_{S \subseteq V: |S| = s} |\calC(S)|.\]
Note that $\calC \subseteq \J$ implies $|\calC(S)| \le |J_{S}|$, where $|J_{S}| = |\pi_{e \cap S} R_{e}|$ for any $S : |S| = s$. Thus, we have
\[\max_{R_e: e \in E, |R_e| \le N} \min_{S \subseteq V: |S| = s} |\calC(S)| \le \min_{S \subseteq V: |S| = s} \max_{R_e: e \in E, |R_e| \le N} \left|\Join_{e \in E} \pi_{e \cap S} (R_e) \right| = \joinUB(G, N, s).\]  
\end{proof}

%% file: appendix_arity2.tex
\section{Missing Results in Section~\ref{section2}}
\subsection{Proof of Theorem~\ref{dbpagmlemma}} \label{sec:pmb-agm}
\begin{proof}
Consider the case when there are only cardinality constraints in $\q$ and they are of the form $\deg(\emptyset, e, N)$ for every $e \in E$.  In this case, note that the LP to compute $\PMB(\q, S)$ can be rewritten as follows:
\begin{align*}
& \min \sum_{(\emptyset, e, N) \in \DC} \delta_{e | \emptyset} \log_{2}(N) \\
& \text{s.t. } \sum_{(\emptyset, e, N) \in \DC, v \in (S \cap e)} \delta_{e | \emptyset} \ge 1, v \in S \\
&  \delta_{e | \emptyset} \ge 0,  (\emptyset, e, N) \in \DC.
\end{align*}
Note that this is exactly the LP that computes $\AGM(G, S)$. In particular, by setting $x_e = \delta_{e | \emptyset}$ for every $e \in E$, we can recover~\eqref{eq:ubs}. Since $\AGM(G, S)$ is a valid upper bound on the log of output size of any join query with worst-case inputs (i.e., $|R_e| \le N$ for every $e \in E$), the stated claim follows. 
\end{proof}

\subsection{Union bound for $\PMB(\q, S)$} \label{dbpunionlemmap}
\begin{lemma} \label{dbpunionlemma}
	Let $\q_1 = (G_1, \{\pi_{V(G_1)} (R_{e}) : e \in E\}), \q_2 = (G_2, \{\pi_{V(G_2)} (R_{e}): e \in E\}), \dots , \q_k = (G_k, \{\pi_{V(G_k)} (R_{e}): e \in E\}),$ be decompositions of $\q$ such that $\cup_{i=1}^{k}V(G_i) = V$. We define $S_i = S \cap V(G_i)$ for every $i \in [k]$. Then, we have
	\begin{equation} \label{eq:pmbunion}
	\PMB(\q, S) \le \sum_{i = 1}^{k} \PMB(\q_i, S_i).
	\end{equation}
\end{lemma}
\begin{proof}
	Let $\DC$ be the set of degree constraints on $G$ of the form $(X, Y, N_{Y | X})$ and let $\DC_i$ be the set of degree constraints on $G_i$ of the form $(X_i, Y_i, N_{Y_i | X_i})$. For each $q_i$, there is a LP that computes $\PMB(q_i, S_i)$ (invoking~\eqref{eq:pmb}) given by:
	\begin{align*}
	& \min \sum_{(X_i, Y_i, N_{Y_i | X_i}) \in \DC_{i}}{\delta_{Y_i | X_i} \log_{2}(N_{Y_i | X_i})} \\
	& \text{s.t. } \sum_{(X_i, Y_i, N_{Y_i | X_i}) \in \DC_{i}, v \in S_i \cap (Y_i \setminus X_i)} \delta_{Y_i | X_i} \ge 1, v \in S_{i} \\
	&  \delta_{Y_i | X_i} \ge 0,  (X_i, Y_i, N_{Y_i | X_i}) \in \DC_{i}. \numberthis \label{eq:pmb-i}
	\end{align*}
	Given an optimal solution $\mathbf{\delta}_{i} = (\delta_{Y_i | X_i})_{(X_i, Y_i, N_{Y_i | X_i}) \in \DC_i}$ to $\PMB(\q_i, S_i)$ for each $i \in [k]$, we only need to argue that we can combine them to imply a feasible solution for the LP computing $\PMB(\q, S)$. In particular, this would prove~\eqref{eq:pmbunion}. Consider the set of constraints that need to be satisfied for each $i \in [k]$:
	\begin{align*}
	\sum_{(X_i, Y_i, N_{Y_i | X_i}) \in \DC_{i}, v \in S_i \cap (Y_i \setminus X_i)} \delta_{Y_i | X_i} \ge 1, v \in S_{i}.
	\end{align*}
	Note that we can rewrite it as follows:
	\begin{align*}
	\sum_{i = 1}^{k} \sum_{(X_i, Y_i, N_{Y_i | X_i}) \in \DC_{i}, v \in S \cap (Y_i \setminus X_i)} \delta_{Y_i | X_i} \ge 1, v \in S.
	\end{align*}
	In particular, for each $v \in S$, the following condition is true on $\DC$:
	\begin{align*}
	\sum_{(X, Y, N_{Y | X}) \in \DC, v \in S \cap (Y \setminus X)} \delta_{Y | X} \ge 1.	
	\end{align*}
	Thus, the constraints in the LP computing $\PMB(\q, S)$ are satisfied for every $v \in S$, proving~\eqref{eq:pmbunion} completing the proof.
\end{proof}

\section{Proof of Theorem~\ref{upperboundtheorem}} \label{sec:pmb-ub}
\begin{proof}
Recall the definition of $\joinUB(G, N, \s)$~\eqref{eq:free-vars}:
\[\min_{S \subseteq V: |S| = \s} \max_{R_e: e \in E, |R_e| \le N}  \left| \Join_{e \in E} \pi_{e \cap S}(R_e) \right|.\]
Since $\PMB(\q, S)$ is a valid upper bound on the log of output size of any join query for worst-case input relations $R_e$ satisfying $|R_e| \le N$ (by Lemma~\ref{dbpagmlemma}), we have
\[\max_{R_e: e \in E, |R_e| \le N} |\Join_{e \in E} {\pi_{e \cap S}(R_e)}| \quad \le \max_{R_e: |R_e| \le N, e \in E} 2^{\PMB(\q, S)} \le 2^{\AGM(G, S)}.\] 
Since the above inequality holds for any $S \subseteq V, |S| = \s$, this completes our proof. 
\end{proof}

\section{Missing Results in Section~\ref{section3}} \label{sec:ari}

\subsection{Missing Details in Proof of Lemma~\ref{lemma2}}

\subsubsection{Correctness of Our Heavy-Light Decomposition} \label{sec:heavy-light-d}
We claim that $\J \subseteq J_{\q'} \cup_{h \in V} J_{\q_{h}}$, where $J_{Q'} =\ \Join_{e \in E} R_e$ and $J_{Q_h} =\ \Join_{e \in E'} R_e$. Let $J'_{\q} \subseteq \J$ such that for every tuple $\mathbf{t} \in J'$ and for every vertex $v \in V$, all relations $R_{e}$ in $\q$ incident on $v$ have only at most $\sqrt{N}$ tuples with the value $\pi_{v}(\mathbf{t})$ for $v$. Since this is the exact criterion satisfied by all values of all vertices in $\q' = (G = (V, E), \{R_e: e \in E\})$ on all relations, we have that $J'_{\q} \subseteq J_{\q'}$. 

For each remaining tuple $\mathbf{t} \in \J \setminus J'_{\q}$, there is at least one vertex $h \in V$ such that, for at least one relation $R_{e}$ in $\q$ incident on $h$, there are at least $\sqrt{N}$ tuples with the value $\pi_{h}(\mathbf{t})$ for $h$ (this is done in the same order the {\em heavy} vertices were processed earlier). If we argue that $\mathbf{t} \in J_{\q_h} =\ \Join_{e \in E'}R_e$, where $\q_h = (G' = (V' = V, E' = (E \setminus\{e \in E: h \in e\}) \cup \{h\}), \{R_h\} \cup \{R_e: e \in E' \})$, then we would be done. We start by recalling that $\pi_{e}(\mathbf{t}) \in R_e$ for every $e \in E$ and show that $\pi_{e}(\mathbf{t}) \in R_e$ for every $e \in E'$. 
%We decompose $J_{\q_h}$ as follows -- $J_{q_h} =\ (\Join_{e \in E \cap E'} R_{e}) \times R_{h}$. 
%Note that the difference between $E'$ and $E$ is that there are no relations incident on $\{h\}$ (except a {\em self-loop}). 
In particular, we have $\pi_{h}(\mathbf{t}) \in \pi_{h}(R_{e})$ for every $e \in E$ such that $h \in e$. In words, all relations incident on $h$ have a tuple with value $\pi_{h}(\mathbf{t})$ for $h$ in $E$. Note that this implies $\pi_{h}(\mathbf{t}) \in R_{h}$ (which follows from the definition of $R_{h}$). To summarize, since $E' \setminus \{h\} \subseteq E$, we have indeed shown that for each edge $e \in E'$, $\pi_{e}(\mathbf{t}) \in R_{e}$ and as a result, $\mathbf{t} \in J_{\q_{h}}$. This completes the proof.

\subsubsection{Upper bound for $\PMB(\q_h, S)$} \label{sec:heavy-case}
Recall that $\q_h = (G' = (V' = V, E' = (E \setminus\{e \in E: h \ni e\}) \cup \{h\}), \{R_h\} \cup \{R_e: e \in E' \})$ and $S = \{h\} \cup V(M')$, where $M' \subseteq M, |M'| = s - 1$ is a subset of edges where no edge is incident on $h'$. We need to argue two things -- $(1)$ the LP~\eqref{eq:pmb} gives a valid upper bound for $\PMB(\q_h, S)$ and $(2)$ the solution $\delta_{h | \emptyset} = 1$ and $\delta_{(u, v) | \emptyset} = 1$ for every $e = (u, v) \in M'$ is a feasible solution for the LP~\eqref{eq:pmb}. 

We start with $(1)$. Since we have only cardinality constraints in $\q_{h}$, the degree constraint graph $\mathcal{G} = (V = S, E = \{ x \rightarrow y : (x, y) \in X \times (Y \setminus X), \forall (X, Y, N_{Y | X}) \in \DC \})$ is acyclic since $X = \emptyset$ for every $(X, Y, N_{Y | X}) \in \DC$ (and as a result, $E(\mathcal{G}) = \emptyset$). Since $\mathcal{G}$ is acyclic, LP~\eqref{eq:pmb} is a valid upper bound for $\PMB(\q_h, S)$ (Section $5.1$ in~\cite{ngo-survey}).

We now prove $(2)$. Note that there are $|S|$ constraints to be satisfied one for each vertex in $S$, each one of the form -- $\sum_{e \in E': v \ni e} \delta_{e | \emptyset} \ge 1$ for every $v \in S$. We now argue that our solution satisfies all these constraints. For $h$, there is only one edge incident on $h$ (i.e., the self-loop on $h$), which we set $\delta_{h | \emptyset} = 1$. We now consider pairs of vertices $(u, v) \in S$ such that the edge $e = (u, v) \in M'$, which we set  $\delta_{(u, v) | \emptyset} = 1$. Thus, for each vertex in $S$, there is at least one edge $e \in E'$ such that $\delta_{e | \emptyset} = 1$, implying that our solution is feasible for $\PMB(\q, S)$. Since our argument is independent of $h$, it holds for every $q_{h}, h \in V$.  

\subsubsection{Algorithm for picking $S$ and Proof that $|S| = \s$} \label{sec:s-pick}
Recall that in this case, $G$ is not a {\em disedge} and $|S| = s = n - \Delta + 1 \in [3, |M|]$ is odd. Our goal is to pick $S$ of size $s$ and argue why such a pick is always possible.

To pick $S$, we use the Light Picking Algorithm $1$ (Algorithm~\ref{lightpickingalgorithm1}), whose summary we provide here. Recall that $G$ is not a \textit{disedge}, implying that there exists at least one connected component ($\cc$ hereon) of size $> 2$ in $G$. We pick as many connected vertices as possible from each $\cc$ in $G$ in decreasing order of size and keep track of the last $\cc$ that was picked with size $> 2$. If we have only one vertex remaining to be picked on a particular $\cc$ (say $g$), we remove any one leaf vertex from the last $\cc$ we picked (say $g'$) and go on to pick two connected vertices in $g$ (such a pick always exists in $G$ since $\s < |M| \le n$). It still needs to be argued that our algorithm picks exactly $s$ vertices, which we do after describing the algorithm.

\begin{algorithm}[hbtp!]
	\caption{Light Picking Algorithm $1$}
	\label{lightpickingalgorithm1}
	\begin{algorithmic}[1]
		\Require{$G$}
		\Ensure{$S$}
		\State $S \gets \emptyset, \ \lastpick \gets \emptyset$
		\State \textbf{for all} connected components $\cc \in G$, in decreasing order of size, \textbf{do}
		\State\hspace*{\algorithmicindent} $\topick  = \min{(s - |S|, |cc|)}$  \label{State:check-cc} 
		\State\hspace*{\algorithmicindent} \textbf{if} $\topick > 1$ \textbf{then} \label{State:topic-more-than-1}
		\State\hspace*{\algorithmicindent}\hspace*{\algorithmicindent} $T \gets \{ \text{any } \topick \text{ length path in } cc \}$
		\State\hspace*{\algorithmicindent}\hspace*{\algorithmicindent} $S \gets S \cup T, G \gets G \setminus T$
		\State\hspace*{\algorithmicindent}\hspace*{\algorithmicindent} \textbf{if} {$|T| > 2$} \textbf{then} $\lastpick \gets T$
		\State\hspace*{\algorithmicindent} \textbf{else}
		\State\hspace*{\algorithmicindent}\hspace*{\algorithmicindent} \textbf{if} $\topick = 1$ \textbf{then}
		\State\hspace*{\algorithmicindent}\hspace*{\algorithmicindent}\hspace*{\algorithmicindent}  $T \gets \{ \text{any one leaf vertex in } \lastpick\}$ \label{State-T-Gets}
		\State\hspace*{\algorithmicindent}\hspace*{\algorithmicindent}\hspace*{\algorithmicindent}  $S \gets S \setminus T$
		\State\hspace*{\algorithmicindent}\hspace*{\algorithmicindent}\hspace*{\algorithmicindent}  $T' \gets \{\text{endpoints of an edge in } \cc\}$ \label{State-Tp-Gets} 
		\State\hspace*{\algorithmicindent}\hspace*{\algorithmicindent}\hspace*{\algorithmicindent}  $S \gets S \cup T'$
		\State \hspace*{\algorithmicindent}\hspace*{\algorithmicindent} \textbf{return} $S$
	\end{algorithmic}
\end{algorithm}

We prove here that $|S| = \s$ after running Algorithm~\ref{lightpickingalgorithm1}. We know that $\lastpick$ is never empty when $\topick \le 1$ as there always exists a $\cc$ in $G$ with size greater than two ($G$ is not a \textit{disedge}) and we start our picking from one such $\cc$ having maximum size. As a result, we can always remove a vertex from $\lastpick$. The constraint $n - \Delta + 1 < |M| \le |V(G)|$ ensures that there will always be two directly connected vertices left to pick in the $\cc$ that we are currently processing. In particular, this implies our algorithm exactly picks $n - \Delta + 1$ and runs in $\poly{(|G|)}$ time. 

\subsubsection{Upper bound for $\PMB(\q', S)$} \label{sec:light-case}
Recall that $\q' = (G = (V, E), \{R_e: e \in E\})$ and $S$ is a single connected component. Further, the set of edges $E_{S} = \{e \in E: e \subseteq S\}$ is a path of the form $\{ e_1 = (p_1, p_2), e_2 = (p_2, p_3), \dots, e_{s} = (p_{s - 1}, p_{s})\}$ where $p_i \in S$ for every $i \in [s]$. In addition to cardinality constraints $(\emptyset, e = (\ell, u), N)$ on each $e \in E$, there are two degree constraints for every edge of the form $(\ell, u, \sqrt{N})$ and $(u, \ell, \sqrt{N})$.We need to argue two things -- $(1)$ the LP~\eqref{eq:pmb} gives a valid upper bound for $\PMB(\q', S)$ and $(2)$ the solution $\delta_{e_1 = (p_1, p_2) | \emptyset} = 1$ with $p_1, p_2 \in S, e_1 \in E(S)$ and $\delta_{p_{i + 1} | p_{i}} = 1$ for every degree constraint $(p_{i}, p_{i + 1}, \sqrt{N}) \in \DC$ guarded by the relation $e_i = (p_{i}, p_{i + 1}) \in E_{S}, i \ge 2$ is a feasible solution to the LP~\eqref{eq:pmb}. 

We start with $(1)$ first. We start by constructing the degree constraint graph $\mathcal{G}$. We set $V(\mathcal{G}) = S$ and build $E(\mathcal{G})$ as follows. For each $e_i = (p_{i}, p_{i + 1}) \in E_{S}, i \ge 1$, we add the edges $p_{i} \rightarrow p_{i + 1}$ and $p_{i + 1} \rightarrow p_{i}$ to $E(\mathcal{G})$ (following the degree constraints $(p_{i}, p_{i + 1}, \sqrt{N})$ and $(p_{i + 1}, p_{i}, \sqrt{N})$). It follows that $E(\mathcal{G})$ is cyclic. If we remove enough edges from $E(\mathcal{G})$ such that the resulting edgeset is -- $(a)$ acyclic and $(b)$ the output $\Join_{e \in E(\mathcal{G})}R_e$ is bounded and covers all vertices in $S$, then~\eqref{eq:pmb} is a valid upper bound for $(\q', S)$ (Section $5.1$ in~\cite{ngo-survey}). We argue $(a)$ first. We retain only the edges of the form $p_{i} \rightarrow p_{i + 1}$ in $E(\mathcal{G})$ for every $i \ge 2$ and remove the others. In particular, $E(\mathcal{G})$ is now a directed path from $p_{2}$ to $p_{|S|}$. We argue $(b)$ in the next paragraph.

We argue $(b)$ and $(2)$ together. Note that there are $|S|$ constraints to be satisfied one for each vertex in $S = \{p_i : i \in [|S|] \}$ -- $\sum_{e = (p_1, p_2)} \delta_{e | \emptyset} \ge 1$ for $p_1$ and $\sum_{e = (p_i, p_{i + 1})} \delta_{e | \emptyset}  + \delta_{p_{i} | p_{i + 1}} \ge 1$ for every $p_i: i \in [2, |S|]$. In our solution, we set $\delta_{e  = (p_1, p_2) | \emptyset} = 1$ and $\delta_{p_{i} | p_{i + 1}} = 1$ for every $p_i: \in [2, |S|]$. Note that this satisfies all the constraints defined above and we have that
\begin{align*}
\PMB(\q', S) & \le \delta_{e  = (p_1, p_2) | \emptyset}  \log(N)  + \sum_{i \in [2, s]} \delta_{p_{i} | p_{i + 1}} \log(\sqrt{N}) \\
& = \log(N)  + (s - 2) \log(\sqrt{N}) = \log(N) + \frac{s - 2}{2} \log(N) = \frac{s}{2} \log(N),
\end{align*}
completing the proof. 

\subsubsection{Generalization for other $E_{S}$} \label{sec:general-es}
Note that we made two assumptions on $E_{S}$ for proving an upper bound for $\PMB(q', S)$ -- $(1)$ $E_{S}$ is a path and $(2)$ $E_{S}$ is a single connected component. We remove these assumptions now.

We argue how to get rid of assumption $(1)$ first. We build $E(\mathcal{G}$ in the same way as in the previous section. We fix an edge $e = (u, v) \in E_{S}$, where $u, v \in S$. We now remove edges from $E(\mathcal{G})$ until becomes a forest with two directed trees rooted at $u$ and $v$ respectively. Such a construction is possible since for every $e = (u', v') \in E_{S}$, there are two edges of the form $u' \rightarrow v'$ and $v' \rightarrow u'$. We can now decompose these two trees into paths from roots to leaves. On each path $\{ e_1 = (p_1, p_2), e_2 = (p_2, p_3), \dots, e_{s'} = (p_{s' - 1}, p_{s'})\}$ where $p_i \in S$ for every $i \in [s']$ (note that $s' \le |S|$), we apply the same argument as in the previous Section~\ref{sec:light-case} to obtain an upper bound of $\frac{s'}{2} \log(N)$ on the path. Later, we can sum the bounds obtain from each path to get a bound of at most $\frac{|S|}{2} \log(N)$. 

Finally, we get rid of the assumption $(2)$. We decompose $E_{S}$ into connected components and for each connected component, we can apply the argument above and later sum the bounds obtained from each connected component. Since each connected component has at least two connected vertices, our arguments hold.

\subsection{$\PMB(\q, S)$ and $\AGM(G, S)$ for specific $G$ and $S$}
We provide a list of well-known results on $\AGM(G, S)$ for specific $G$ and $S = V$. All of them can be computed by invoking~\eqref{eq:ubs} and subsequently, using Lemma~\ref{dbpagmlemma}.

\begin{lemma} \label{agmcycle}
Given $\q = (G, \{R_e: e \in E\})$, where $G$ is a collection of vertex-disjoint cycles with $k$ vertices in total and the input relations have size at most $N$, we have
\[\max_{R_e: |R_e| \le N, e \in E} 2^{\PMB(\q, S)} = 2^{\AGM(G, S)} \le N^{\frac{k}{2}} \]
for $S = V$.
\end{lemma}

\begin{lemma} \label{agmstar}
Given $\q = (G, \{R_e: e \in E\})$, where $G$ is a collection of $t$ vertex-disjoint stars with $k$ vertices in total and the input relations have size at most $N$, we have
\[\max_{R_e: |R_e| \le N, e \in E} 2^{\PMB(\q, S)} = 2^{\AGM(G, S)} \le N^{k - t}\]
for $S = V$.
\end{lemma}

\begin{lemma} \label{agmevenlengthpath}
Given $\q = (G, \{R_e: e \in E\})$, where $G$ is a collection of vertex-disjoint even length paths with $k$ vertices in total and the input relations have size at most $N$, we have
\[\max_{R_e: |R_e| \le N, e \in E} 2^{\PMB(\q, S)} = 2^{\AGM(G, S)} \le N^{\frac{k}{2}}\]
for $S = V$.
\end{lemma}

\begin{lemma} \label{agmsingleton}
Given $\q = (G, \{R_e: e \in E\})$, where $G$ is a collection of \textit{singletons} with $k$ vertices in total and the input relations have size at most $N$, we have
\[\max_{R_e: |R_e| \le N, e \in E} 2^{\PMB(\q, S)} = 2^{\AGM(G, S)} \le N^{k}\]
for $S = V$.
\end{lemma}

\input{g_decomp}

\section{Missing Cases in Table~\ref{tab:results}} \label{sec:arity2lemmas}
We start by restating Table~\ref{tab:results} below. Recall that we proved only the case corresponding to only rows $3$ and $4$ in the main paper. We prove all the remaining results here.
\begin{table*}[thbp!]
	{\small
		\centering
		{
			\hspace{0.5cm} {
				\begin{tabular}{| c | c | c |}
					\hline
					$n - \Delta + 1$ & $G$ & Is AGM tight?\\
					\hline
					\rowcolor{red!30} $= 1$ & - & Yes \\
					\hline
					\rowcolor{red!30} $\in [2, |M|]$, even &  -  & Yes \\
					\hline
					\rowcolor{green!50} $\in [3,|M| - 1]$, odd & not \textit{disedge} & No \\
					\hline
					\rowcolor{red!30} $\in [3,|M| - 1]$, odd & \textit{disedge} & Yes \\
					\hline
					\rowcolor{red!30} $\in [|M| + 1, |M| + n_I]$ & -  & Yes \\
					\hline
					\rowcolor{red!30}  $\in [|M| + n_I + 1, n]$ & - & Yes \\
					\hline
			\end{tabular}}
			\caption{The first column denotes the value of $n - \Delta + 1$ and the second column denotes constraints on $G$. `-' denotes no restrictions on $G$. The final column shows if the $\AGM$ bound is tight i.e., it is tight if $\max_{R_e, |R_e| \le N, e \in E}\PMB(\q = (G, \{R_e: e \in E\}), S) = \AGM(G, S)$ and not tight otherwise. The resulst in rows $3$ and $4$ were proved in the main paper and we prove the remaining results here.} \label{tab:results-app}
		}
	}
\end{table*}
We would like to note here that the $\AGM$ bound is tight for all the below arguments.
\subsection{Row $1$ of Table~\ref{tab:results-app}} \label{app:remaining-rows}
\begin{lemma} \label{lemma1}
Consider the case from row $1$ in Table~\ref{tab:results}. $G$ is any graph and $s = n - \Delta + 1 = 1$. Then, there exists a subset $S \subseteq V$ with size $1$ such that
\begin{equation} \label{eq:lemma1}
\joinUB(G, N, 1) \le  N \le 2 \cdot \cover(G, N, n).
\end{equation}
\end{lemma}
\begin{proof}
Our proof proceeds in two steps. First, we argue
\begin{equation} \label{eq:lemma1-1}
\joinUB(G, N, 1) \le  N 
\end{equation}
followed by
\begin{equation} \label{eq:lemma1-2}
\cover(G, N, n) \ge  \frac{N}{2}.
\end{equation}
Note that these two inequalities would immediately imply~\eqref{eq:lemma1}. We begin with proving~\eqref{eq:lemma1-1}. In particular, if we prove
\begin{equation} \label{eq:lemma1-3}
\min_{S \subseteq V, |S| = 1}  \max_{R_e: |R_e| \le N, e \in E} 2^{\PMB(\q, S)} \le N,
\end{equation}
then~\eqref{eq:lemma1-1} follows from Theorem~\ref{upperboundtheorem}. If we can show that there exists a $S$ with $|S| = 1$ such that
\begin{equation} \label{eq:lemma1-4}
\PMB(\q, S) \le \log(N),
\end{equation}
then that would imply $\min_{S \subseteq V, |S| = 1}  \max_{R_e: |R_e| \le N, e \in E} 2^{\PMB(\q, S)} \le N$ using~\eqref{eq:lemma1-3}. Note that this implies $\joinUB(G, N, 1) \le  N$ from Theorem~\ref{upperboundtheorem}, proving~\eqref{eq:lemma1-1}.

We now prove~\eqref{eq:lemma1-4}. Towards this end, we pick our subset $S$. Pick any arbitrary vertex $v \in V$ and let $S=\{v\}$. We compute $\PMB(\q, S)$, which is upper bounded by $\log(N)$ with a feasible solution $\delta_{e' | \emptyset} = 1$ for any one relation $e' \in E$ incident on $h$ and $\delta_{e | \emptyset} = 0$ for all other edges in $E$. Recall that $N_{e' | \emptyset}, N_{e | \emptyset} \le N$ for every $e, e' \in N$ (since all relation sizes are upper bounded by $N$).

Next, we argue why this is a feasible solution. Since $S = \{h\}$, note that there is only one constraint to be satisfied -- $\sum_{e \in E: h \ni e} \delta_{e | \emptyset} \ge 1$. Let $e'$ be one such edge in $\{e \in E: h \ni e\}$. Hence, we can set $\delta_{e' | \emptyset} = 1$ and $\delta_{e | \emptyset} = 0$ for all $e \in E$ to obtain a feasible solution. This is sufficient to argue $\PMB(\q, S)$ is upper bounded by $\log(N)$.

To complete the proof, we prove~\eqref{eq:lemma1-2}. As shown in Section~\ref{sec:lb}, we only need to construct an instance $\q = (G, \{R_e: e \in E, |R_e| \le N\})$ and a join cover $\J$. Let $q$ be the largest power of $2$ such that $q\le N$. In particular, we have $q\ge N/2$. We now instantiate a $(n, N, \Delta)_{q}$ RS code $\mathcal{C}$ over $\mathbb{F}_{q}$. In particular, $\mathcal{C}$ is of the form $\{ (i)_{i \in [n]} : i \in \mathbb{F}_{q} \}$. For each edge $e\in E$, we set $R_e =\mathcal{C}_{e}$, where $\mathcal{C}_e = \{(i, i) | i \in \mathbb{F}_q\}$ implying that $|R_e| = |\mathcal{C}_e|= q \le N$, as desired. Note that $\mathcal{C} = J =\ \Join_{e \in E} R_e$. Finally, we have 
\begin{align*}
\cover(G, N, \Delta) \ge |\mathcal{C}| \ge \frac{N}{2}, 
\end{align*}
which proves~\eqref{eq:lemma1-2}, completing the proof.
\end{proof}

\subsection{Row $2$ of Table~\ref{tab:results-app}}
\begin{lemma} \label{lemma11}
Consider the case from row $2$ in Table~\ref{tab:results}. $G$ is any graph and $s = \s \in [2, |M|]$ is even. Then, there exists a subset $S \subseteq V$ of size $s$ such that
\begin{equation} \label{eq:lemma3}
\joinUB(G, N, s) \le N^{\frac{s}{2}} \le O_{n}(\cover(G, N, \Delta)).
\end{equation}
\end{lemma}
\begin{proof}
Our proof proceeds in two steps. First, we argue
\begin{equation} \label{eq:lemma3-1}
\joinUB(G, N, s) \le  N^{\frac{s}{2}} 
\end{equation}
followed by
\begin{equation} \label{eq:lemma3-2}
\cover(G, N, \Delta) \ge \Omega_{n}(N^{\frac{s}{2}}).
\end{equation}
Note that these two inequalities would immediately imply~\eqref{eq:lemma3}. We begin with proving~\eqref{eq:lemma3-1}. In particular, if we prove
\begin{equation} \label{eq:lemma3-3}
\min_{S \subseteq V, |S| = s}  \max_{R_e: |R_e| \le N, e \in E} 2^{\PMB(\q, S)} \le N^{\frac{s}{2}},
\end{equation}
then~\eqref{eq:lemma3-1} follows from Theorem~\ref{upperboundtheorem}. If we can show that there exists a $S$ of size $s$ such that 
\begin{equation} \label{eq:lemma3-4}
2^{\PMB(\q, S)} \le \frac{s}{2} \cdot \log(N),
\end{equation}
then using Definition~\ref{def:dbpgs}, we have $\min_{S \subseteq V, |S| = s}  \max_{R_e: |R_e| \le N, e \in E} 2^{\PMB(\q, S)} \le N^{\frac{s}{2}}$, which in turn implies $\joinUB(G, N, s) \le N^{\frac{s}{2}}$ from Theorem~\ref{upperboundtheorem}, proving~\eqref{eq:lemma3-1}. 

We now prove~\eqref{eq:lemma3-4}. Towards this end, we pick our subset $S$. Pick any $\frac{s}{2}$ edges from $M$ and for each edge, add its endpoints to $S$. The constraint $s \le |M|$ ensures the existence of one such choice of $S$. We compute $\PMB(\q, S)$, which is given by $\frac{s}{2} \cdot \log(N)$ corresponding to a feasible solution $\delta_{e | \emptyset} = 1$ for every $e \in M, e \subseteq S$ and $\delta_{e | \emptyset} = 0$ for all the other edges $e \in E \setminus \{e \in M, e \subseteq S \}$. 

We now argue why this is a feasible solution. Note that there are constraints for every $v \in S$ of the form $\sum_{e \in E: v \ni e} \delta_{e | \emptyset} \ge 1$. We set $\delta_{e | \emptyset} = 1$ for every $e \in M, e \subseteq S$. This satisfies all constraints above since for every $v \in S$, there is a unique edge $e \in M$ incident on $v$ (this follows from our picking algorithm where we pick only matching edges). For all the remaining edges, we can set $\delta_{e | \emptyset} = 0$ and we have a feasible solution.

To complete the proof, we prove~\eqref{eq:lemma3-2}. As shown in Section~\ref{sec:lb}, we only need to construct an instance $\q = (G, \{R_e: e \in E, |R_e| \le N\})$ and a join cover $\J$. Let $q$ be the largest power of $2$ such that $q\le \sqrt{N}$. In particular, we have $q \ge \frac{\sqrt{N}}{2}$. We now instantiate a $\left(n, q^{(n - \Delta + 1)}, \Delta\right)_{q}$ RS code $\mathcal{C}$ over $\mathbb{F}_{q}$. For each edge $e\in E$, we set $R_e =\mathcal{C}_{e}$, where $\mathcal{C}_e \subseteq \mathbb{F}_{q} \times \mathbb{F}_{q} \le N$ implying that $|R_e| = |\mathcal{C}_e| \le q^2 \le N$, as desired. Note that $\mathcal{C} = J =\ \Join_{e \in E} R_e$. Finally, we have $\cover(G, N, \Delta) \ge |\mathcal{C}| = q^{n - \Delta + 1} = \Omega_{n}\left(N^{\frac{\s}{2}} \right)$, which proves~\eqref{eq:lemma3-2}, completing the proof.
\end{proof}

\subsection{Row $5$ of Table~\ref{tab:results-app}} \label{sec:row-5}
\begin{lemma} \label{lemma12}
Consider the case from row $5$ in Table~\ref{tab:results}. $G$ is any graph and $s = \s \in [|M| + 1, |M| + n_{I}]$. Then, there exists a subset $S \subseteq V$ with size $s$ such that
\begin{equation} \label{eq:lemma5}
\joinUB(G, N, s) \le N^{\frac{s}{2}} \le O_{n}(\cover(G, N, \Delta)).
\end{equation}
\end{lemma}
\begin{proof}
Our proof proceeds in two steps. First, we argue
\begin{equation} \label{eq:lemma5-1}
\joinUB(G, N, s) \le  N^{\frac{s}{2}} 
\end{equation}
followed by
\begin{equation} \label{eq:lemma5-2}
\cover(G, N, \Delta) \ge \Omega_{n}(N^{\frac{s}{2}}).
\end{equation}
Note that these two inequalities would immediately imply~\eqref{eq:lemma5}. We begin with proving~\eqref{eq:lemma5-1}. In particular, if we prove
\begin{equation} \label{eq:lemma5-3}
\min_{S \subseteq V, |S| = s}  \max_{R_e: |R_e| \le N, e \in E} 2^{\PMB(\q, S)} \le N^{\frac{s}{2}},
\end{equation}
then~\eqref{eq:lemma5-1} follows from Theorem~\ref{upperboundtheorem}. If we can show that there exists a $S$ with size $s$ such that 
\begin{equation} \label{eq:lemma5-4}
2^{\PMB(\q, S)} \le \frac{s}{2} \cdot \log(N),
\end{equation}
then using Definition~\ref{def:dbpgs}, we have $\min_{S \subseteq V, |S| = s}  \max_{R_e: |R_e| \le N, e \in E} 2^{\PMB(\q, S)} \le N^{\frac{s}{2}}$, which in turn implies $\joinUB(G, N, s) \le N^{\frac{s}{2}}$
from Theorem~\ref{upperboundtheorem}, proving~\eqref{eq:lemma5-1}. By definition, we have $n_I \ge 0$. 
	
We now prove~\eqref{eq:lemma5-4}. Towards this end, we pick our subset $S$. Add $V(M)$ to $S$. For the remaining vertices, choose any subset $S' \subseteq (V(G_c) \setminus V(M))$ such that $|S'| = s - |M|$.  The constraint $|M| + 1 \le s \le |M| + n_I$ ensures the existence of one such choice of $S$. 

Using Corollary~\ref{cor:matching-half}, it follows that $M$ can be decomposed as $M_{c} \cup M_{s}$, where $M_{c}$ is a maximum matching on $G_{c}$ and $M_{s}$ is a maximum matching on $G_{s}$. We now consider two vertex-disjoint subgraphs -- $(1)$ the subgraph $G'_{c} = (V_{c} = \{S \cap V(G_c)\} , E_{c} = M_{c} \cup \{e \in E(G_c): e \subseteq S\})$ and $(2)$ the subgraph $G'_{s} = (V_{s} = \{S \cap  V(G_s)\}, E_{s} = \{e \in E: e \subseteq V_{s}\})$. It follows that $|V_{c}| + |V_{s}| = s$. We have two query instances $\q_{c} = (G'_{c}, \{R_e: e \in E_{c} \})$ and $\q_{s} = (G'_{s}, \{R_e: e \in E_{s} \})$. Since $V_c \cap V_s \neq \emptyset$, we have (by Lemma~\ref{dbpunionlemma})
\[\PMB(\q, S) \le \PMB(\q_{c}, V_c) + \PMB(\q_{s}, V_{s}). \]

We compute $\PMB(\q_{c}, V_c)$ first. By Corollary~\ref{cor:matching-half}, there exists an optimal dual solution for $\q_{c}$ such that $y_{v} = \frac{1}{2}$ for every $v \in V_{c}$. Since it is an optimal dual solution, it follows that $\AGM(G_{c}, V_{c}) \le \frac{|V_c|}{2} \log(N)$. Using Theorem~\ref{dbpagmlemma}, we have 
\[\PMB(q_{c}, V_{c}) \le \AGM(G_{c}, V_{c}) \le \frac{|V_c|}{2} \cdot \log(N). \] 

We compute $\PMB(\q_{s}, S)$, which is upper bounded by $\frac{|V_s|}{2} \cdot \log(N)$ corresponding to a feasible solution $\delta_{e | \emptyset} = 1$ for every $e \in M \cap  E_{s}$ and $\delta_{e | \emptyset} = 0$ for all the other edges $E_{s} \setminus M$.  

We now argue why this is a feasible solution. Note that there are constraints for every $v \in V_s$ of the form $\sum_{e \in E_{s}, v \ni e} \delta_{e | \emptyset} \ge 1$. We set $\delta_{e | \emptyset} = 1$ for every $e \in M \cap E_{s}$. This satisfies all the constraints above since for every $v \in S$, there is a unique edge $e \in M \cap E_{s}$ incident on $v$ (this follows from our picking algorithm where we pick only matching edges in $V_{s}$). For all the remaining edges $E_{s} \setminus M$, we can set $\delta_{e | \emptyset} = 0$ and the solution still remains feasible.

Thus, we have
\[\PMB(\q, S) \le \PMB(\q_{c}, V_c) + \PMB(\q_{s}, V_{s}) \le \frac{|V_c|}{2} \cdot \log(N) + \frac{|V_s|}{2} \cdot \log(N) = \frac{s}{2} \cdot \log(N), \]
, which proves~\eqref{eq:lemma5-4}, as desired.

To complete the proof, we prove~\eqref{eq:lemma5-2}. We use the same instance and arguments as in the proof of  Lemma~\ref{lemma11} -- using a $\left(n, q^{(n - \Delta + 1)}, \Delta\right)_{q}$ RS code $\mathcal{C}$ over $\mathbb{F}_{q}$ with $q \ge \frac{\sqrt{N}}{2}$. This completes the proof.
\end{proof}

\subsection{Row $6$ of Table~\ref{tab:results-app}} \label{sec:row-6}
\begin{lemma} \label{lemma7}
Consider the case from row $6$ in Table~\ref{tab:results}. $G$ is any graph and $s = \s  \in [|M| + n_I + 1, n]$. Then, there exists a subset $S \subseteq V$ with size $s$ such that
\begin{equation} \label{eq:lemma6}
\joinUB(G, N, s) \le N^{s - \frac{|M| + n_I}{2}} \le O_{n}(\cover(G, N, \Delta)).
\end{equation}
\end{lemma}
\begin{proof}
Our proof proceeds in two steps. First, we argue
\begin{equation} \label{eq:lemma6-1}
\joinUB(G, N, s) \le  N^{s - \frac{|M| + n_I}{2}} 
\end{equation}
followed by
\begin{equation} \label{eq:lemma6-2}
\cover(G, N, \Delta) \ge \Omega_{n}\left(N^{s - \frac{|M| + n_I}{2}} \right).
\end{equation}
Note that these two inequalities would immediately imply~\eqref{eq:lemma6}. We begin with proving~\eqref{eq:lemma6-1}. In particular, if we prove
\begin{equation} \label{eq:lemma6-3}
\min_{S \subseteq V, |S| = s}  \max_{R_e: |R_e| \le N, e \in E} 2^{\PMB(\q, S)} \le N^{s - \frac{|M| + n_I}{2}},
\end{equation}
then~\eqref{eq:lemma6-1} follows from Theorem~\ref{upperboundtheorem}. If we can show that there exists a $S$ with size $s$ such that 
\begin{equation} \label{eq:lemma6-4}
2^{\PMB(\q, S)} \le \left(s - \frac{|M| + n_I}{2} \right) \cdot \log(N),
\end{equation}
then we have $\min_{S \subseteq V, |S| = s}  \max_{R_e: |R_e| \le N, e \in E} \tilde{e}^{\widehat{\DBP}(\q, S)} \le N^{s - \frac{|M| + n_I}{2}}$ using~\eqref{eq:lemma6-3}, which in turn implies $\joinUB(G, N, s) \le N^{s - \frac{|M| + n_I}{2}}$ from Theorem~\ref{upperboundtheorem}, proving~\eqref{eq:lemma6-1}.
	
We now prove~\eqref{eq:lemma6-4}. Towards this end, we pick our subset $S$. We start by picking $S$ - we first add $V(M)$ to $S$, followed by vertices from $V(G_c) \setminus V(M)$. For the remaining picks, we can choose any remaining vertex from $V \setminus V(G_c) \setminus V(M)$. The constraint $s > |M| + n_I$ ensures the existence of one such choice of $S$.

Using Corollary~\ref{cor:matching-half}, it follows that $M$ can be decomposed as $M_{c} \cup M_{s}$, where $M_{c}$ is a maximum matching on $G_{c}$ and $M_{s}$ is a maximum matching on $G_{s}$. We now consider two vertex-disjoint subgraphs -- $(1)$ the subgraph $G_{c} = (V(G_c), E(G_c))$ and $(2)$ the subgraph $G'_{s} = (V_{s} = \{S \cap  (V(G_s) \cup V(G_t))\}, E_{s} = \{e \in E: e \subseteq (V_{s} \cup V(G_t))\})$. It follows that $|V(G_c)| + |V_{s}| = \s$. We have two query instances $\q_{c} = (G_{c}, \{R_e: e \in E(G_c) \})$ and $\q_{s} = (G'_{s}, \{R_e: e \in E_{s} \})$. Since $V_c \cap V_s \neq \emptyset$, we have (from Lemma~\ref{dbpunionlemma})
\[\PMB(\q, S) \le \PMB(\q_{c}, V(G_c)) + \PMB(\q_{s}, V_{s}). \]

We compute $\PMB(\q_{c}, V(G_c))$ first. By Corollary~\ref{cor:g-decomp}, there exists an optimal dual solution for $\q_{c}$, where $y_{v} = \frac{1}{2}$ for every $v \in V(G_c)$. Since it is an optimal dual solution, it follows that $\AGM(G_{c}, V(G_c)) \le \frac{|V(G_c)|}{2} \log(N)$. Using Theorem~\ref{dbpagmlemma}, we have 
\[\PMB(q_{c}, V(G_c)) \le \AGM(G_{c}, V(G_c)) \le \frac{|V(G_c)|}{2} \cdot \log(N). \] 

We now compute $\PMB(\q_{s}, V_{s})$. By Corollary~\ref{cor:g-decomp}, there exists an optimal dual solution for $\q_{s}$ where for every $e = (u, v) \in M_{s}$, either $y_{u} = 0, y_{v} = 1$ (or) $y_{u} = 1, y_{v} = 0$. For all the remaining vertices in $V_{s} \setminus V(M_{s})$, we have $y_{v} = 1$. Since it is an optimal dual solution, it follows that $\AGM(G_{s}, V_{s}) \le \left(\frac{|V_s|}{2} + s - |V(G_c)| - |V_s| \right)\log(N)$. Using Theorem~\ref{dbpagmlemma}, we have 
\[\PMB(q_{s}, V_{s}) \le \AGM(G_{s}, V_{s}) \le \left(\frac{|V_s|}{2} + s - |V(G_c)| - |V_s| \right)\log(N).\]

Thus, we have
\begin{align*}
\PMB(\q, S) & \le \PMB(\q_{c}, V_c) + \PMB(\q_{s}, V_{s}) \\
 & \le  \frac{|V(G_c)|}{2} \cdot \log(N) + \left(\frac{|V_s|}{2} + s - |V(G_c)| - |V_s| \right)\log(N) \\ 
& = \left(s - \frac{|V(G_c)| + |V_s|}{2} \right) \cdot \log(N) = \left(s - \frac{|M| + n_{I}}{2} \right) \log(N), \\.
\end{align*}
where the last equality follows from the fact that $|V(G_c)| = |V(M_c)| + n_I$, $|V_s| = |V(M_s)|$ and $V(M) = |V(M_c)| + |V(M_s)| = |M|$.

To complete the proof, we prove~\eqref{eq:lemma6-2}. We only need to construct an instance $\q = (G, \{R_e: e \in E, |R_e| \le N\})$ and a join cover $\J$. We construct a CRT code (Definition~\ref{defn:crt}) from the AGM optimal dual solution~\eqref{eq:fecdual} as follows. Using Theorem~\ref{optimal theorem}, there always exists an optimal dual solution $\mathbf{y} = (y_{v})_{v \in V}$ such that $y_{v} \in \{0, \frac{1}{2}, 1\}$ for every $v \in V$ (i.e., a half integral optimal dual solution). We start with the assignment of $q_{v}$s for every $v \in V$. For every vertex having $y_{v} = 0$, we assign a unique prime of order $O(1)$ for $q_{v}$ and for every vertex having $y_{v} = 1$, we assign a prime of order $\Theta{(N)}$ for $q_{v}$. Note that all the remaining vertices have $y_{v} = \frac{1}{2}$ and we assign a unique prime of order $\Theta{(\sqrt{N})}$ for $q_{v}$. We define $\mathbf{q} = (q_{v})_{v \in V}$.

We instantiate a $\left(n,  M, n - \Delta + 1\right)_{\mathbf{q}}$-CRT code $\mathcal{C}$ with $M = \Omega_{n}\left(N^{n - \Delta + 1 - \frac{|M| + n_I}{2}} \right)$. For each edge $e = (u, v) \in E$, the optimal dual solution ensures that $y_{u} + y_{v} \le 1$. Note that this implies $|\mathcal{C}_e| \le O(N)$
for every $e \in E$. We set $R_{e} = \mathcal{C}_{e}$ for every $e \in E$ and it follows that $|R_{e}| \le N$. 
Further, we have $\mathcal{C} = \J =\ \Join_{e \in E}  R_e$. Since $M = \Omega_{n}\left(N^{n - \Delta + 1 - \frac{|M| + n_I}{2}} \right)$, we have $\cover(G, N, \Delta) \ge M$, proving~\eqref{eq:lemma6-2} as desired. 

To complete the proof, we still need to argue $M = \Omega_{n}\left(N^{n - \Delta + 1 - \frac{|M| + n_I}{2}} \right)$. By definition, the size $M$ is the product of the smallest $n - \Delta + 1$ $q_v$s. We now use  Corollary~\ref{cor:g-decomp} to identify these smallest $q_{v}$s. We start with the $\frac{|V(M) \setminus V(G_c)|}{2}$ vertices with $y_{v} = 0$. Then, we have $|V(G_c)|$ vertices with $y_{v} = \frac{1}{2}$ and finally, there are $\s - |V(G_c)| - \frac{|V(M) \setminus V(G_c)|}{2}$ vertices with $y_{v} = 1$. This gives us \[M \ge \Omega_{n} \left(O(1)^{\frac{|V(M) \setminus V(G_c)|}{2}} \cdot N^{\frac{|V(G_c)|}{2}} \cdot N^{\s - |V(G_c)| - \frac{|V(M) \setminus V(G_c)|}{2}} \right).\]
Finally, we have
\[M \ge \Omega_{n} \left(N^{\s - \frac{|V(G_c)|}{2} - \frac{|V(M) \setminus V(G_c)|}{2}} \right) = \Omega_{n} \left(N^{\s - \frac{|M| + n_{I}}{2}}  \right).\]
\end{proof}

\subsection{Summary of Results}
Note that for all the cases in Table~\ref{tab:results-app}, we always picked a $S \subseteq V$ in $\poly(|G|)$ time such that $\joinUB(G, N,\Delta) \le O_{n}(\cover(G, N, \Delta))$. Since these lemmas together cover all simple graphs $G$ and $\Delta \in [1, n]$, the below result follows.
\begin{corollary} \label{pickslemma}
For all simple graphs $G$ and every $\Delta \in [1, n]$, there exists a picking algorithm $A$ that picks a subset $S \subseteq V: |S| = n - \Delta + 1$ in $\poly(|G|)$ time, returning a $\joinUB(G, N, |S|)$ that can potentially be off from the optimal $\cover(G, N, \Delta)$ bound by a factor of $2^{n}$ (note that in the database setting $n \ll N$ and typically treated as a constant).
\end{corollary}
The above corollary implies the following.
\begin{corollary}
For all simple graphs $G$ and every $\Delta \in [1, n]$, the bounds $\joinUB(G, N,\Delta)$ and $\cover(G, N, \Delta)$ are tight within a factor of $2^{n}$, where $n$ is a constant.
\end{corollary}
We can convert our results in Corollary~\ref{pickslemma} to algorithms for computing a join cover, which we present in Section~\ref{sec:algos}.

%% file: g_decomp.tex
\subsection{Decomposition of $G$} \label{sec:g-decomp}
Recall that our goal here is to decompose any graph $G$ into three induced subgraphs $G_c$, $G_s$ and $G_t$ such that $|V(G)| = |V(G_s)| + |V(G_t)| + |V(G_c)| = n$. Here, $G_s$ is a collection of stars of size greater than $1$, $G_t$ is a collection of singletons and $G_c = (V(G_c) = V(G) \setminus \{ V(G_s) \cup V(G_t)\}, E(G_c) = \{ e \in E: e \subseteq V(G_c) \})$. $M$ denotes a maximum matching on $G$ and $|M|$ denote the total number of vertices in $M$ ($|M|$ is even). We define $n_{I} = |V(M) \setminus V(G_c)|$. We present an example of our decomposition here.

\begin{example} \label{ex:g-decomp-ex}
Consider a graph $G= (V = \{1, 2, 3, 4, 5, 6\}, E = \{(1, 2), (2, 3), (1, 3) (1, 4), (4, 5), (6, 6)\})$ from Figure~\ref{fig:example2}. In this case, we decompose $G$ as follows -- $G_{c} = (V(G_c) = \{1, 2, 3\}, E(G_c) = \{(1, 2), (2, 3), (1, 3)\}$, $G_{s} = (V(G_s) = \{4, 5\}, E(G_s) = \{(4, 5)\})$ and $G_{t} = (V(G_t) = \{6\}, E(G_t) = \{(6, 6)\})$. Further, we have $M = \{(1, 2), (4, 5)\}$, $|M| = 4$ and $n_I = |\{(4, 5)\}| = 2$.
\end{example}

\begin{figure}[!htbp]
\centering
\begin{tikzpicture}[myn/.style={circle, draw,inner sep=0.1cm,outer sep=2pt}]
\node[myn] (A) at (0, 0) {$1$};
\node[myn] (B) at (2, 0) {$2$};
\node[myn] (C) at (0, -2) {$3$};
\node[myn] (D) at (-2, 0) {$4$};
\node[myn] (E) at (-2, -2) {$5$};
\node[myn] (F) at (-4, -1) {$6$};
\draw[line width = 2mm, magenta] (A) -- (B) node [magenta, above left = 7.5pt, sloped = 90] (TextNode) {};
\draw[line width = 2mm, orange] (B) -- (C) node [orange, above right = 20pt, sloped = 90] (TextNode) {};
\draw[line width = 2mm, green] (A) -- (C) node [green,  above left = 9pt, sloped = 90] (TextNode) {};
\draw[line width = 2mm, blue] (A) -- (D) node [blue, above right = 7.5pt, sloped = 90] (TextNode) {};
\draw[line width = 2mm, cyan] (D) -- (E) node [cyan, above left = 9pt, sloped = 90] (TextNode) {};
%\fill (-4,-1) circle[radius=1pt];
\draw (F)  to[red,in=50,out=130,loop] (F);
%\path  (F)   edge[my loop] node [above]  {$e_{1}$} (F);
%\draw[line width = 2mm,  pink] (F) -- (F) node [pink, below = 3pt, sloped = 90] (TextNode) {};
\end{tikzpicture}
\caption{$G$ is defined as in Example~\ref{ex:g-decomp-ex} and has a cycle $(1, 2, 3)$, a star $(4, 5)$ and a singleton $(6)$.} \label{fig:example2}
\end{figure}
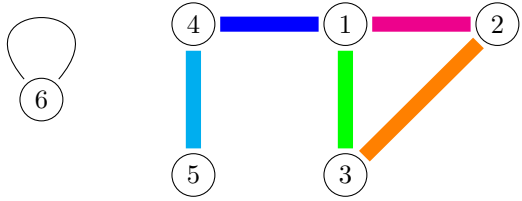

We start by restating the $\AGM(G, S)$ from~\eqref{eq:ubs}. 
\begin{equation} \label{eq:fecprimal}
\begin{alignedat}{1}
\min{\sum_{e}x_e} \\
\text{s.t } \sum_{v \ni e}x_e \ge 1, \text{for all } v \in S \\
x_e  \ge 0,  e \in E. \\
\end{alignedat}
\end{equation}
We now establish structural properties for the dual of \eqref{eq:fecprimal}, which will come handy while picking $S$ and start by stating the dual.
\begin{equation} \label{eq:fecdual}
\begin{alignedat}{2} 
\max \sum_{v}y_v \\	
\text{s.t. } \sum_{v \in e}y_{v} \le 1, \text{ for all } e\in E \\
y_v \ge 0, v \in S. \\
\end{alignedat}
\end{equation}
Unless stated otherwise, $S = V$. We now state the properties of this decomposition that we will be using in the proofs of our Lemmas in Table~\ref{tab:results}. Our goal in this section will be to prove these two corollaries. The following result will be used in our upper bound argument in Appendix~\ref{sec:row-5}.
\begin{corollary} \label{cor:matching-half}
The maximum matching $M$ of $G$ can be decomposed as $M = M_{c} \cup M_{s}$, where $M_{c}$ is a maximum matching on $G_c$ and $M_{s}$ is a maximum matching on $G_{s}$. For every subset of vertices $S \subseteq I$ (where $I$ is a maximum independent set of $G_{c}$) the graph $G_{ms}$ induced by $M_{c} \cup S$ has an optimal dual solution to LP~\eqref{eq:fecdual} on $G_{ms}$ with all vertices getting a value of $\frac{1}{2}$.
\end{corollary}
Note that in Example~\ref{ex:g-decomp-ex}, we have $M_{c} = \{ 1, 2\}$, $M_{s} = \{4, 5\}$ and $I = \{3\}$. When $S = \emptyset$, $G_{ms} = M_{c}$ has an optimal dual solution with all vertices getting value of $\frac{1}{2}$. Finally, when $S = I$, we have $G_{ms} = G_{c}$ and has an optimal dual solution with all vertices getting a value of $\frac{1}{2}$. The following result will be used in our lower bound argument in Appendix~\ref{sec:row-6}.
\begin{corollary} \label{cor:g-decomp}
There exists an half-integral optimal dual solution to the LP~\eqref{eq:fecdual} for any graph $G$, where each vertex in $G_c$ gets a value of $\frac{1}{2}$, each vertex in $G_t \cup \{V(G_s) \setminus V(M)\}$ gets a value of $1$ and each edge $e = (u, v) \in M \setminus G_c$ has $(u = 0, v = 1)$ or $(u = 1, v = 0)$.
\end{corollary}
Note that in Example~\ref{ex:g-decomp-ex}, we have that $G$ has an optimal dual solution, where each vertex in $G_{c}$ gets a value of $\frac{1}{2}$, the vertex $6$ in $G_{t}$ has a value of $1$ and the vertex $4$ gets value of $0$ and vertex $5$ gets a value of $1$ in $G_{s}$.  

As stated earlier, we will prove Corollaries~\ref{cor:g-decomp} and~\ref{cor:matching-half} in this section and we will proceed as follows. First, we claim the following:
\begin{theorem} \label{optimal theorem}
There exists a half-integral $\mathbf{y} =  (y_{v})_{v \in V}$ for LP~\eqref{eq:fecdual} such that $\mathbf{y}$ is an optimal BFDS.
\end{theorem}
We will use the following result to prove the above theorem.
\begin{theorem}\label{feasible theorem}
Every basic feasible dual solution (BFDS) $\mathbf{y} = (y_{v})_{v \in V}$ is half-integral i.e., $y_{v} \in \{0, \frac{1}{2}, 1\}$ for every $v \in V$.
\end{theorem}
Assuming the above theorem is true, Theorem~\ref{optimal theorem} follows directly since there always exists an optimal solution for any linear program that is BFDS. We will prove Theorem~\ref{feasible theorem} at the end of this section and going forward, we will be working with an half-integral optimal dual solution $\mathbf{y}$ of LP~\eqref{eq:fecdual} for $G$ to extract some structural properties.

We will now state two related results, which when combined proves Corollary~\ref{cor:matching-half}. 
\begin{lemma} \label{allhalflemmafinal}
For every subset of vertices $S \subseteq I$, the graph $G_{ms}$ induced by $M \cup S$ has an optimal dual solution to LP~\eqref{eq:fecdual} on $G_{ms}$ with all vertices getting a value of $1/2$.
\end{lemma}
\begin{lemma} \label{maximummatching}
If $M_c$ is a maximum matching on $G_c$ and $M_s$ is a maximum matching on $G_s$, $M = M_c \cup M_s$ is a maximum matching on $G$.
\end{lemma}
Finally, to prove Corollary~\ref{cor:g-decomp}, we will need the following result as well.
\begin{lemma} \label{perfectmatching}
In $G_s$, every vertex $v$ having a value $y_v = 0$ can be uniquely matched to a vertex $u$ having a value $y_u = 1$, directly connected to it.
\end{lemma}
We are now ready to prove Corollary~\ref{cor:g-decomp}.
\begin{proof} [Proof of Corollary~\ref{cor:g-decomp}]
The proof follows from the results we proved earlier in this section. From Lemma~\ref{allhalflemmafinal}, it follows that each vertex in $G_c$ gets a value of $\frac{1}{2}$. Note that by half-integrality, all the other vertices in $V(G) \setminus V(G_c)$ get either a $0$ or $1$. Since each vertex in $G_t$ is a singleton, each vertex gets a value $1$. Using Lemma~\ref{perfectmatching}, we have in $G_s$ that every vertex $v$ having a value $y_v = 0$ can be uniquely matched to a vertex $u$ having a value $y_u = 1$. Since $M = M_c \cup M_s$, where $M_c$ is a maximum matching on $G_c$ and $M_s$ is a maximum matching on $G_s$ by Lemma~\ref{maximummatching}, each $e = (u, v) \in M \setminus G_c$ has $(u = 0, v = 1)$ or $(u = 1, v = 0)$. Finally, note that all other vertices in $\{V(G_s) \setminus V(M)\}$ can be assigned a value of $1$. 
\end{proof}
Our remaining task in this section is to prove the above results (that we assumed to be true so far) one-by-one. 

\subsubsection{Proof of Lemma~\ref{allhalflemmafinal}}
We split this proof into two cases, when $S = I$ and $S \subset I$. We first state and argue the former,  which is fairly straightforward. The latter, while fairly intuitive, needs a more technically involved argument and to the best of our knowledge, was not explicitly before.
\begin{lemma} \label{allhalfoptimal}
Let $\mathbf{y}$ be an half-integral optimal basic feasible dual solution to LP~\eqref{eq:fecdual}. Consider the subgraph $G_c$ of $G$ in which every vertex $u \in V(G_c)$ has $y_{u} = \frac{1}{2}$. Then, there exists an optimal dual solution for $G_c$, where all vertices $u \in V(G_c)$ get a value of $\frac{1}{2}$. 
\end{lemma}
\begin{proof} [Proof of Lemma~\ref{allhalfoptimal}]
The proof is by contradiction. We start by assuming that there exists a better optimal dual solution $\mathbf{y'}$ for $G_c$ such that $\sum_{i = 1}^{|V(G_c)|}y'_i > \frac{|V(G_c)|}{2}$. The constraints of~\eqref{eq:fecdual} imply that any vertex $u$ such that $y_{u} = 1/2$ can only be directly connected to vertices $w$ such that $y_{w} \in \{0, 1/2\}$. If $y_{w} = 1/2$, $w$ will already be a part of $G_c$. Since $\mathbf{y}$ is half-integral, every vertex $v \in V \setminus V(G_c)$ such that there exists at least one edge $e = (u, v) \in E \setminus E(G_c), u \in V(G_c)$ has $y_{v} = 0$. In particular, this implies that we can replace $\mathbf{y}_{V(G_c)}$ by $\mathbf{y'}$ (i.e., we can set $y_{v} = y'_{v}$ for every $v \in V(G_c)$), without violating any constraints. Note that this implies 
\[\sum_{i = 1}^{|V(G)|}y_i = \sum_{v \not \in V(G_c)} y_i + \sum_{v \in V(G_c)} y_i > \sum_{v \not \in V(G_c)} y_i + \frac{|V(G_c)|}{2}.\]
Thus, we can obtain a $\mathbf{y}$ with a strictly greater value, resulting in a contradiction $\mathbf{y}$ is an optimal dual solution for $G$.
\end{proof}
To obtain the structural properties of vertices $v$ having $y_v = 1/2$ in $\mathbf{y}$, we start by assuming that $G_c$ is a single connected component. Then, there always exists a maximum matching $M$ for $G_c$ that can be computed in $\poly(|G_c|)$ time such that $M$ is a single connected component. Since the vertices in the independent set of $G_c$ do not have any direct edges between them, it should be possible to traverse between any pair of vertices in $G_c$ using the graph $G_m$ induced by $M$. We define $I = V(G_c) \setminus V(G_m)$ as the independent set of vertices and $n_I = |V(I)|$.

\begin{proof} [Proof of Lemma~\ref{allhalflemmafinal}]
When $S = I$, we already know by Lemma~\ref{allhalfoptimal} that there exists an optimal dual solution where all vertices get a value of $1/2$. For the rest of this proof, we assume that $S \subset I$. 
	
We start by assuming that $G_{ms}$ has an optimal dual solution $\mathbf{y}_{ms}$ with a value greater than $(|M| + |S|)/2$. Among all possible optimal dual solutions, $\mathbf{y}_{ms}$ is picked as the one having the maximum number of vertices set to a value $1/2$. Note that $\mathbf{y}_{ms}$ can be extended to all of $G$ in at least one way, without violating any of its constraints by retaining the half-integrality of the solution. One such trivial extension is when all vertices in $V(G) \setminus V(M \cup S)$ get a value of $0$. Among all such feasible half-integral extensions, we pick the one that has the maximum value. Let's call the finally chosen solution $\mathbf{y}$. We now show that there will be at least one vertex $v \in I \setminus S$ such that $y_{v} = 0$.
	
The proof is by contradiction. We start by assuming there exists no vertex $v \in I \setminus S$ such that $y_{v} = 0$. Note that this implies for every vertex $v \in I \setminus S$, $y_{v} \in \{1/2, 1\}$. We already know that the optimal dual solution for $M \cup S$ is greater than $(|M| + |S|)/2$. This, when combined with the fact that for every $v \in I \setminus S : y_{v} \in \{1/2, 1\}$, implies that the optimal dual solution for $M \cup I$ (which is all of $V(G_c)$) is greater than $(|M| + |I|)/2$. Note that we have ended up contradicting Lemma~\ref{allhalfoptimal}. Thus, we have shown that there should exist at least one vertex $v \in I \setminus S$ such that $y_{v} = 0$. We claim the following.
	
\begin{claim} \label{claim1}
If $v \in I \setminus S$ and $y_{v} = 0$, there exists a subset of matching edges in $M$ with the {\em left} side being $V_{1}$ and {\em right} side being $V_{0}$ such that:
\begin{itemize}
\item{For every $u \in V_{1}$, $y_{u}=1$.}
\item{For every $w \in V_{0}$, $y_{w}=0$. }
\item{All neighbors $x$ of $v$ such that $y_{x} = 1$ are contained in $V_{1}$.}
\item{All neighbors $x$ of every $w \in V_{0}$ such that $y_{x} = 1$ are contained in $V_{1}$.}
\end{itemize}
\end{claim}

Assuming that the above claim is true, we can set all the nodes in $V_{1}$ and $V_{0}$ to a value $1/2$ without violating any constraints. This is true since every neighbor $u$ of every node $w \in V_{0}$ such that $y_{u} = 1$ is already contained in $V_{1}$ and if the value of any node $u \in V_{1}$ is changed from $1$ to $1/2$, it does not affect its neighbors (as they are already set to $0$ or $1/2$). Note that our new solution $\mathbf{y'}_{ms}$ is feasible and has a strictly greater number of vertices with a value $1/2$ than the original optimal dual solution $\mathbf{y}$ and $\sum_{i=1}^{|M \cup S|} y'_i = \sum_{i = 1}^{|M \cup S|}y_{{ms}_i}$. This is a contradiction to our choice of $\mathbf{y}_{ms}$. This implies no such vertex $v \in I \setminus S : y_{v} = 0$ can exist, which in turn proves that there can never be a better optimal dual solution to $M \cup S$ than all vertices getting a value of $1/2$.
	
Finally, we prove Claim~\ref{claim1} algorithmically. Initially, we add all neighbors $u$ of $v$ such that $y_{u} = 1$, to $V_{1}$. Note that given $V_{1}$, the corresponding $V_{0}$ is automatically defined using $M$. We now make a crucial observation that any neighbor $u$ of $V_{0}$ such that $y_{u} = 1$ has to be a node in $M$. To see why this is true, let's assume for now that $u \in I$. In that case, we would end up with an augmenting path between $v$ and $u$ i.e., a path between two unmatched vertices ($u$ and $v$) in which edges belong alternatively to the matching and not to the matching. A matching is considered maximal if and only if it does not have any augmenting path. Since a maximum matching is also a maximal one, $u$'s presence in $I$ would contradict the fact that $M$ is a maximum matching. As a result, we can add all such $u$'s to $V_{1}$ (if they are not already present) and keep repeating this process. Our algorithm will always terminate as it will visit each node in $G_c$ at most once, in the worst case. In particular, we obtain the sets $V_{0}$ and $V_{1}$ with the properties we specified in the statement of Claim~\ref{claim1}.
\end{proof}

\subsubsection{Proofs of Lemma~\ref{maximummatching} and Lemma~\ref{perfectmatching}}
To prove Lemma~\ref{maximummatching}, we use Lemma~\ref{perfectmatching}, which we prove next.
\begin{proof} [Proof of Lemma~\ref{maximummatching}]
We start by assuming that there exists a better maximum matching $M'$ on $G$. Note that the additional matching edges can exist in $M'$ only between vertices $u \in V(G_c)$ and $v \in V(G_s)$ such that $y_{u} = 1/2$ and $y_{v} = 0$. (This follows from the constraints of~\eqref{eq:fecdual} and the fact that $M = M_c \cup M_s$.) We have already shown in Lemma~\ref{perfectmatching} that for every $v \in V(G_s \cup G_t) : y_{v} = 0$, there exists a corresponding uniquely mapped vertex $w \in V(G_s \cup G_t) : y_{w} = 1$. As a result, we can {\em flip} each edge in $M' \setminus M$ to a corresponding edge in $M_s$. This concludes that $|M'| \le |M|$, resulting in a contradiction that $M'$ can never be a better matching than $M$.
\end{proof}
\begin{proof} [Proof of Lemma~\ref{perfectmatching}]
The proof is by contradiction. Towards this end, we start by assuming that such a unique matching does not exist. This can happen in two ways:
\begin{itemize}
\begin{item}
When a vertex $v$ such that $y_{v} = 0$ is not directly connected to any other vertex $u : y_u = 1$, we can always assign $y_v = 1/2$ without violating any constraint of~\eqref{eq:fecdual}. This results in a better dual solution $\mathbf{y'}$ (where all other vertices except $v$ have the same value as in $\mathbf{y}$), contradicting the optimality of $\mathbf{y}$.  
\end{item}
\begin{item}
There exists a set of vertices $U$ such that for every vertex $u \in U: y_{u} = 0$ and its corresponding $1$-neighbor set (i.e. the set of vertices $V$ where every vertex $v \in V : y_{v} = 1$ and $(u, v) \in E$) satisfies $|U| > |V|$. (Otherwise by Hall's Theorem, there exists a required matching.) Without loss of generality, we can remove all the edges $e = (u_1, u_2)$ such that $u_1 \in U$ and $u_2 \in U$. (Note that even if $y_{u_1} = 1/2$ and $y_{u_2} = 1/2$, these edges can still be removed as they do not violate any constraints. We will be using this observation later.) Further, there can be no edges $e = (v_1, v_2)$ such that $v_1 \in V$ and $v_2 \in V$ as that would violate the constraints of~\eqref{eq:fecdual}.  It follows that any vertex $v$ in $V$ can only be directly connected to a vertex $u$ in $U$ and for one such $e = (v, u)$, $y_u$ and $y_v$ can be set to $1/2$ without violating any constraints, provided the process is repeated for all such $e$s.
			
Since $|U| > |V|$, we can set all vertices in vertex-sets $U$ and $V$ to a value of $1/2$ without violating constraints as shown in the previous paragraph. This would result in a better solution $\mathbf{y'}$  such that $\sum_{v \in (U \cup V)} y'_v = (|U| + |V|)/2 > |V|$, contradicting the optimality of $y$ since we can replace $\mathbf{y}_{V(G_s)}$ by $\mathbf{y'}$ (i.e. replace the value $y_v$ of every vertex $v \in V(G_s)$ by $y'_v$), without violating any constraints. Note that $\sum_{v \in V(G)}y_v$ strictly increases resulting in a contradiction that it is an optimal dual solution for $G$ and establishes our claim that $|U| \le |V|$.
\end{item}
\end{itemize}
As a result, we have shown that every vertex $v : y_{v} = 0$ will have an unique matching candidate $u: y_{u} = 1$, completing our proof. 
\end{proof}

\subsubsection{Proof of Theorem~\ref{feasible theorem}}
We would like to note that this is a standard result. However, we prove it here for the sake of completeness.
\begin{proof} [Proof of Theorem~\ref{feasible theorem}]
The proof is by contradiction and uses the fact that a BFDS cannot be expressed as a convex combination two other feasible dual solutions. 
	
We start by assuming that there exists a BFDS $\mathbf{y'}$ for the LP~\eqref{eq:fecdual} for $G$ containing a value $y'_{i} \notin \{0, 1/2, 1\}$. It follows that $G$ can never have an edge $e \in E$ with endpoints $(i, j)$ such that $y'_i \ge 1/2$ and $y'_j > 1/2$ as it would violate the constraint $y'_{i} + y'_{j} \le 1$ on $e$, resulting in an infeasible solution.
	
We now pick a really small value $\epsilon > 0$ such that there are two solutions $\mathbf{y^+}$ and $\mathbf{y^-}$ to the LP, where:
	
\[
y^+_i = 
\begin{cases}
y'_i + \epsilon & \textit{if $y'_i \in (0, 1/2)$} \\
y'_i - \epsilon & \textit{if $y'_i \in (1/2, 1)$} \\
y'_i & otherwise
\end{cases}
\numberthis \label{eq:y1} \]
	
\[
y^-_i = 
\begin{cases}
y'_i - \epsilon & \textit{if $y'_i \in (0, 1/2)$} \\
y'_i + \epsilon & \textit{if $y'_i \in (1/2, 1)$} \\
y'_i & otherwise
\end{cases}.	
\numberthis \label{eq:y2}\]
	
It follows from the definitions of $\mathbf{y^+}$ and $\mathbf{y^-}$ that for each $i \in V$, $y'_i = (y^+_i + y^-_i)/2$. We now argue that $\mathbf{y^+}$ and $\mathbf{y^-}$ are both feasible solutions for the LP~\eqref{eq:fecdual}.	
	
\textbf{$\mathbf{y^+}$ is a feasible solution: } Recall that $\mathbf{y'}$ is a feasible solution to LP~\eqref{eq:fecdual}. Note that this implies for every $e = (i, j) \in V$, we have $y'_i + y'_j \le 1$. We now consider each case in~\eqref{eq:y1}.  When $y'_i \in (0, 1/2)$ and $y'_j \in (0, 1/2)$, we have $y^+_i = y'_i + \epsilon + y'_j + \epsilon = y'_i + y'_j + 2 \epsilon$ with $y'_i < 1/2, y'_j < 1/2$. If we choose $\epsilon$ to be small enough, we can ensure that $y^+_{i} \le 1$. When $y'_i \in (1/2, 1)$ and $y'_j \in (1/2, 1)$, we have $y^+_i = y'_i - \epsilon + y'_j - \epsilon = y'_{i} + y'_{j} - 2 \epsilon$. Note that such an edge cannot exist in $E$ by definition. Finally, we consider the case when $y'_i \in (0, 1/2)$ and $y'_j \in (1/2, 1)$ or vice-versa, we have $y^+_i = y'_i + \epsilon + y'_j -\epsilon = y'_i + y'_j $. Note that $y'_{i} + y'_{j} \le 1$ by definition of $\mathbf{y'}$. Thus, we have shown that $\mathbf{y^+}$ is a feasible solution to the LP~\eqref{eq:fecdual}. Note that we can use a similar argument to show that $\mathbf{y^-}$ is a feasible solution as well.
	
We have shown that $\mathbf{y'}$ can be expressed as a convex combination of two feasible dual solutions and as a result, it cannot be a BFDS. Therefore, a solution that is not half integral cannot be a BFDS.
\end{proof}

%% file: appendix_algorithms.tex
\section{Missing Details in Section~\ref{sec:algos}}
\subsection{Proof of Lemma~\ref{lemma5.2}} \label{sec:arity2algo}
\begin{proof}
Given an instance $\q$ where all the input $R_e$s have size at most $N$, $\calA$ tells us if $\BCQ_{\q}$ is $1$ or $0$ (i.e., $J = \Join_{e \in E} R_{e}$ is empty or not). For any graph $G$, by Corollary~\ref{pickslemma}, we can always pick a $S$ with size $n - \Delta  + 1$ in $O(\poly(|G|))$ time such that 
\[\joinUB(G, N, \s \le O_{n}(\cover(G, N, \Delta)).\]
Recall that 
\[\joinUB(G, N, \s) \le \min_{S: |S| = n - \Delta + 1} \max_{R_e: |R_e| \le N, e \in E} 2^{\PMB(\q, S)}\]
from Theorem~\ref{upperboundtheorem}. Further, from the proof of Theorem~\ref{lowerboundtheorem} we have that 
\begin{equation} \label{eq:dbp-cover}
\min_{S: |S| = n - \Delta + 1} \max_{R_e: |R_e| \le N, e \in E} 2^{\PMB(\q, S)} \le O_{n}(\cover(G, N, \Delta)).
\end{equation}
Using Algorithm $3$ in~\cite{ngo-survey} (which is a worst-case optimal join algorithm), we can obtain a candidate $J_{Q, S}$ such that $|J_{Q, S}| = \min_{S: |S| = n - \Delta + 1} 2^{\PMB(\q, S)}$ in time $O(|J_{Q, S}|)$ (ignoring polylog factors in $N$ and a factor of $2^{n}$). In particular, using~\eqref{eq:dbp-cover}, we have
\begin{equation} \label{eq:cover-bound}
\dfrac{|J_{Q, S}|}{\cover(G, N, \Delta)} \le 2^{n}.
\end{equation}

We describe our algorithm $\calB$ here. The input to $\calB$ is $P_{0} = J_{Q, S}$ and the expected output is a join cover $\calS$ such that $|\calS| = |J_{Q, S}|$. We now consider the remaining vertices in $v \in V(G) \setminus S$ one-by-one, computing $P_{1} \gets P_{0} \times \Dom(v)$ for the currently chosen vertex $v$. Then, we remove all tuples $\mathbf{p} \in P_{1}$ such that there exists no tuple $\mathbf{t} \in \J$ with $\pi_{S \cup \{v\}}\mathbf{t} = \mathbf{p}$ for some $\mathbf{p} \in P_1$. One obvious way to do this is to compute $\J$ beforehand and do this filtering. Naturally, we would like to do something better (i.e., without computing the entire $J$).

We use the algorithm $\calA$ for this purpose. Recall that $\calA$ tells us whether $\BCQ_{\q}$ is $1$ or $0$. We now modify the input instance $\q$ as follows -- for each $\mathbf{p} \in P_1$, we project down all the input $R_e$s to values in $\mathbf{p}$. In particular, we update $R_e$ for every $e \in E$ as follows:
\[R'_e = \{\mathbf{t} : \pi_{(S \cup \{v\}) \cap v(e)} (\mathbf{t}) = \pi_{(S \cup \{v\}) \cap v(e)} (\mathbf{p})\}.\]
The new input for $\calA$ is $\q = (G, \{R'_e : e \in E\})$ and $N$. The new output of $\calA$ can be interpreted as follows -- $A$ now tells if there exists a $\J$ such that $\pi_{S \cup \{v\}} (J) = \mathbf{p}$. Note that this computation can be performed in time $O(\Time(\calA))$. If $J_{Q, \mathbf{p}}$ is not empty, we retain this tuple $\mathbf{p}$ in $P_1$. Otherwise, we remove it from $P_1$. We now claim that by the end of $\Delta - 1$ steps, $\calB$ returns a join cover $\calS= P_{\Delta - 1}$. We still need to argue the runtime and correctness of $\calB$, which we do one-by-one.

We argue that our algorithm takes $O(N \cdot |J_{Q, S}| \cdot \Time(\calA))$ time, where $\Time(\calA)$ denotes the runtime of algorithm $\calA$. Towards this end, we start by claiming that after the $i$-th step of our algorithm, we have $|P_{i}| \le |J_{Q, S}|$. Assuming that our claim is true, at each step of computing $P_{i} = P_{i - 1} \cdot \Dom(v)$, we would be performing only $|P_{i - 1}| \cdot N$ calls to $\calA$, as all the vertex-domain sizes are (effectively) upper bounded by $N$ as well. Since we run our algorithm for $\Delta - 1$ steps, our overall runtime is $O(N \cdot  |J_{Q, S}| \cdot \Time(D))$, ignoring a constant factor of $\Delta$. 

To complete the proof, we need to argue two things -- $(1)$ after each step $i: 1 \le i \le \Delta - 1$ of the algorithm, $|P_{i}| \le N \cdot |J_{Q, S}|$ and $(2)$ $\calB$ returns a valid join cover upon completion. We start with $(1)$ first. Recall that in the beginning of our algorithm, we set $P_{0} = J_{Q, S}$, where $|J_{Q, S}| = \min_{S: |S| = n - \Delta + 1} 2^{\PMB(\q, S)}$. Since $P_{0}$ is a set, all pairs of tuples are pairwise distinct. Further, for every $\mathbf{p} \in P_{0}$ of length $|S| = n - \Delta + 1$, there exists at most one extension to a tuple $\mathbf{t} \in J$ of length $n$. This follows from the distance property (of Hamming distance) discussed in Section~\ref{sec:results_ub}. Note that each step, we compute $P_{i} = P_{i - 1} \times \Dom(v)$ for a unique vertex $v \in V \setminus S$. Since $P_{i}$ satisfies the same property as $P_{0}$ as well, we have $|P_{i}| \le N \cdot |J_{Q, S}|$ for every $i \in [\Delta - 1]$. We now argue $(2)$. The output of $\calB$ is $P_{i}$. By construction, we have that for every tuple $\mathbf{t} \in \J$, there exists at least tuple $\mathbf{p} \in P_{i}$ such that $\Dist(\mathbf{t}, \mathbf{p}) \ge \Delta$. Thus, we can set $\calU = P_{\Delta - 1}$ returning a join cover in time $O(N \cdot |J_{Q, S}| \cdot \Time(\calA))$.
\end{proof}
The same proof applies for Lemma~\ref{lemma5.3} as well.

%% file: appendix_garity.tex
\newcommand{\Var}{{\mathrm{Var}}}
\newcommand{\newalpha}{c}

\section{Missing Details in Section~\ref{sec:gen-arity}} \label{sec:garity-proof}

Our main goal in this section is to prove Theorem~\ref{generalarityubtheorem}. We start by arguing that we can always lower bound $\cover(G, N, \Delta)$ by $\Omega_{n}(N^{\LP_{\lb}(G, \s)})$, where $\LP_{\lb}(G, \s)$ is the objective value of the linear program denoted by $\LP_{\lb}$ (details in Appendix~\ref{sec:cover-lb}). We then argue that we can always upper bound $\joinUB(G, N, \s)$ by $N^{\LP_{\ub}(G, \s)}$, where $\LP_{\ub}(G, \s)$ is the objective value of the linear program that computes $\AGM(G, S)$ (details in Appendix~\ref{sec:proj-ub}). Finally, we reason about the gap between $\LP_{\ub}(G, \s)$ and $\LP_{\lb}(G, \s)$ -- we prove an upper bound on this gap in Appendix~\ref{gaub} (which in turn proves~\eqref{eq:gubt1}) and a lower bound on this gap in Appendix~\ref{sec:galb}. This is sufficient to complete the proof.

\subsection{Lower Bound for $\cover(G, N, \Delta)$ using $\LP_{\lb}(G, \s)$} \label{sec:cover-lb}
We claim the following:
\begin{claim} \label{claim:gaclaim1}
Given any hypergraph $G$, $N$ and $\Delta$, we have
\[\cover(G, N, \Delta) \ge \Omega_{n}\left(N^{\LP_{\lb}(G, \s)} \right).\]
\end{claim}
\begin{proof}
To prove this inequality, we only need to construct an instance $\q = (G, \{R_e: e \in E, |R_e| \le N\})$ and a join cover $\J$. We construct a CRT code (Definition~\ref{defn:crt}) from the optimal dual solution to the following linear program:
\begin{align*}
& \max{L} \\
& \text{s.t. } \sum_{v \in e} y_{v} \le 1, \text{for all } e \in E \\
& \sum_{v \in S: S \subseteq V, |S| = s} y_{v} \ge L \\
& y_{v} \ge 0, v \in V. \numberthis \label{eq:crt-1}
\end{align*}
Let $\mathbf{y} = (y_{v})_{v \in V}$ be an optimal solution to the above LP and $L(\mathbf{y})$ be the objective value. For each $v$, we define $q_{v}$ as a unique prime of order $\Theta(N^{y_{v}})$ and $\mathbf{q} = (q_{v})_{v \in V}$.

We now instantiate a $\left(n,  M, n - \Delta + 1\right)_{\mathbf{q}}$-CRT code $\mathcal{C}$ with $M = \Omega_{n}\left(N^{\LP_{\lb}(G, \s)} \right)$, where $\LP_{\lb}(G, \s) = L(\mathbf{y})$. For each edge $e \in E$, the optimal dual solution ensures that $\sum_{v \in e} y_{v} \le 1$. Note that this implies $|\mathcal{C}_e| \le O(N)$ for every $e \in E$. We set $R_{e} = \mathcal{C}_{e}$ for every $e \in E$ and it follows that $|R_{e}| \le N$. Further, we have $\mathcal{C} = \J =\ \Join_{e \in E}  R_e$. Since the size $M$ of $\mathcal{C}$ is determined by the smallest $n - \Delta + 1$ values, we have
\begin{align*}
\cover(G, N, \Delta) \ge |\mathcal{C}| \ge  \Omega_{n}\left(N^{\LP_{\lb}(G, n - \Delta + 1)} \right), 
\end{align*}
where the second equation follows from the fact that 
\begin{align*}
\LP_{\lb}(G, \s) &= L(\mathbf{y}) \\
& \geq \min_{S: \subseteq V: |S| = s} \sum_{v \in S} y_{v} \quad \text{(from the LP)}. 
\end{align*}
This completes the proof.
\end{proof}

\subsection{Upper Bound for $\joinUB(G, N, \s)$ using $\LP^*_{\ub}(G, \s)$} \label{sec:proj-ub}
Our goal here is to prove the following result.
\begin{lemma} \label{lemma:galemma2}
\begin{align*}
\joinUB(G, N, \s) & \le N^{\LP^*_{\ub}(G, \s)} \\
& = 2^{\AGM(G, S)}.
\end{align*}
\end{lemma}
We start with the following LP.
\begin{align*}
& \min{\sum_{e\in E}x_e} \\
& \text{s.t. } \sum_{e \ni v}x_e \ge z_{v}, \text{ for all } v \in V \\
& \sum_{v \in V}z_{v} \ge \s \\ 
& z_{v} \le 1, v \in V \\
& x_e  \ge 0,  e \in E. \numberthis \label{eq:primal-1}
\end{align*}
Let the objective value of this LP be denoted by $\LP_{p}(G,\s)$. We argue that~\eqref{eq:primal-1} is the dual of LP~\eqref{eq:crt-1}.
\begin{lemma} \label{lemma:galemma1}
For any hypergraph $G$, $N$ and $\Delta$, we have
\[\LP_{p}(G, \s) = \LP_{\lb}(G, \s).\]
\end{lemma}
Assuming that the above lemma is true, we now claim that by making the $z_{v}$s integral in LP~\eqref{eq:primal-1}, we have a LP that computes $\AGM(G, S)$. We are now ready to prove Lemma~\ref{lemma:galemma2}.
\begin{proof} [Proof of Lemma~\ref{lemma:galemma2}]
We start by restating the $\AGM(G, S)$ LP (from~\eqref{eq:ubs}):
\begin{align*}
& \min{\sum_{e \in E}x_e \log(N)} \\
& \text{s.t. } \sum_{e \ni v}x_e \ge 1, \text{ for all } v \in S \\
& x_e \ge 0,  e \in E. \numberthis \label{eq:agm}
\end{align*}
We claim that in LP~\eqref{eq:primal-1}, $\sum_{v \in V}z_{v} = \s$ at optimality. If $\sum_{v \in V} z_{v} > \s$, then we can decrease any $x_e$ variable continuously; to make sure that $\sum_{e \ni v} x_{e} \ge z_{v}$ is satisfied, we decrease the related $z_v$ variables. This will decrease the objective value. Thus, for an optimal solution to LP~\eqref{eq:primal-1}, we have $\sum_{v \in V} z_{v} = \s$. 

Further, we claim that at optimality for LP~\eqref{eq:primal-1}, $z_v = \min\{\sum_{e \ni v} x_{e}, 1\}$, for every $v \in V$. Note that if $z_v < \min\{\sum_{e \ni v} x_{e}, 1\}$ for some $v \in V$, we can make $z_v$ equal to $\min\{\sum_{e \ni v} x_{e}, 1\}$. Then, we would have $\sum_{v \in V} z_{v} > \s$. We can use the above procedure to decrease the objective value and as a result, for an optimal solution to LP~\eqref{eq:primal-1}, we have $z_v = \min\{\sum_{e \ni v} x_{e}, 1\}$, for every $v \in V$.

Based on the above claims, we will now argue that a version of LP~\eqref{eq:primal-1} when $z_{v}$s are integral computes $\AGM(G, S)$. To this end, we first write this specific integral version, whose objective value we denote by $\LP^*_{\ub}(G, \s)$.
\begin{align*}
& \min \quad \sum_{e\in E}x_e \\
& \text{s.t. } \sum_{e \ni v}x_e \ge 1, \text{ if } z_{v} = 1 \quad \forall v \in V \\
& \sum_{v \in V}z_{v} = \s \\ 
& z_v \in \{0, 1\}, v \in V \\
& x_e  \ge 0,  e \in E.  \numberthis \label{eq:agmgs}
\end{align*}
Given a feasible solution to LP~\eqref{eq:agm}, we can always convert it to a feasible solution to LP~\eqref{eq:primal-1} and the other way round. Thus, LP~\eqref{eq:primal-1} computes $\AGM(G, S)$ and since $2^{\AGM(G, S)}$ is a valid upper bound on $\joinUB(G, \s)$, the stated result follows. 
\end{proof}
Finally, we prove Lemma~\ref{lemma:galemma1}.
\begin{proof} [Proof of Lemma~\ref{lemma:galemma1}]
Consider the primal version of the LP~\eqref{eq:crt-1}:
\begin{align*}
& \min{\sum_{e\in E}x_e} \\
& \text{s.t } \sum_{e \ni v}x_e - \sum_{S \ni v}w_{S} \ge 0, \text{ for all } v \in V \\
& \sum_{S: S \subseteq V, |S| = k}w_{S} \ge 1 \\ 
& w_{S} \ge 0, S \subseteq V, |S| = \s \\
& x_e  \ge 0,  e \in E. \numberthis \label{eq:oldprimal}
\end{align*}
 We now replace $\sum_{S \ni v}w_{S}$ by $z_{v}: 0 \le z_{v} \le 1$ for all $v \in V$. Notice that 
 \begin{align*}
 \sum_{v \in V} z_{v} & \geq (\s) \cdot \sum_{S: S \subseteq V, |S| = \s} w_{S}\\
 &  \geq \s, 
 \end{align*}
 where the last inequality follows from the constraint $\sum_{S: S \subseteq V, |S| = \s} w_{S} \ge 1$. In particular, we have now reduced LP~\eqref{eq:oldprimal} to LP~\eqref{eq:primal-1}. We now do the reduction the other way round i.e., starting with LP~\eqref{eq:primal-1}:
\begin{align*}
& \min{\sum_{e\in E}x_e}\\
& \text{s.t, } \sum_{e \ni v}x_e \ge z_{v}, \text{ for all } v \in V \\
& \sum_{v \in V}z_{v} \ge \s \\ 
& z_{v} \le 1, v \in V. \\
& x_e  \ge 0,  e \in E.
\end{align*}
We can replace $z_{v} = \sum_{S \ni v} w_{S}$, where $S: S \subseteq V, |S| = \s$ for all $v \in V$. Notice that 
\begin{align*}
\sum_{S: S \subseteq V, |S| = \s} w_{S} \quad \geq \frac{\sum_{v \in V} z_{v}}{\s} \quad \geq 1, 
\end{align*}
where the last inequality follows from the constraint $\sum_{v \in V} z_{v} \ge \s$. In particular, we have reduced LP~\eqref{eq:primal-1} to LP~\eqref{eq:oldprimal} and shown that these LPs can be reduced to each other in both ways. Thus, $\LP_{p}(G, \s) = \LP_{\lb}(G, \s)$, completing the proof.
\end{proof}
We now define a useful notion of {\em cover}ing a vertex.
\begin{definition} [Cover]
A vertex $v \in V$ in the LPs \eqref{eq:primal-1} and \eqref{eq:agmgs} is considered {\em cover}ed iff $z_{v} = 1$.
\end{definition}

\subsection{Proof of Theorem~\ref{thm:ubga1}} \label{gaub}
In this section, we prove an upper bound for gap between $\LP_{\ub}(G, s)$ and $\LP_{\lb}(G, s)$.

We first prove Theorem~\ref{thm:ubga2} in three steps. First, we come up with a dependent randomized rounding (DRR hereon) algorithm $\mathcal{A}$ to round LP~\eqref{eq:primal-1}. Next, we show that this rounding is within constant multiplicative and additive factors and covers at least $s$ vertices with a non-zero probability. Finally, we show that the multiplicative factor converges to $10.37$ and the additive factor is $1$. We then sketch the proof for Theorem~\ref{thm:ubga3} based on the proof of Theorem~\ref{thm:ubga2} showing that the multiplicative factor converges to $3.73$ and the additive factor is $O(\log(s))$.

Before we describe our DRR algorithm $\mathcal{A}$, we use the following observation on $\LP$~\eqref{eq:agmgs}, which we argued earlier.
\begin{claim} \label{claimub1}
For a given optimal solution $\mathbf{x}$ to $\LP$~\eqref{eq:agmgs}, we have $\sum_{v \in V} z_{v} = s$, and $z_v = \min\{\sum_{e \ni v} x_{e}, 1\}$, for every $v \in V$. 
\end{claim}

\iffalse
\begin{proof}
If $\sum_{v \in V} z_{v} > s$, then we can decrease any $x_e$ variable continuously; to make sure that $\sum_{e \ni v} x_{e} \ge z_{v}$ is satisfied, we decrease the related $z_v$ variables. This will decrease the value of the objective. Thus, for an optimum solution, we have $\sum_{v \in V} z_{v} = s$. 
	
Similarly, if $z_v < \min\{\sum_{e \ni v} x_{e}, 1\}$ for some $v$, we can change $z_v$ to $\min\{\sum_{e \ni v} x_{e}, 1\}$. Then, we have $\sum_{v \in V} z_{v} > s$. We can use the above procedure to decrease the objective value. Thus, for an optimum solution, we have $z_v = \min\{\sum_{e \ni v} x_{e}, 1\}$, for every $v \in V$. 
\end{proof}
\fi
We now state a well-known result regarding dependence rounding.
\begin{theorem} [From~\cite{Srinivasan01}]\label{drrthm1}
Let $\tilde{x} \in [0, 1]^E$. Then, there exists a randomized procedure that outputs $\tilde{X}_{e} \in \{0, 1\}^E$, such that
\begin{itemize}
\item {$\mathbb{E}[\tilde{X_e}] = \tilde{x}_{e}$ for every $e \in E$.} 
\item {With probability $1$, we have $\sum_{e \in E}\tilde{X}_{e} \in \left\{\floor*{\sum_{e \in E}\tilde x_e}, \ceil*{\sum_{e \in E}\tilde x_e}\right\}$.}
\item {For every subset $E' \subseteq E$, we have 
\begin{align*}
\Pr[\forall e \in E', \tilde{X}_e = 0] \leq \prod_{e \in E}(1-\tilde x_e).
\end{align*}}
\end{itemize}
\end{theorem}
We now round the LP~\eqref{eq:primal-1} using Algorithm $\mathcal{A}$ (i.e., Algorithm $2$). We assume that $\newalpha  > 1$ is a parameter that will be decide later.
\begin{algorithm}[!htbp]
	\caption{Algorithm $\mathcal{A}$}
	\label{algorithm_a}
	\begin{algorithmic}[1]
		\Require{$\mathbf{x} = (x_e)_{e \in E}$}.
		\Ensure{$X \in [0, 1]^E, Z \in \{0, 1\}^V$}
		\State \textbf{for every} $e \in E$: let $\tilde x_e \gets \min\{c x_e, 1\}$
		\State Compute $(\tilde X_e)_{e \in E}$ by applying Theorem~\ref{drrthm1} to $\tilde x$. 
		\State \textbf{for every} $e \in E$:  Let $X_e \gets \max\{\tilde{X}_e, \min\{\newalpha \cdot x_e, 1\}\}$. 
		\State \textbf{for every} $v \in V$:  
		\State\hspace*{\algorithmicindent}   \textbf{if} $\sum_{e \in V: e \ni v}X_e \geq 1$ \textbf{then} $Z_v \gets 1$
		\State\hspace*{\algorithmicindent} \textbf{else} $Z_v \gets 0$ 
		\Return $X$ and $Z$
	\end{algorithmic}
\end{algorithm}
For $\mathbf{X} = (X_{e})_{e \in E}$ returned by Algorithm $\mathcal{A}$, the following is true:
\begin{lemma} \label{lemma:ublemma1} The $X$ vector returned by Algorithm $\mathcal{A}$ always has
	$\sum_{e \in E} X_e \leq 2c \sum_{e \in E}x_e + 1$.
\end{lemma}
\begin{proof}
	Observe from Algorithm $\mathcal{A}$ that 
	\begin{align*}
		\sum_{e \in E} X_{e} &= \sum_{e \in E} \max\{\tilde{X}_e, \min\{\newalpha \cdot x_{e}, 1\}\}\\
		&\leq \sum_{e \in E}(\tilde X_e + \newalpha \cdot x_e) \\
		&\leq \ceil*{\sum_{e \in E}\tilde x_e} + \newalpha \sum_{e \in E}x_e
		\leq \sum_{e \in E}\tilde x_e + \newalpha \sum_{e \in E}x_e + 1\\
		&\leq c\sum_{e \in E}x_e + \newalpha \sum_{e \in E}x_e + 1\\
		&= 2 \cdot c \cdot \sum_{e \in E}x_e + 1,
	\end{align*}
	where the second inequality is by Theorem~\ref{drrthm1} and the third inequality is by the definition of $\tilde{x}_e$. 
\end{proof}
We still need to argue that $\sum_{v \in v} z_{v}\ge k$ happens with non-zero probability. To this end, we break $V$ into two parts -- $V_1$ is the subset of vertices in $v \in V$ such that $z_{v} \ge \frac{1}{\newalpha}$ and $V_2$ is the subset of vertices $v \in V$ such that $z_{v} \le \frac{1}{\newalpha}$.
\begin{claim} \label{claim4.9}
For every vertex $v \in V_{1}$, $Z_{v} = 1$. 
\end{claim} 
\begin{proof}
Recall that $X_{e} = \max\{\tilde{X}_{e}, \newalpha \cdot x_{e}\}$ for every $e \in E$. Fix a vertex $v \in V_1$. If some edge $e \in E$ incident on $v$ has $x_e \geq \frac{1}{\newalpha}$, then $\sum_{e \in E: v \in e}X_e \geq 1$ and we have $Z_v = 1$; otherwise, we have
\[\sum_{e \ni v} X_e = \sum_{e \ni v} \max\{\tilde{X}_e, \newalpha \cdot x_{e}\} \geq \newalpha \sum_{e \ni v}x_e \ge c z_{v} \geq 1,\]
where the last inequality used the definition of $V_1$. Thus, we have $Z_v = 1$.
\end{proof}

Thus, the algorithm $\mathcal{A}$ \textit{covers} at least $|V_1|$ vertices since for each vertex $v \in V_1$, $Z_{v} = 1$. We now consider vertices in $V_2$. Let $s' = \sum_{v \in V_2}z_{v}$;  then we have $s' + |V_1| \geq s$ since every $v \in V_1$ has $z_v \leq 1$ and $\sum_{v \in V}z_v = s$ (by Claim~\ref{claimub1}).  Further, if some edge $e \in E$ has $|e \cap V_2| \geq s'$, then we could set $X_{e'} = \min{(\newalpha x_{e'}, 1)}$ for every other edge $e' \in E$ and then set $X_e = 1$ for the edge with $|e \cap V_2| \geq s'$. In this way, we covered at least $s$ vertices, with $\sum_{e \in E}X_e \leq \newalpha \sum_{e \in E}x_e + 1$; the theorem is proved i.e., the bound that we obtain from this case is no worse than the general one ($c \le 2c$). Thus, from now on, we assume every edge $e \in E$ has $|e \cap V_2| < s'$. 

We prove the following result for covering vertices in $V_2$.
\begin{lemma} \label{lemma:ublemma2}
	The following is true:
	\begin{enumerate}[label=(\ref{lemma:ublemma2}\alph*), itemsep=0pt, leftmargin=*]
		\item For every $v \in V_2$, we have $\mu_v := \E[Z_v] \geq c_1 z_v$, where $c_1 = \left(1 - \dfrac{1}{\tilde{e}} \right)\newalpha$. \label{property:ubprop3}
		\item Let $|Z| = \sum_{v \in V_2}Z_v$ denote the number of vertices in $V_2$ covered by $\mathcal{A}$. Then, there exists a choice of $c \ge 0$ such that $Pr[|Z| \ge k'] > 0$. \label{property:ubprop4}
	\end{enumerate}
\end{lemma}

\begin{proof}
For any $v \in V_2$, note that $\mu_{v} = Pr[Z_v = 1]$. Since for $v \in V_2$, we have $\sum_{e \in E: v \in e}x_e < \frac1\newalpha$, we have that $Z_v = 1$ if and only if $\tilde X_e = 1$ for some $e \ni v$. Thus, $\mu_v = Pr[\exists e \ni v, \tilde x_e = 1]$.  Notice that every $e \ni v$ has $x_e \leq 1/\newalpha$; thus, we have $\tilde{x}_{e} = c \cdot x_{e}$ for every $e \ni v$. Using negative correlation, we have
	\begin{align*}
		\mu_{v} & = 1 - \Pr[v \text{ is not covered}] \\
		& \ge 1 - \prod_{e \ni v} (1 - \tilde{x}_e) \\
		& \geq 1-\exp(-\sum_{e \ni v}\tilde x_e) \\
		& =  1 - \exp(-c \cdot \sum_{e \ni v} x_{e}).
	\end{align*}
	Then, using the fact that $\sum_{e \ni v} x_{e} = z_{v}$ for every $v \in V_2$ (from Claim~\ref{claimub1}) yields
	\begin{align*}
		\mu_{v} \ge 1 - \exp(-c \cdot z_v).
	\end{align*}
	
Notice that $\frac{1-\exp(-ct)}{t}$ is an decreasing function for $t > 0$, and $z_v < 1/\newalpha$ for every $v \in V_2$, we have that $\frac{1-\exp(-cz_v)}{z_v} \geq \frac{1-\exp(-c/\newalpha)}{1/\newalpha}$. This implies that $\mu_v \geq \left(1-\dfrac{1}{\tilde{e}} \right)\newalpha z_v = c_1z_v$.
	
We finally prove property~\ref{property:ubprop4}. Using Chebyshev's inequality, we have
\begin{align*}
\Pr [|Z| \ge k'] \le \dfrac{Var[|Z|]}{(\mu - k')^2}. \numberthis \label{eq:cheby1}
\end{align*}
To prove this lemma, we first need to upper bound $Var(|Z|) = \mathbb{E}[|Z|^2] - (\mathbb{E}[|Z|])^2$ and we start by upper bounding $\mathbb{E}[|Z|^2]$. By linearity of expectation:
\begin{align*}
\mathbb{E}[|Z|^2] = \sum_{u, v \in V_2} \mathbb{E}[Z_u Z_v]
- \sum_{u \in V_2} \mathbb{E}[Z_u] \sum_{v \in V_2} \mathbb{E}[Z_v].
\end{align*}
Note that $u$ and $v$ can either be covered by an edge $e \in E$ containing both $\{u, v\}$  or containing exactly only one of them. Observe that in the latter case, $Z_{u}$ and $Z_{v}$ are independent. As a result, we can write the above equation as
\[\mathbb{E}[|Z|^2] \le \sum_{u, v \in V_2} \Pr [\exists e \supseteq \{u, v\}, \tilde X_e = 1] + \sum_{u, v \in V_2} \mathbb{E}[Z_u]\E [Z_v].\]
Applying $\Pr[\exists e \supseteq \{u, v\}, \tilde X_e = 1] \le \sum_{e \ni \{u, v\}} c \cdot x_e$ (since each edge is picked in Algorithm $\mathcal{A}$ with probability $\tilde{x}_{e} = \min{(1, c \cdot x_e)} \le c \cdot x_e$), we get
\begin{align*}
\mathbb{E}[|Z|^2] & \le \sum_{u, v \in V_2} \sum_{e \supseteq \{u, v\}} c \cdot x_e + \sum_{u, v \in V_2}   \mathbb{E}[Z_u] \mathbb{E}[Z_v] \\ 
& = \sum_{u, v \in V_2} \sum_{e \supseteq \{u, v\}} c \cdot x_e + \sum_{u, v \in V_2}\mu_{u} \mu_{v} \\
& = \sum_{u, v \in V_2} \sum_{e \supseteq \{u, v\}} c \cdot x_e + \mu^2.
\end{align*}
Notice that  $\sum_{u, v \in V_2} \sum_{e \supseteq \{u, v\}} c \cdot x_e = c \cdot \sum_{e \in E}x_{e} \cdot |e \cap V_2|^2$. We now substitute this in the above equation to get
\begin{align*}
\mathbb{E}[|Z|^2] & \le c \cdot \sum_{e \in E}x_{e} \cdot |e \cap V_2|^2 + \mu^2 \\
& < c \cdot s' \cdot \sum_{e \in E}x_{e} \cdot |e \cap V_2| + \mu^2, \numberthis \label{eq:cheby2}
\end{align*}
since $|e \cap V_2| < s'$. Finally, we observe that
\begin{align*}
\sum_{e \in E} x_{e} \cdot |e \cap V_2| = \sum_{v \in V_2} \sum_{e \ni v} x_{e} = \sum_{v \in V_2}z_v = s',
\end{align*}
where the second equality follows from Claim~\ref{claimub1}. Thus, $\mathbb{E}[|Z|^2] \le c \cdot s'^2 + \mu^2$, implying
\begin{align*}
\Var[|Z|] & = \mathbb{E}[|Z|^2] - (\mathbb{E}[|Z|])^2 \\ 
& < c \cdot s'^2 + \mu^2 - \mu^2  < c \cdot s'^2. 
\end{align*}
Equation~\eqref{eq:cheby1} now becomes
\begin{align*}
\Pr[|Z| \ge s'] & \le \dfrac{Var(|Z|)}{(\mu - s')^2}  \\
& < \dfrac{c \cdot s'^2}{(\mu - s')^2} \\
& \le \dfrac{c \cdot s'^2}{(c_1 \cdot s' - s')^2} \\
& = \dfrac{c}{(c_1 - 1)^2},
\end{align*}
where the final inequality follows from the fact that $\mu \ge c_1 \cdot s'$ for all $c_1 \ge 1$. We can always choose $\newalpha \ge 0$ so that $\frac{c}{(c_1 - 1)^2} < 1$, completing the proof.
\end{proof}
By Claim~\ref{claim4.9} and Lemma~\ref{lemma:ublemma2}, we can conclude that the output of Algorithm~$\mathcal{A}$ is a feasible solution for LP~\eqref{eq:primal-1}.

The final step is to determine the value of $2 \cdot c$.  Our goal here is to minimize $2c$ such that $c_1 > 1$ and $\dfrac{c}{(c_1 - 1)^2} < 1$, where $c_1 = (1 - \tilde{e}^{-1})c$.  It is not hard to see that $c$ can be made arbitrarily close to $\frac{2(1-1/\tilde{e}) + 1 + \sqrt{4 (1-1/\tilde{e}) + 1}}{2(1-1/\tilde{e})^2} <5.184$. Thus, we have $2c \leq 10.37$. Now, Lemma~\ref{lemma:ublemma1} proves Theorem~\ref{thm:ubga2}.

We now sketch the proof of Theorem~\ref{thm:ubga3}. The only difference with the proof of Theorem~\ref{thm:ubga2} is that in the proof, we shall guarantee that $|e \cap V_2| < \eps s'$ for every edge $e$, where $\eps$ is a small constant. Then, in the proof, we need to guarantee that $c_1 > 1$ and $\Var[|Z|] \leq \frac{\epsilon c}{(c_1-1)^2} < 1$.  The latter can be guaranteed by making $\eps$ small enough. Thus, we only need to guarantee that $c_1 = \left(1-\frac{1}{\tilde{e}} \right)c > 1$; thus, we can have $2c$ arbitrarily close to $\frac{2}/{1-\frac{1}{\tilde{e}}} = \frac{2 \cdot \tilde{e}}{\tilde{e}-1} < 3.72$. 

Now we show how to guarantee $|e \cap V_2| < \eps s'$ for every edge $e$. If we see an edge $e$ with $|e \cap V_2| \geq \eps s'$, we then choose the edge $e$ by letting $X_e = 1$, removing all vertices in $e \cap V_2$ from $V_2$ and updating $s'$ to $s' - |e \cap V_2|$.  We repeat this procedure until no such edge $e$ can be found. Notice that in each iteration we scale $s'$ by a factor of at most $(1-\eps)$, in at most $\log_{1+\eps}(s) = O(\log(s))$ iterations (for constant $\eps > 0$), the procedure will terminate. Moreover, after the procedure, we have $\sum_{v \in V_2}z_v \geq s'$. The remaining arguments follow. The number of edges $e$ for which we manually set $X_e$ to 1 is at most $O(\log(s))$, which leads to additive factor of $O(\log(s))$.

We are now ready to prove the upper bound~\eqref{eq:gubt1}. By Theorems~\ref{thm:ubga2} and~\ref{thm:ubga3}, we have shown that for any hypergraph $G$, $N$ and $\Delta$, the following is true:
\begin{align*}
\LP_{\ub}(G, \s) \le 10.37 \cdot \LP_{\lb}(G, \s) + 1 \numberthis \label{eq:finalub1}
\end{align*}
and
\begin{align*}
\LP_{\ub}(G, \s) \le 3.73 \cdot \LP_{\lb}(G, \s) + O(\log(\s))),  \numberthis \label{eq:finalub2}
\end{align*}
We can combine them to write:
\begin{align*}
\LP_{\ub}(G, s) \le \min\left(10.37 \cdot \LP_{\lb}(G, s) + 1,  3.73 \cdot \LP_{\lb}(G, s) + O(\log(s)) \right),
\end{align*}
where $s = \s$. We now raise both sides of the equation to the power of $N \ge 1$ to get
\begin{align*}
N^{\LP_{\ub}(G, s)} \le \min\left( N^{10.37 \cdot \LP_{\lb}(G, s) + 1}, N^{ 3.73 \cdot \LP_{\lb}(G, s) + O(\log(s))}\right).
\end{align*}
Using Claim~\ref{claim:gaclaim1}, we have 
\begin{align*}
N^{10.37 \cdot \LP_{\lb}(G, s)} \le O_{n}\left(\cover(G, N, \Delta)^{10.37} \right)
\end{align*}
and 
\begin{align*}
N^{3.73 \cdot \LP_{\lb}(G, s)} \le O_{n}\left(\cover(G, N, \Delta)^{3.73} \right).
\end{align*}
Finally, using Lemma~\ref{lemma:galemma2}, we have 
\begin{align*}
\joinUB(G, N, \s) & \le N^{\LP_{\ub}(G, \s)} \\ 
& \le \min\left( N \cdot \cover(G, N, \Delta)^{10.37}, N^{O(\log(n - \Delta + 1))} \cover(G, N, \Delta)^{3.73} \right).
\end{align*}

To complete the proof, we need to argue~\eqref{eq:glb1}, which would follow from Theorem~\ref{thm:lbga1}.  

\subsubsection{Proof of Theorem~\ref{thm:lbga1}} \label{sec:galb}
We will need the following set cover instance for our proof:
\begin{lemma} \label{lemma:set-cover-gap}
For a large enough constant $C > 0$ the following is true. Let $\epsilon > 0$ be a small enough constant.  Then for every large enough integer $n$, there is a graph $G = (V, E)$ ($V = [n]$, $E = \{E_1, E_2, \dots, E_n\}$) and an integer $d  = \ceil{\frac{C\ln n}{\eps^2}}$ such that:
\begin{enumerate}[label=(\ref{lemma:set-cover-gap}\alph*), itemsep=0pt, leftmargin=*]
\item For every $i \in [n]$, we have $|E_i| \leq (1+\eps)d$; \label{property:set-cover-s-degree}
\item For every $v \in V$, we have $\|\{{i \in [n]: v \in E_i}\}\| \geq (1-\eps)d$;  \label{property:set-cover-v-degree}
\item For every $\alpha \in  [0, 2]$ and every $I \subseteq V$ of size at most $\alpha  \frac{n}{d}$, we have $\left|\union_{i \in I}E_i \right| \leq (1-\tilde{e}^{-\alpha} + \epsilon) n$. \label{property:set-cover-gap}
\end{enumerate}
\end{lemma}
Assuming such an instance exists, we prove Theorem~\ref{thm:lbga1}. 
\begin{proof}[Proof of Theorem~\ref{thm:lbga1}]
To this end, we first construct the lower bound instance to obtain the required gap. Let $\eps > 0$ be an arbitrary constant. Construct a graph $G = (V, E)$ ($|V| = |E| = t, E = \{E_1, E_2, \dots, E_t\}$) and integer $d$ using Lemma~\ref{lemma:set-cover-gap}. We now create a new graph $G' = (V', E')$, which is the instance that we be would working with for the remaining of this proof, as follows. We initially assign $V' = V$ (assuming that the vertices are indexed from $1$ to $t$) and $E' = E$. We add $td'$ more vertices to $V'$ (i.e. $|V'| = t + d't$), where $d' = \floor*{(1-1/\tilde{e})d}$. For each $E'_i \in E'$, we expand it by adding the first $d'$ vertices ({\em private} vertices of $E'_i$) that have been not picked by any other edge $E'_j \in E'$ in the range $[t + 1, t+d't]$. Let $k = \left(2- \frac{1}{\tilde{e}} \right)t$. We do this argument in two steps -- we first upper bound $\LP_{\lb}(G, \s)$ and then lower bound $\LP_{\ub}(G, \s)$. This is sufficient to lower bound $\frac{\LP_{\ub}(G, \s)}{\LP_{\lb}(G,  \s)}$.
	
To cover $k$ vertices fractionally, for each edge $E'_i \in E'$, we let $x_{E'_i} = \frac{1}{(1-\eps)d}$. Further, the number of sets in which each vertex $j \in [t]$ is contained in at least $(1 - \eps)d$ edges (by Property~\ref{property:set-cover-v-degree}). In particular, this implies $z_{v} = \sum_{E'_i \ni v}x_{E'_i} \ge 1$ for every $v \in [t]$. As a result, by picking all the $t$ hyperedges, $\sum_{v \in [t]}z_{v} \ge t$ and we can set $z_{v} = 1$ for all $v \in [t]$. For all $v \in (t, t+d't)$, we have $z_{v} = \frac{1}{(1-\eps)d}$ since they are incident to only one edge. Further, by picking the $t$ hyper-edges we have also ended up covering $\frac{td'}{(1-\eps)d} \geq (1-1/\tilde{e})t$ fractional vertices in $(t, t+d't]$. Thus, overall, we can cover $t + \left(1- \frac{1}{\tilde{e}} \right)t = k$ vertices and $\LP_{\lb}(G, \s) \le\frac{t}{(1 - \eps)d}$.
	
We now cover $k$ vertices integrally and show that we need at least $\left(1+\frac{1}{\tilde{e}} - 3\eps \right)\frac{t}{d}$ fractional edges. The proof is by contradiction. We start by assuming that there exists a pick with $\frac{\alpha t}{d}$ integral hyperedges and $\frac{\beta t}{d}$ fractional hyperedges such that $\alpha + \beta < 1+1/\tilde{e}-3\eps, \alpha > 0, \beta > 0$ to cover $k$ vertices integrally. Note that using $\frac{\alpha t}{d}$ integral edges we can cover $ \le (1 - \exp(-\alpha) + \eps)t$ vertices from $[t]$ (by Property~\ref{property:set-cover-gap}) and $\alpha \frac{t}{d} d'$ vertices from $(t, t + td']$. The $\beta\frac{ t}{d}$ fractional hyperedges can cover at most $\beta \frac{t}{d} (1 + \eps)d$ vertices from $[t]$ and no {\em private} vertex is covered by any of them. Thus, in total we have covered
\begin{align*}
(1-\exp(-\alpha) + \eps)t + \alpha \frac{t}{d} d' + \beta \frac{t}{d} (1+\eps)d \numberthis \label{eq:lbagm1}
\end{align*}
vertices. Substituting $d' = \floor*{(1 - 1/\tilde{e})d}$ in \eqref{eq:lbagm1}, we get
\begin{align*}
& (1-\exp(-\alpha) + \eps)t + \alpha(t/d) \ceil{(1 - 1/\tilde{e})d} + \beta(t/d) (1+\eps)d \\ 
& \leq  (1-\exp(-\alpha) + \eps)t + \alpha(t/d)(1 - 1/\tilde{e})d + \beta (t/d) (1+\eps)d. \numberthis \label{eq:lbagm2}
\end{align*}
We can take $(1 + \eps)$ common in \eqref{eq:lbagm2} to get
\begin{align*}
& (1+\eps)(1-\exp(-\alpha) + \alpha(1-1/\tilde{e}) + \beta)t + \eps t \\
& = (1+\eps)(1-\exp(-\alpha) + \alpha + \beta - \alpha/\tilde{e})t + \eps t \\
& \leq (1+\eps)(1-\exp(-\alpha)-\alpha/\tilde{e}+1+1/\tilde{e}-3\eps)t + \eps t,
\end{align*}
where the final inequality follows from the fact that $\alpha + \beta \le 1 + 1/\tilde{e} - 3\eps$. Observe that the derivative of $-\exp(\alpha)-\alpha/\tilde{e}$ over $\alpha$ is $\exp(-\alpha)-1/\tilde{e}$ and thus the bound is maximized when $\alpha = 1$. Thus, the number of vertices covered is at most 
\begin{align*}
(1+\eps)(1 - 1/\tilde{e} -1/\tilde{e} + 1 + 1/\tilde{e}-3\eps)t + \eps t &= ((1+\eps)(2-1/\tilde{e}-3\eps) + \eps)t \\
& < (2-1/\tilde{e})t \\
& = k,
\end{align*}
resulting in a contradiction i.e., we are not able to cover $k$ vertices. Finally, we have proved that we need 
\[\LP_{\ub} (G, \s) \ge (1 + 1/\tilde{e} - 3\eps) \frac{t}{d}.\]
Hence, we have
\begin{align*}
\frac{\LP_{\ub}(G, \s)}{LP_{\lb}(G, \s)} & \ge \dfrac{(1 + 1/\tilde{e} - 3\eps)}{(1 - \eps)} \\ 
&\ge 1 + 1/\tilde{e} - 4\eps \\
& = 1 + 1/\tilde{e} - \delta, 
\end{align*}
where the last equality follows from setting $\eps = \delta/4$. This completes the proof.
\end{proof}
To complete this section, we prove Lemma~\ref{lemma:set-cover-gap}.
\begin{proof} [Proof of Lemm~\ref{lemma:set-cover-gap}]
Observe that if $n$ is large enough, $d < n$. We now select the $n$ edges $E_1, E_2, \cdots, E_n$ by picking them randomly and independently i.e., for each $E_i$, we choose it by including each vertex of $[n]$ with probability $\dfrac{d}{n}$. Thus, the expected size of each edge $E_i$ is given by 
\begin{align*}
\mathbb{E}[|E_i|] & = \sum_{i = 1}^{n}\dfrac{d}{n} \\
& = d
\end{align*}
and the expected number of edges in which each vertex $v \in V$ is contained is
\begin{align*} 
\mathbb{E}[|\set{i \in [n]: v \in E_i}|] & = \sum_{i = 1}^{i = n}\dfrac{d}{n} \\
& = d
\end{align*}
(using linearity of expectation). By Chernoff bound, we have for every $i \in [n]$
\begin{align*}
\Pr[|E_i| > (1+\eps)d] & < \exp\left[-\frac{\eps^2d}{3} \right] \\
& \leq \exp \left[-\frac{C\ln n}{3} \right]  \\
& \leq \frac{1}{n^2}
\end{align*}
and
\begin{align*}
 \Pr[|\{{i \in [n]: v \in E_i}\}| < (1-\eps)d] & < \exp\left[-\frac{\eps^2d}{2} \right] \\
 & \leq \exp \left[-(\frac{C\ln n)}{2} \right]  \\
 & \leq \frac{1}{n^2}
\end{align*}
for every $v \in V$. In particular, using the union bound, we get
\begin{align*}
\Pr[|E_i| > (1 + \eps)d \text{ for at least one } i \in [n]] & \le \sum_{i = 1}^{n} \frac{1}{n^2} \\
& = \frac{1}{n}
\end{align*}
and
\begin{align*}
\Pr[|\{{i \in [n]: v \in E_i}\}| < (1-\eps)d \text{ for at least one } v \in V] & \le \sum_{i = 1}^{i = n} \frac{1}{n^2} \\
& = \dfrac{1}{n}.
\end{align*}  
Hence, the Properties~\ref{property:set-cover-s-degree} and~\ref{property:set-cover-v-degree} happen simultaneously with probability
\begin{align*}
& \Pr[(|E_i| \le (1 + \eps) \cdot d \text{ for every } i \in [n]) \text{ and } (|\{{i \in [n]: v \in E_i}\}| \geq (1-\eps)d  \text{ for every } v \in V)] \\
& \ge 1 - \dfrac{2}{n}.
\end{align*}
	
To prove Property~\ref{property:set-cover-gap}, it suffices to prove it for $\alpha \in [0, 2]$ such that $\frac{ \alpha n}{d}$ is an integer. Note that $|I|= \frac{\alpha \cdot n}{d}$ cannot be fractional. We fix such an $\alpha$ and a set $I \subseteq [n]$ of size $\frac{\alpha \cdot n}{d}$. Each vertex $v \in V$ is contained in $\union_{i \in I}E_i$ with probability
\begin{align*}
\mathcal{P}_{j} = 1 - \prod_{i \in I} Pr[v \notin E_i] = 1 - \left(1 - \frac{d}{n} \right)^{\frac{\alpha \cdot n}{d}}.
\end{align*}
Further, if $n$ is big enough, then $\frac{d}{n}$ is smaller than a constant $\gamma$, implying
\begin{align*}
(1 - d/n)^{n/d}  \geq \exp[-(1 + \eps/(2 \cdot \alpha))],
\end{align*}
since $\eps/(2 \alpha) > 0$ and $\lim_{n \rightarrow \infty} (1 - d/n)^{n/d} = \exp(-1)$. Thus, we have 
\begin{align*}
(1-d/n)^{\alpha \cdot (n/d)} & \geq \exp(-\alpha \cdot (1+ \eps/(2\alpha))) \\
& = \exp(- (\alpha + \eps/2)).
\end{align*}
It follows that $\mathcal{P}_{j} \le 1 - \exp(- (\alpha + \eps/2))$. 
It is important to note that every event of checking if a vertex $v$ is contained in $\union_{i \in I}E_i$ is independent. By Chernoff bound again, the probability that $\union_{i \in I}E_i$ contains more than $\left (1+\dfrac{\eps}{2} \right )  (1-\exp(-(\alpha + \eps/2)))n$ vertices i.e., 
\begin{align*}
\Pr \left[|\union_{i \in I}E_i| > \left (1+\dfrac{\eps}{2} \right ) (1-\exp(-(\alpha + \eps/2)))n \right] & < \exp \left(-(\eps/2)^2(1-\exp(-(\alpha + \eps/2)))\frac{n}{3} \right) \\ 
& = \exp \left(-\dfrac{\eps^2(1-\exp(-(\alpha+\eps/2)))n}{12} \right).
\end{align*}
Taking the union bound over all sets $I$ of size $\frac{\alpha n}{d}$, the probability that the event
\begin{align*}
\left|\cup_{i \in I}E_i \right| > \left(1+\frac{\eps}{2} \right) (1-\exp(-(\alpha + \eps/2)))n
\end{align*}
happens for some set $I \subseteq [n]$ of size $\alpha n/d$ is upper bounded by 
\begin{align*}
n^{\alpha n /d}\exp \left[-\frac{\eps^2(1-\exp(-(\alpha+\eps/2)))n}{12} \right] = \exp \left( \left(\frac{\alpha\ln n}{d} -\frac{\eps^2(1-\exp(-(\alpha + \eps/2)))}{12} \right)n \right).
\end{align*}
Substituting $d = \ceil{C \ln{n}/\eps^2}$, we get
\begin{align*}
\exp\left(\frac{\alpha\ln n}{\ceil{C \ln{n}/\eps^2}}n -\frac{\eps^2(1-\exp(-(\alpha + \eps/2)))}{12}n \right) 
&\leq \exp \left(\frac{\alpha}{C}\eps^2n - \frac{1-\exp(-(\alpha+\eps/2))}{12}\eps^2n \right) \\ 
&= \exp \left( \left(\frac{\alpha}{C} - \frac{1-\exp(-(\alpha+\eps/2))}{12} \right)\eps^2n \right). \numberthis \label{eq:lbeq1}
\end{align*}
We split $\alpha \in [0, 2]$ into two cases - $\alpha \in [0, 1 - \eps/2]$ and $\alpha \in (1-\eps/2, 2]$. We start with the first case, where we have
\begin{align*}
\exp \left( \left(\frac{\alpha}{C} - \frac{1-\exp(-(\alpha+\eps/2))}{12} \right)\eps^2n \right) 
& \leq \exp\left( \left(\frac{\alpha}{C} - \frac{(1-1/\tilde{e})(\alpha + \eps/2)}{12} \right)\eps^2n \right) \\
& = \exp \left(\frac{\alpha - (\alpha + \eps/2)}{12} (1 - 1/\tilde{e})\eps^2n \right) \\
& = \exp \left(-\frac{(1 - 1/\tilde{e}) \eps^3 n}{24} \right) \numberthis \label{eq:alphacase1} \\
& \leq \dfrac{1}{n^2}.
\end{align*}
Here, the first inequality follows from the fact that $\alpha + \eps/2 \le 1$, which in turn implies
\begin{align*}
1-\exp(-(\alpha+\eps/2)) \ge (1 - 1/\tilde{e}) (\alpha + \eps/2)
\end{align*}
and the first equation follows assuming $C \geq \frac{12}{1-1/\tilde{e}}$ and the others follows directly. The final inequality holds when $n$ is big enough.

Next, we consider the second case $\alpha \in (1-\eps/2, 2]$. We have
\begin{align*}
\exp \left( \left(\frac{\alpha}{C} - \frac{1-\exp(-(\alpha+\eps/2))}{12} \right)\eps^2n \right) 
& \leq \exp\left(\left(\frac{\alpha}{C} - \frac{1-\exp(-(\alpha+\eps/2))}{12}\right)\eps^2n \right)  \\
& \leq \exp \left( \left(\frac{2}{C} - \frac{1-1/\tilde{e}}{12} \right)\eps^2n \right) \\
& \leq \exp \left(-\frac{\tilde{e} - 1}{24\tilde{e}}\eps^2n \right) \\
& \leq \frac{1}{n^2}.
\end{align*}
Here, the first and second inequalities follow from the facts that $\alpha \le 2$ and $\alpha + \frac{\eps}{2} > 1$, which in turn implies $1 - \exp^{-(\alpha + \eps/2)} \ge (1 - 1/\tilde{e})$. The fourth final and final inequalities follow when $C$ and $n$ are big enough,
	
Now, we take the union bound over all $\alpha \in [0, 2]$ such that $\alpha d/n$ is an integer. Note that there can only be at most $n$ such $\alpha$'s. Hence, with probability at least $1-\frac1n$, for every $\alpha \in [0, 2]$ and every set $I \subseteq [n]$ of size at most $\alpha n/d$, $\union_{i \in I}E_i$ contains at most 
\begin{align*}
(1+\eps/2)(1-\exp(-(\alpha+\eps/2)))n  & \leq (1-\exp(-(\alpha+\eps/2)) + \eps/2)n \\
& = (1-\exp(-\alpha)+\eps/2 + (1-\exp(-\eps/2))\exp(-\alpha))n  \\
& \leq (1 - \exp(-\alpha) + \eps/2 + (\eps/2) \exp(-\alpha))n  \\
& \leq (1-\exp(-\alpha)+\eps/2+\eps/2)n \\
& = (1-\exp(-\alpha)+\eps)n
\end{align*}
elements. Here, the second inequality follows since $1 - \exp(-\eps/2) \leq \frac{\eps}{2}$ for any $\eps \ge 0$. The others follow pretty much directly. Overall, with probability at least $1-\frac{3}{n}$, Properties~\ref{property:set-cover-s-degree} to~\ref{property:set-cover-gap} hold.
\end{proof}